\documentclass[amsmath,amssymb,aps,prd,superscriptaddress,10pt]{revtex4-2}
\usepackage{mathtools,amsthm}
\usepackage[T1]{fontenc}
\usepackage{microtype}
\usepackage{xcolor}
\usepackage{graphicx}
\usepackage{booktabs}
\usepackage{placeins}
\usepackage{needspace}
\usepackage{tikz}
\usetikzlibrary{arrows.meta,calc,positioning}
\usepackage{enumitem}
\usepackage{mathrsfs}
\makeatletter
\def\label#1{\@bsphack
  \begingroup
  \UseHookWithArguments{label}{1}{#1}%
  \protected@write\@auxout{}%
    {\string\newlabel{#1}{{\@currentlabel}{\thepage}%
      {\@currentlabelname}{\@currentHref}{\@kernel@reserved@label@data}}}%
  \endgroup
  \@esphack}
\makeatother
\usepackage[colorlinks=true,allcolors=blue!55!black]{hyperref}
\hypersetup{
  pdftitle={Schwarzschild spectral ladders on the negative imaginary axis: Endpoint nonselection, branch-cut phase, and Jost classification},
  pdfauthor={Davide Batic, Denys Dutykh, and Mark Essa Sukaiti},
  pdfsubject={Axial Schwarzschild perturbations, spectral approximation, and quasinormal-mode classification},
  pdfkeywords={Schwarzschild black hole, quasinormal modes, branch cut, Jost function, spectral pollution, Chebyshev collocation}
}
\graphicspath{{figures/}}
\newtheorem{theorem}{Theorem}[section]
\newtheorem{lemma}[theorem]{Lemma}
\newtheorem{proposition}[theorem]{Proposition}
\newtheorem{corollary}[theorem]{Corollary}
\theoremstyle{remark}

\begin{document}

\title{Schwarzschild spectral ladders on the negative imaginary axis:
Endpoint nonselection, branch-cut phase, and Jost classification}
\author{Davide Batic}
\email{davide.batic@ku.ac.ae}
\affiliation{Mathematics Department, Khalifa University of Science and Technology, PO Box 127788, Abu Dhabi, United Arab Emirates}

\author{Denys Dutykh}
\email{denys.dutykh@ku.ac.ae}
\affiliation{Mathematics Department, Khalifa University of Science and Technology, PO Box 127788, Abu Dhabi, United Arab Emirates}

\author{Mark Essa Sukaiti}
    \email{100064482@ku.ac.ae}
    \affiliation{Mathematics Department, Khalifa University of Science and Technology, PO Box 127788, Abu Dhabi, United Arab Emirates}

\date{August 17, 2026}

\begin{abstract}
Compactified spectral discretisations of black-hole perturbations can produce stable negative-imaginary-axis (NIA) eigenvalue ladders, but finite-matrix convergence does not establish quasinormal-pole character. We make this distinction precise for axial Schwarzschild perturbations. After the usual ingoing--outgoing factorisation, the unwanted horizon sector at $\Omega=-i\alpha$ behaves as $t^{4\alpha}$, whereas the unwanted infinity sector is $C^\infty$ and flat. We prove that endpoint $C^k$ regularity is therefore nonselective whenever $4\alpha>k$. In particular, the $C^2$ continuum problem used by the pencils cannot possess a discrete NIA spectrum throughout the reported range $\alpha\simeq15$--32. We also prove exact Chebyshev-grid reach formulas. For $x\sim C(1-y)^{-p}$, the largest finite radius on a roots grid grows as $n^{2p}$, yielding the observed difference between the linear map $\mathrm{C1}\colon x=2/(1-y)$ and the quadratic map $\mathrm{C2}\colon x=4/(1-y)^2$. Arbitrary-precision pencils nevertheless reveal a reproducible 68-point $\mathrm{C1}$ ladder, while the $\mathrm{C2}$ pencil reorganises the NIA spectrum. At every $\mathrm{C1}$ frequency, evaluations of both physical lateral Jost determinants are stable and remain separated from zero relative to the numerical variations in the adopted normalisation. A $0.05$-spaced uniform scan resolves no on-cut zero, and resolution-refined argument-principle computations return zero numerical winding in the specified right and left physical-continuation strips. The pointwise evidence strongly rejects the 68 candidates. The mesh and strip evidence is noncertified and does not constitute an interval theorem. Same-damping controls recover the known $n_{\rm QNM}=60$, $100$, and $130$ Schwarzschild quasinormal modes (QNM)s across the ladder interval, with $|\widehat{\mathcal D}|$ between $10^{-48}$ and $10^{-56}$ at the refined frequencies. The ladder's organisation has a known asymptotic scale, i.e. its quarter spacing is the surface-gravity scale, and all 68 roots closely follow the parameter-free Casals--Ottewill branch-cut-strength phase while pairing with the damping projections of QNMs $n_{\rm QNM}=62,\ldots,129$. At the first ladder member, a finite-frequency calculation resolves a numerical branch-strength zero at $\alpha_q=15.07832396512359$, between the asymptotic phase prediction and the $\mathrm{C1}$ root. This anchor supports the cut-phase mechanism without establishing sequence-wide phase locking. Thus, highly accurate finite-pencil eigenvalues can fail the invariant Jost/Evans pole criterion. The principal open problem is whether representation-dependent nodes, equipped with Keldysh weights, converge collectively to the Schwarzschild cut response and Price tail.
\end{abstract}
\maketitle
\vspace{-9pt}

\section{Introduction}
\label{Intro}

A black-hole QNM is a resonance, not an ordinary normal
mode. For the $e^{-i\omega t}$ convention, it is selected by a solution
that is ingoing at the future event horizon and outgoing at infinity. In the frequency domain its invariant signature is a pole of the analytically continued Green function, represented away from exceptional normalisation or numerator-cancellation points by a zero of a horizon--infinity Jost determinant~\cite{Leaver1985PRSLA, Leaver1986PRD, Kokkotas1999LR, Berti2009CQG, Konoplya2011RMP}. This definition is global and analytic. It is not replaced by a small residual, smooth endpoint values, or the convergence
of an eigenvalue of a truncated matrix pencil. The distinction is particularly delicate on the NIA. The retarded Green function has a branch point at the origin and a cut that may be placed on the NIA. The infinity-outgoing Jost solution consequently has two lateral boundary values, reached from the right and left half-planes without crossing the cut~\cite{Leaver1986PRD, CasalsOttewill2012PRD, CasalsOttewill2013PRD}. Moreover, after the formal horizon-ingoing and infinity-outgoing factors have been removed, the opposite local sectors can become bounded, differentiable, or even flat at compactified endpoints. A numerical boundary problem based on
finite endpoint regularity can then cease to encode the radiation condition that defined the resonance. This mechanism is not merely hypothetical. Hyperboloidal and Chebyshev discretisations of asymptotically flat black-hole problems represent the
continuous cut by finite collections of nonconvergent matrix eigenvalues~\cite{Jaramillo2021PRX, AnsorgMacedo2016PRD, BessonJaramillo2025GRG}. Fortuna and Vega showed that an axial algebraically-special spectral eigenpair can be computed to extreme precision while its radial solution is an inseparable mixture of the two scattering sectors rather than a QNM~\cite{FortunaVega2023EPJC}. Conversely, Besson and Jaramillo demonstrated that such non-QNM eigenvalues need not be useless because, when equipped with Keldysh weights, their aggregate can approximate the cut integral and recover
Price tails~\cite{BessonJaramillo2025GRG}. The correct question is therefore not only whether an individual matrix eigenvalue is a pole, but also whether a weighted family approximates a physical continuum observable. Recent Green-function work sharpens this distinction from the continuum side. Su et al. decompose the Schwarzschild Green function into a branch-cut direct response, a QNM-pole contribution, and a late-time tail, and validate the reconstruction against time-domain Regge--Wheeler evolution~\cite{Su2026PRD}. Rosato, De Amicis, and Pani analyse the low-frequency cut together with surface-gravity-governed redshift terms and their causal dependence on the source location~\cite{RosatoDeAmicisPani2026PRD}. Arnaudo, Carballo, and Withers give a complementary positive construction in which the Schwarzschild branch-cut contribution is represented as a convergent sum of de Sitter QNMs in the $\Lambda\to 0^+$ limit~\cite{ArnaudoCarballoWithers2026PRD}. These continuum results do not address the finite-$C^k$ nonselection theorem, the C1--C2 polynomial pullback obstruction, or the 68-point lateral-Jost classification established below. They do, however, reinforce that a positive discrete representation of the cut must be formulated for a specified Green-function component or time-domain observable, with its residues or weights and mode of convergence made explicit, rather than inferred from unweighted node locations alone. A distinct recent route uses complex scaling to convert the outgoing-wave condition into a non-Hermitian eigenproblem for Schwarzschild and Reissner--Nordstr\"om perturbations~\cite{OgawaHiroseMorikawa2026arXiv}. Its potential role as an independent cross-check is discussed in Sec.~\ref{sec:conclusions_outlook}.

The present work uses axial $s=\ell=2$ Schwarzschild perturbations as an analytically controlled laboratory for that distinction. It was motivated by highly regular overdamped sequences found in compactified spectral calculations, both in Schwarzschild and in deformed black-hole geometries~\cite{Batic2018PRD, Macedo2019Comment, Batic2019Reply, Batic2024PRD, Batic2025PRSA, Batic2026EPJC, Batic2026PRD, Batic2026GaussBonnet, Batic2026KazakovSolodukhin}. Those calculations pose a general classification problem, i.e. is a stable NIA ladder a family of physical poles, a cross-cut resonance family, a finite representation of continuous spectrum, or spectral pollution tied to the trial space? Schwarzschild is the natural place to answer the question because its Jost functions, branch-cut strength, high-overtone QNMs, and algebraically-special frequency are independently
known. We address five logically distinct questions.
\begin{enumerate}[label=(\roman*),leftmargin=2.2em]
\item 
What continuum endpoint condition is actually imposed after the
ingoing--outgoing factorisation and compactification?
\item 
Which features of the finite spectrum survive changes of grid,
resolution, precision, and radial map?
\item 
Do the resulting candidates vanish under either physical lateral Jost
trace?
\item 
What do sampled scans and argument-principle calculations establish
about nearby frequencies, and what remains noncertified?
\item 
Is the observed quarter-spaced phase related to the known Schwarzschild
branch-cut strength and high-overtone resonance comb?
\end{enumerate}
Our first result is analytic and more decisive than a numerical
representation comparison. Let $t$ be a local horizon coordinate. On
$\Omega=-i\alpha$, the factorised unwanted horizon sector behaves as
$t^{4\alpha}$, up to the standard resonant logarithmic qualification, while the unwanted infinity sector becomes exponentially flat. We prove a finite-regularity no-go theorem, i.e. for any fixed integer $k$, endpoint $C^k$ regularity is nonselective once $4\alpha>k$. Thus, the $C^2$ endpoint class used in the present pencils admits both independent local sectors at both endpoints for every reported high-damping candidate. The associated continuum regularity problem cannot have a discrete spectrum there. Discreteness is introduced by the finite polynomial trial spaces and their collocation realisations, not by the stated continuum boundary condition. The next results isolate the geometry and approximation power of those trial spaces. The linear compactification C1 is $x_1=2/(1-y_1)$, whereas the quadratic compactification C2 is $x_2=4/(1-y_2)^2$. Under their nonlinear
coordinate change, the exact common core of the two degree-$d$ polynomial spaces has only $\lfloor d/2\rfloor+1$ dimensions. Notice that representing every C1 polynomial with $n$ coefficients requires at least $2n-1$ C2 coefficients. Equal-degree truncation therefore does not commute with coordinate pullback. We also derive exact Chebyshev coefficients and uniform derivative-error rates for the endpoint Frobenius powers, which explain how an unwanted sector can
have an extremely small coefficient tail without satisfying a physically selective boundary condition. The compactification geometry supplies a complementary exact result. For the $n$-point Chebyshev roots grid, we find $1-y_{\max}=2\sin^2[\pi/(4n)]$. Hence, the largest represented finite radius is exactly $\csc^2[\pi/(4n)]$ for the linear compactification and $\csc^4[\pi/(4n)]$ for the quadratic one. More generally, a map $x\sim C(1-y)^{-p}$ has effective reach $O(n^{2p})$. This theorem explains the geometric origin of the $n^2$ versus $n^4$ radial scales and provides a rigorous starting point but it is not by itself a proof for the different empirical spectral-edge laws. We then examine the finite pencils on their own terms. The first compactification (C1) yields a 68-point, approximately quarter-spaced
high-damping ladder that is stable under the roots-to-Lobatto grid change. The second compactification (C2) recovers ordinary off-axis QNMs but reorganises the NIA population and exhibits different edge scalings. Arbitrary-precision assembly, raw-versus-equilibrated pencils, exact normwise backward errors, an exact singular-value characterisation of the finite-pencil pseudospectrum, first-order condition numbers, coefficient tails, endpoint checks, mode matching, and eigenvector overlaps show that the C1 roots are accurate
eigenvalues of the stated finite problems. This conclusion is deliberately finite-dimensional.

The physical classification is supplied independently by Jost solutions of the original Regge--Wheeler equation. A Jaff\'e series represents the horizon-ingoing solution and a Leaver--Tricomi-$U$ series represents the infinity-outgoing solution. At each of the 68 C1 frequencies, both physical lateral-determinant evaluations are stable and remain separated from zero relative to the reported numerical variations in the adopted normalisation. A uniform 354-point scan with step $0.05$ resolves no candidate on-cut zero. Resolution-refined
argument-principle calculations return zero winding in specified strips on both physical continuations. We state these latter results at their actual evidentiary level: the scan is discrete, the subsequent ball arithmetic is reduced to high-precision midpoints, and zero winding counts zeros minus poles until pole-freeness of the enclosed analytic domain is established. The candidate-point exclusion is strong, and the strip-wide result is high-precision numerical evidence rather than an interval-certified theorem.

Three same-regime controls substantially strengthen the Jost calculation. Using independently tabulated $s=\ell=2$ Schwarzschild QNMs as seeds, the production recurrence recovers the $n_{\rm QNM}=60$, $100$, and $130$ poles at
$-\operatorname{Im}\Omega\simeq14.60$, $24.60$, and $32.11$, respectively. The refined values are stable across matching radii and truncation orders and give $|\widehat{\mathcal D}|$ between $10^{-48}$ and $10^{-56}$. Thus, the method resolves genuine poles at three damping values spanning the interval in which it rejects the C1 candidates. The most important interpretive result is that the ladder's organisation is not new or mysterious. Casals and Ottewill derived that the large-frequency Schwarzschild branch-cut strength has quarter-spaced zeros in the present $\alpha=M\nu$ units and that spin-2 QNM damping projections lie near those zeros~\cite{CasalsOttewill2012PRD, CasalsOttewill2013PRD}. All 68 C1
frequencies follow their parameter-free asymptotic phase and slow drift. Moreover, the phase-form residual has RMS $1.84\times10^{-4}$. The same candidates pair consecutively with the damping projections of QNMs $n_{\rm QNM}=62,\ldots,129$, whose real parts lie just beyond the outer sampled strip. Because two asymptotically equivalent truncations already differ beyond the very small phase-form residual, this comparison is compelling motivation, not an accuracy claim for the finite-frequency cut strength. The direct calculation in Sec.~\ref{sec:cut_phase} resolves one numerical finite-frequency zero and moves the phase prediction toward the first C1 point. The remaining sequence has not been established.

The paper therefore establishes a hierarchy rather than a binary verdict:
\begin{equation}
\label{eq:intro_evidence_hierarchy}
\text{accurate finite-pencil eigenpair}
\;\not\!\Longrightarrow\;
\text{continuum eigenvalue}
\;\not\!\Longrightarrow\;
\text{Jost pole}.
\end{equation}
At the same time, failure of the last implication does not preclude useful collective convergence. The strongest open direction is to compute left/right Keldysh residues and test whether the map-dependent C1 and C2 node sets define discrete measures converging to the same branch-cut Green response or Price tail. That programme would turn a negative pole classification into a positive theory of continuum approximation. Figure~\ref{fig:logic_diagram} summarises the logical structure. The upper route continues the continuum resolvent and identifies its poles and cut. The lower route factorises, compactifies, and projects before taking a matrix spectrum. The square need not commute on the NIA. The independent Jost trace is the map that classifies the finite candidates physically.
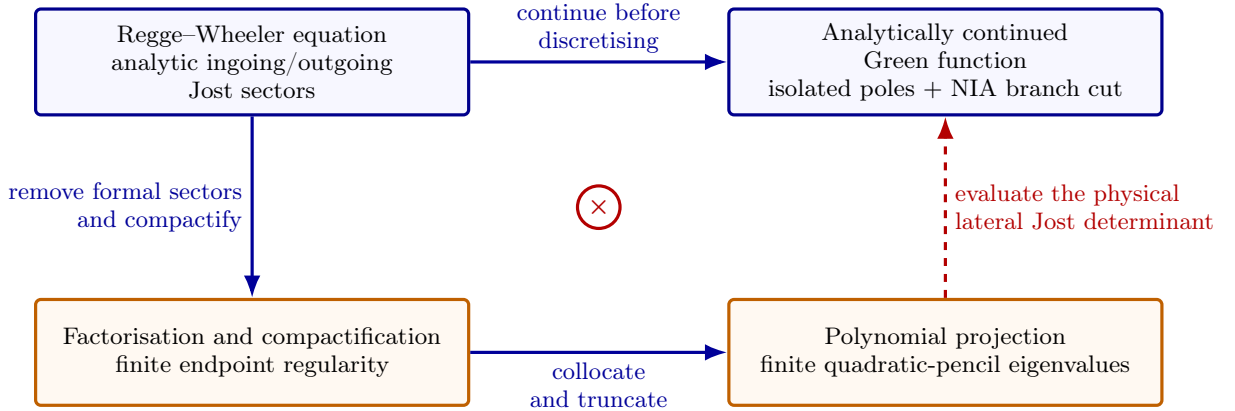
\begin{figure}[htbp]
\centering
\begin{tikzpicture}[
  >=Latex,
  node distance=24mm and 34mm,
  box/.style={draw=blue!55!black,rounded corners=2pt,very thick,
    fill=blue!3,align=center,minimum height=14mm,text width=0.30\textwidth,
    inner sep=5pt},
  finite/.style={box,draw=orange!75!black,fill=orange!5},
  route/.style={->,very thick,blue!60!black},
  test/.style={->,very thick,dashed,red!70!black}
]
\node[box] (ode) {Regge--Wheeler equation\\analytic ingoing/outgoing\\Jost sectors};
\node[box,right=of ode] (resolvent) {Analytically continued\\Green function\\isolated poles $+$ NIA branch cut};
\node[finite,below=of ode] (regularity) {Factorisation and compactification\\finite endpoint regularity};
\node[finite,right=of regularity] (pencil) {Polynomial projection\\finite quadratic-pencil eigenvalues};
\draw[route] (ode) -- node[above,align=center]{continue before\\discretising} (resolvent);
\draw[route] (ode) -- node[left,align=right]{remove formal sectors\\and compactify} (regularity);
\draw[route] (regularity) -- node[below,align=center]{collocate\\and truncate} (pencil);
\draw[test] (pencil) -- node[right,align=left]{evaluate the physical\\lateral Jost determinant} (resolvent);
\node[draw=red!70!black,circle,very thick,inner sep=1.8pt,
  fill=white,text=red!70!black] at ($(ode)!0.5!(pencil)$) {$\boldsymbol{\times}$};
\end{tikzpicture}
\caption{The classification problem as a noncommuting diagram. Continuing
the physical resolvent preserves the distinction between isolated poles and
the branch cut. Factorisation followed by finite endpoint regularity and
polynomial projection can instead produce accurate matrix eigenvalues for
which the unwanted scattering sectors remain admissible. The dashed Jost
trace is therefore an independent classification map, not a posteriori
decoration.}
\label{fig:logic_diagram}
\end{figure}
The remainder of the article follows the claim hierarchy.
Section~\ref{sec:schwarzschild_problem} fixes the Schwarzschild and Jost conventions and proves endpoint nonselection for both compactifications. Section~\ref{sec:numerical_method} constructs the finite quadratic pencils, proves the pullback and Chebyshev-grid results, and defines their numerical diagnostics. Section~\ref{sec:jost_numerics} specifies the two lateral Jost
continuations and separates the exact contour theorem from the numerical winding protocol. Section~\ref{sec:finite_pencil_results} establishes the finite spectra and their representation dependence.
Section~\ref{sec:cut_phase} compares the C1 ladder with the branch-cut phase and high-overtone QNM comb, while Section~\ref{sec:jost_results} assesses its physical-pole character. Section~\ref{sec:keldysh_cut_limit} formulates the
representation-independent weighted-continuum criterion. Finally,
Section~\ref{sec:conclusions_outlook} separates proved results,
high-precision evidence, and open problems.

\section{Schwarzschild radial problem, physical traces, and endpoint nonselection}
\label{sec:schwarzschild_problem}

We use geometrised units $G=c=1$ and signature $(-,+,+,+)$. In standard
Schwarzschild coordinates, the exterior metric is~\cite{Weinberg1972}
\begin{equation}\label{LE}
ds^2=-F(r)dt^2+\frac{dr^2}{F(r)}+r^2\left(d\vartheta^2+\sin^2\vartheta\,d\varphi^2\right),
\qquad
F(r)=1-\frac{2M}{r},
\end{equation}
with mass $M>0$ and event horizon $r_h=2M$. The radial equations for
massless scalar, electromagnetic, and axial gravitational perturbations have the unified Regge--Wheeler form~\cite{Hildreth1963, Matzner1968JMP, Regge1957PR}
\begin{align}
\label{ODE01}
F(r)\frac{d}{dr}\left(F(r)\frac{d\psi_{\omega\ell\epsilon}}{dr}\right)
+\left[\omega^2-U_\epsilon(r)\right]\psi_{\omega\ell\epsilon}(r)
&=0,\\
U_\epsilon(r)
&=F(r)\left[
\frac{\epsilon}{r}\frac{dF}{dr}+\frac{\ell(\ell+1)}{r^2}
\right],
\qquad
\epsilon=1-s^2.
\end{align}
Here, $\epsilon=1$ describes a massless scalar field, with $\ell\geqslant 0$, $\epsilon=0$ describes electromagnetic perturbations, with $\ell\geqslant 1$, and $\epsilon=-3$ describes axial, or odd-parity, gravitational perturbations, with $\ell\geqslant2$. We use the time convention $e^{-i\omega t}$ and
spherical-harmonic angular dependence. With $x=r/(2M)$, and $\Omega=M\omega$, Eq.~\eqref{ODE01} becomes
\begin{equation}\label{ourODE}
F(x)\frac{d}{dx}\left(F(x)\frac{d\psi_{\Omega\ell\epsilon}}{dx}\right)
+\left[4\Omega^2-V_\epsilon(x)\right]
\psi_{\Omega\ell\epsilon}(x)=0,
\qquad
F(x)=1-\frac{1}{x},
\end{equation}
with effective potential
\begin{equation}\label{Veff}
V_\epsilon(x)=F(x)\left[\frac{\epsilon}{x}\frac{dF}{dx}+\frac{\ell(\ell+1)}{x^2}\right].
\end{equation}
The aim of the present work is not to assume a priori that every eigenvalue returned by a compactified spectral problem is a QNM, but to
determine the analytic nature of the structures that appear on the negative imaginary axis. We write $\Omega=\Omega_R+i\Omega_I$, with $\Omega_I<0$, and parametrise the negative imaginary axis by $\Omega=-i\alpha$, with $\alpha>0$. A QNM is defined as an isolated frequency for which the radial solution is horizon ingoing and infinity outgoing, where these asymptotic sectors are understood as analytic continuations of the corresponding Jost solutions. This qualification is particularly important on the negative imaginary axis, where boundedness or smoothness of the compactified radial remainder need not eliminate the
opposite asymptotic solution. We therefore use two inequivalent compactifications of $x\in[1,+\infty)$ onto $[-1,1]$. The first is discretised using both Chebyshev roots and an endpoint-inclusive Chebyshev--Lobatto grid, while the second is discretised using Chebyshev roots. These three finite-dimensional realisations separate the effects of the collocation grid, endpoint-row treatment, radial compactification, resolution, and arithmetic precision. Resolution drift, scaled backward errors, eigenvalue conditioning, and coefficient decay are used to assess the corresponding finite-dimensional eigenpairs. Their interpretation as physical QNMs is then tested independently by constructing the horizon and infinity Jost solutions and locating isolated zeros of their analytically continued Wronskian, with the two lateral continuations around the negative-imaginary-axis branch cut treated separately. The numerical study reported below is restricted to the axial gravitational sector $s=\ell=2$.

\subsection{Asymptotic sectors and Jost solutions}

It is useful to introduce the dimensionless Schwarzschild tortoise
coordinate
\begin{equation}\label{tortoise}
\xi=\frac{r_*}{2M},\qquad
\frac{d\xi}{dx}=\frac{1}{F(x)},\qquad
\xi=x+\ln(x-1),
\end{equation}
where an irrelevant additive constant has been omitted. Since $d/d\xi=F(x)d/dx$, the radial equation \eqref{ourODE} assumes the Schr\"odinger-like form
\begin{equation}\label{SchrodingerForm}
\frac{d^2\psi_{\Omega\ell\epsilon}}{d\xi^2}+\left[4\Omega^2-V_\epsilon(x)\right]
\psi_{\Omega\ell\epsilon}=0.
\end{equation}
The effective potential tends to zero both at the event horizon,
$\xi\to-\infty$, and at spatial infinity, $\xi\to+\infty$. Accordingly, the two local wave sectors at the event horizon are characterised by
\begin{equation}\label{HorizonJostAsymptotics}
f_H^{\mathrm{in}}(x,\Omega)\sim e^{-2i\Omega\xi},\qquad
f_H^{\mathrm{out}}(x,\Omega)\sim e^{+2i\Omega\xi},\qquad
\xi\to-\infty,
\end{equation}
whereas the corresponding sectors at spatial infinity satisfy
\begin{equation}\label{InfinityJostAsymptotics}
f_\infty^{\mathrm{out}}(x,\Omega)\sim e^{+2i\Omega\xi},\qquad
f_\infty^{\mathrm{in}}(x,\Omega)\sim e^{-2i\Omega\xi},\qquad
\xi\to+\infty.
\end{equation}
With the assumed time dependence $e^{-i\omega t}$, $f_H^{\mathrm{in}}$
represents a wave entering the future event horizon, while
$f_\infty^{\mathrm{out}}$ represents a wave propagating towards future null
infinity. Indeed, with advanced and retarded null coordinates
$v=t+r_*$ and $u=t-r_*$, respectively,
$e^{-i\omega t}e^{-i\omega r_*}=e^{-i\omega v}$ is horizon ingoing, while
$e^{-i\omega t}e^{+i\omega r_*}=e^{-i\omega u}$ is infinity outgoing. In terms of the coordinate $x$, the horizon-normalised solutions have
the Frobenius behaviours
\begin{align}
f_H^{\mathrm{in}}(x,\Omega)&=(x-1)^{-2i\Omega}
\sum_{k=0}^{\infty}h_k^{\mathrm{in}}(\Omega)(x-1)^k,\label{HorizonInSeries}\\
f_H^{\mathrm{out}}(x,\Omega)&=(x-1)^{+2i\Omega}\sum_{k=0}^{\infty}h_k^{\mathrm{out}}(\Omega)(x-1)^k,\label{HorizonOutSeries}
\end{align}
where the leading coefficients may be normalised according to $h_0^{\mathrm{in}}=h_0^{\mathrm{out}}=1$. Indeed, substituting a Frobenius expansion into \eqref{ourODE} yields the
indicial equation $\nu(\nu-1)+P_0\nu+Q_0=0$, with $P_0=1$, and $Q_0=4\Omega^2$,
whose roots are $\nu_\pm=\pm2i\Omega$. At spatial infinity, the two formal asymptotic sectors take the form
\begin{align}
f_\infty^{\mathrm{out}}(x,\Omega)&=e^{+2i\Omega x}x^{+2i\Omega}
\sum_{k=0}^{\infty}\frac{c_k^{\mathrm{out}}(\Omega)}{x^k},
\label{InfinityOutSeries}\\
f_\infty^{\mathrm{in}}(x,\Omega)&=
e^{-2i\Omega x}x^{-2i\Omega}\sum_{k=0}^{\infty}\frac{c_k^{\mathrm{in}}(\Omega)}{x^k},\label{InfinityInSeries}
\end{align}
with $c_0^{\mathrm{out}}=c_0^{\mathrm{in}}=1$. These expressions follow either directly from
\eqref{InfinityJostAsymptotics} and \eqref{tortoise}, or from the
formal asymptotic construction~\cite{Olver1994MAA}
\begin{equation}
\psi^{(j)}_{\Omega\ell\epsilon}(x)=x^{\mu_j}e^{\lambda_jx}
\sum_{k=0}^{\infty}\frac{a_{k,j}}{x^k},
\end{equation}
for which $\lambda_\pm=\pm2i\Omega$, and $\mu_\pm=\pm2i\Omega$. For $\operatorname{Im}\Omega>0$, the solution $f_H^{\mathrm{in}}$ decays as $\xi\to-\infty$ and $f_\infty^{\mathrm{out}}$ decays as $\xi\to+\infty$. They can therefore be defined unambiguously in the upper half of the frequency plane and subsequently continued analytically towards the lower half-plane. The terms \emph{ingoing} and \emph{outgoing} used below refer to these analytically continued Jost sectors, rather than merely to local boundedness or to an oscillatory flux interpretation on the negative imaginary axis. For a generic complex frequency, the horizon-ingoing solution may be expanded in the infinity basis as
\begin{equation}\label{ConnectionAtInfinity}
f_H^{\mathrm{in}}=A_{\mathrm{out}}(\Omega)f_\infty^{\mathrm{out}}+A_{\mathrm{in}}(\Omega)f_\infty^{\mathrm{in}},
\end{equation}
whereas the infinity-outgoing solution may be expanded in the horizon
basis according to
\begin{equation}\label{ConnectionAtHorizon}
f_\infty^{\mathrm{out}}=B_{\mathrm{in}}(\Omega)f_H^{\mathrm{in}}+B_{\mathrm{out}}(\Omega)f_H^{\mathrm{out}}.
\end{equation}
A quasinormal frequency is therefore characterised by
\begin{equation}\label{ConnectionQNMCondition}
A_{\mathrm{in}}(\Omega)=0,\qquad\text{equivalently}\qquad
B_{\mathrm{out}}(\Omega)=0,
\end{equation}
so that one and the same nontrivial solution is ingoing at the event
horizon and outgoing at spatial infinity. For any two solutions $u$ and $v$ of \eqref{ourODE}, define the weighted Wronskian
\begin{equation}\label{WeightedWronskian}
\mathcal{W}[u,v](\Omega)=F(x)\left[
u(x,\Omega)\frac{\partial v}{\partial x}(x,\Omega)-v(x,\Omega)\frac{\partial u}{\partial x}(x,\Omega)\right].
\end{equation}
A direct differentiation using \eqref{ourODE} gives
\begin{equation}
\frac{d}{dx}\mathcal{W}[u,v](\Omega)=0,
\end{equation}
and hence, \eqref{WeightedWronskian} is independent of the matching
point $x$. With the normalisations adopted in \eqref{HorizonJostAsymptotics} and
\eqref{InfinityJostAsymptotics}, one has
\begin{equation}
\mathcal{W}\left[f_\infty^{\mathrm{in}},f_\infty^{\mathrm{out}}\right]=
\mathcal{W}\left[f_H^{\mathrm{in}},f_H^{\mathrm{out}}\right]=4i\Omega.
\end{equation}
It follows from \eqref{ConnectionAtInfinity} and \eqref{ConnectionAtHorizon} that the Jost determinant 
\begin{equation}\label{JostDeterminant}
\mathcal{D}_{\ell\epsilon}(\Omega)=\mathcal{W}\left[f_H^{\mathrm{in}},f_\infty^{\mathrm{out}}\right](\Omega)
\end{equation}
satisfies
\begin{equation}\label{JostConnectionRelation}
\mathcal{D}_{\ell\epsilon}(\Omega)=4i\Omega A_{\mathrm{in}}(\Omega)=4i\Omega B_{\mathrm{out}}(\Omega).
\end{equation}
Away from exceptional normalisation points, an isolated zero of
$\mathcal{D}_{\ell\epsilon}$ is therefore the frequency-domain
condition for a QNM pole. The distinction between an isolated zero and a generic point on the negative imaginary axis is essential. The infinity-normalised
Schwarzschild Jost solution, and consequently the frequency-domain
Green function, possess a branch point at $\Omega=0$ and a corresponding branch cut that may be chosen along the negative imaginary axis~\cite{Leaver1986PRD, CasalsOttewill2012PRD, CasalsOttewill2013PRD}. For $\Omega=-i\alpha$ with $\alpha>0$, we therefore introduce the two lateral continuations
\begin{equation}\label{LateralWronskians}
\mathcal{D}_{\ell\epsilon}^{\pm}(-i\alpha)=\lim_{\delta\downarrow0}\mathcal{D}_{\ell\epsilon}\left(\pm\delta-i\alpha\right).
\end{equation}
The superscripts $+$ and $-$ refer respectively to the physical lower-half-plane continuations reached from the right and left without crossing the NIA cut. In general, $\mathcal{D}_{\ell\epsilon}^{+}(-i\alpha)\neq\mathcal{D}_{\ell\epsilon}^{-}(-i\alpha)$. These boundary values should not be identified with unphysical sheets reached by analytic continuation through the cut; such sheets carry additional Schwarzschild structure near the algebraically special frequency~\cite{MaassenVanDenBrink2000PRD, Leung2003CQG}, and that cross-cut continuation is not computed here. Consequently, a finite-dimensional eigenvalue lying on the negative imaginary axis cannot be classified as a physical QNM merely from its numerical position. Mathematically, physical-pole character requires an isolated zero of a specified physical continuation of $\mathcal{D}_{\ell\epsilon}$, subject to the standard noncancellation and normalisation qualifications. Numerically, stability under changes of the matching point, integration contour, precision, and asymptotic truncation is evidence that a computed zero represents that analytic object. Near exceptional frequencies, such as the algebraically special gravitational frequency, possible cancellations between connection coefficients and Jost normalisations must also be examined directly.

\subsection{Green kernel, Jost trace, and determinant line}
\label{subsec:jost_geometry}

The pole criterion can be seen without matrix language. On a fixed
continuation domain and away from numerator cancellations, the radial Green
kernel has the form
\begin{equation}
\label{eq:radial-green-kernel}
G_\Omega(x,x')=
\frac{f_H^{\mathrm{in}}(x_<,\Omega)
      f_\infty^{\mathrm{out}}(x_>,\Omega)}
     {\mathcal D_{\ell\epsilon}(\Omega)},
\qquad
x_<:=\min(x,x'),\quad x_>:=\max(x,x').
\end{equation}
The derivative jump fixes the denominator because the weighted Wronskian is
constant. A zero of $\mathcal D_{\ell\epsilon}$ is therefore the candidate singularity of the continued resolvent, and the usual analytic-Fredholm and noncancellation hypotheses promote it to a Green-function pole. There is also a coordinate-free formulation. Let $U$ be an open subset of one continuation sheet, excluding the branch point and exceptional normalisation frequencies. For each $\Omega\in U$, let $\mathscr S_\Omega$ denote the two-dimensional complex vector space of solutions of the radial Regge--Wheeler equation on $1<x<\infty$. Assume that $\mathscr S:=\bigcup_{\Omega\in U}\{\Omega\}\times\mathscr S_\Omega$ forms a rank-two holomorphic vector bundle over $U$, with fibre $\mathscr S_\Omega$ at $\Omega$. Furthermore, we write
\begin{equation}
\det\mathscr S:=\bigwedge\nolimits^2\mathscr S
\end{equation}
for its determinant line bundle, whose fibre over $\Omega$ is the
one-dimensional vector space
\begin{equation}
(\det\mathscr S)_\Omega=
\bigwedge\nolimits^2\mathscr S_\Omega.
\end{equation}
The horizon-ingoing and infinity-outgoing solutions define holomorphic line subbundles $\mathscr L_H^{\rm in}$ and $\mathscr L_\infty^{\rm out}$ of $\mathscr S$, respectively. The weighted Wronskian provides a holomorphic trivialisation of $\det\mathscr S$ by sending
\begin{equation}
u\wedge v\in\bigwedge\nolimits^2\mathscr S_\Omega
\quad\longmapsto\quad
W[u,v](\Omega)\in\mathbb C.
\end{equation}

\begin{proposition}[Intrinsic Jost divisor]
\label{prop:intrinsic-jost-divisor}
The fibrewise wedge map
\begin{equation}
\label{eq:intrinsic-jost-wedge-map}
\mathscr L_H^{\mathrm{in}}\otimes
\mathscr L_\infty^{\mathrm{out}}
\longrightarrow\det\mathscr S,
\qquad
f_H\otimes f_\infty\longmapsto f_H\wedge f_\infty,
\end{equation}
vanishes at $\Omega$ if and only if the two physical Jost lines coincide. Under the weighted-Wronskian trivialisation introduced above, its scalar representative is $\mathcal D(\Omega)=W[f_H,f_\infty](\Omega)$. If the frames are changed to $\widetilde f_H=a f_H$ and
$\widetilde f_\infty=b f_\infty$, with $a$ and $b$ nowhere-zero and
holomorphic on $U$, then
\begin{equation}
\label{eq:holomorphic-normalization-covariance}
\widetilde{\mathcal D}=ab\mathcal D,
\qquad
\operatorname{div}_U\widetilde{\mathcal D}
=\operatorname{div}_U\mathcal D.
\end{equation}
Hence, zero locations and multiplicities are intrinsic, whereas the magnitude of $\mathcal D$ is normalisation dependent. If $a$ or $b$ is meromorphic, its divisor must be added and separated from the physical divisor.
\end{proposition}

\begin{proof}
Two nonzero vectors in a two-dimensional vector space have zero wedge exactly when they are linearly dependent. The weighted Wronskian is a nonzero alternating bilinear form on each solution space and therefore trivialises the dual determinant line. Under this trivialisation, evaluating the wedge $f_H\wedge f_\infty$ gives its scalar representative $\mathcal W[f_H,f_\infty]=\mathcal D$. Bilinearity gives
$\mathcal W[a f_H,b f_\infty]=ab\mathcal W[f_H,f_\infty]$ and  multiplication by a nowhere-zero holomorphic function leaves the zero divisor unchanged. If $a$ or $b$ is meromorphic instead, the local order of the product is the sum of the local orders, so its divisor is added to that of $\mathcal D$.
\end{proof}

Figure~\ref{fig:determinant_line} makes the normalisation covariance
operational by separating the intrinsic wedge map from its frame-dependent
scalar representative.

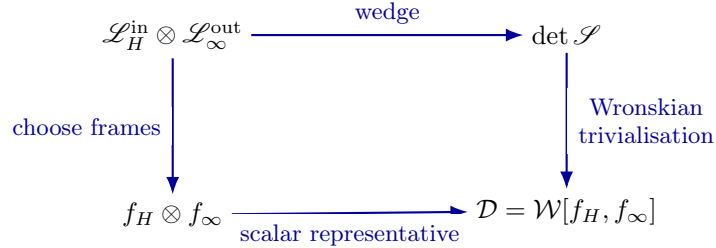
\begin{figure}[htbp]
\centering
\begin{tikzpicture}[
  node distance=18mm and 36mm,
  every node/.style={align=center,font=\normalsize,inner sep=3pt},
  arrowlabel/.style={font=\small,align=center,inner sep=4pt,
    text=blue!60!black,fill=white},
  arr/.style={->,>=Latex,thick,blue!60!black}
]
\node (lines) {$\mathscr L_H^{\mathrm{in}}\otimes
  \mathscr L_\infty^{\mathrm{out}}$};
\node[right=of lines] (det) {$\det\mathscr S$};
\node[below=of lines] (frames) {$f_H\otimes f_\infty$};
\node[below=of det] (scalar)
  {$\mathcal D=\mathcal W[f_H,f_\infty]$};
\draw[arr] (lines) -- node[arrowlabel,above]{wedge} (det);
\draw[arr] (lines) -- node[arrowlabel,left]{choose frames} (frames);
\draw[arr] (det) -- node[arrowlabel,right]{Wronskian\\trivialisation} (scalar);
\draw[arr] (frames) -- node[arrowlabel,below]{scalar representative} (scalar);
\end{tikzpicture}
\caption{The Jost determinant as an Evans-function-type section of a
determinant line. The upper arrow is intrinsic. Choosing local Jost frames and a Wronskian trivialisation produces the scalar determinant below. Nonvanishing holomorphic frame changes rescale the scalar but leave its zero divisor unchanged.}
\label{fig:determinant_line}
\end{figure}

This viewpoint explains why a large nonzero determinant is meaningful only relative to its numerical error and normalisation, while the location and multiplicity of a stable zero, together with the winding of a nonvanishing contour under nowhere-zero holomorphic frame changes, are invariant. It also separates the physical Jost trace from a compactified regularity trace. The former records unwanted scattering coefficients, whereas the latter asks only whether endpoint extensions exist.

\subsection{Boundary factorisation and endpoint nonselection}

For the three finite-pencil realisations developed below, we extract a single common factor containing the formal horizon-ingoing and infinity-outgoing asymptotic behaviours. We write
\begin{equation}\label{Ansatz}
\psi_{\Omega\ell\epsilon}(x)=\mathcal{A}(x,\Omega)\Phi_{\Omega\ell\epsilon}(x),
\qquad
\mathcal{A}(x,\Omega)=x^{4i\Omega}(x-1)^{-2i\Omega}e^{2i\Omega(x-1)}.
\end{equation}
As $x\to1^+$, $\mathcal{A}(x,\Omega)\sim(x-1)^{-2i\Omega}$, whereas, as $x\to+\infty$, $\mathcal{A}(x,\Omega)\sim e^{-2i\Omega}x^{2i\Omega}e^{2i\Omega x}$. The frequency-dependent constant $e^{-2i\Omega}$ has no effect on the
boundary sector. Thus, when the original radial solution belongs to
the horizon-ingoing and infinity-outgoing sectors, $\Phi_{\Omega\ell\epsilon}$ admits regular asymptotic expansions at both endpoints. Regularity of $\Phi_{\Omega\ell\epsilon}$ is, however, not equivalent to the Jost conditions. To see this, consider first the opposite horizon sector. Dividing $f_H^{\mathrm{out}}$ by the factor $\mathcal{A}(x,\Omega)$ gives
\begin{equation}\label{WrongHorizonRemainder}
\Phi_H^{\mathrm{out}}(x,\Omega)
\sim
(x-1)^{4i\Omega},
\qquad
x\to1^+.
\end{equation}
Similarly, the infinity-incoming solution gives
\begin{equation}\label{WrongInfinityRemainder}
\Phi_\infty^{\mathrm{in}}(x,\Omega)\sim x^{-4i\Omega}e^{-4i\Omega x},\qquad
x\to+\infty,
\end{equation}
up to a nonzero frequency-dependent constant and an inverse-power
asymptotic series. On the negative imaginary axis, $\Omega=-i\alpha$ with $\alpha>0$, these expressions become
\begin{equation}\label{WrongRemaindersNIA}
\Phi_H^{\mathrm{out}}(x,-i\alpha)\sim(x-1)^{4\alpha},\qquad
\Phi_\infty^{\mathrm{in}}(x,-i\alpha)\sim x^{-4\alpha}e^{-4\alpha x}.
\end{equation}
Both unwanted sectors therefore tend to zero at the corresponding endpoint. For generic $4\alpha\notin\mathbb{N}$, the horizon remainder is nonanalytic but possesses $\lfloor4\alpha\rfloor$ continuous derivatives. At resonant values
$4\alpha\in\mathbb{N}$, the second local solution may instead involve a Frobenius logarithm or, in exceptional cases, a second analytic solution. At infinity, setting $z=1/x$ gives
\begin{equation}
\Phi_\infty^{\mathrm{in}}(z,-i\alpha)\sim z^{4\alpha}e^{-4\alpha/z},\qquad
z\to 0^+,
\end{equation}
which is $C^\infty$ and flat at $z=0$ when extended by zero, but is
not analytic there. It follows that endpoint boundedness, smoothness, or apparent
Chebyshev convergence does not by itself force the unwanted connection coefficients in \eqref{ConnectionQNMCondition} to vanish. The compactified spectral problems below are therefore used to generate and diagnose candidate frequencies, while the Jost determinant \eqref{JostDeterminant} supplies the independent pole criterion. Substitution of \eqref{Ansatz} into \eqref{ourODE} gives
\begin{equation}\label{ReducedFactorisedEquation}
\mathcal{P}_2(x,\Omega)\Phi_{\Omega\ell\epsilon}''+\mathcal{P}_1(x,\Omega)
\Phi_{\Omega\ell\epsilon}'+\mathcal{P}_0(x,\Omega)\Phi_{\Omega\ell\epsilon}=0
\end{equation}
with
\begin{equation}
\mathcal{P}_2(x,\Omega)=x^2(x-1),\quad
\mathcal{P}_1(x,\Omega)=4i\Omega x^3-8i\Omega x+x,\quad
\mathcal{P}_0(x,\Omega)=[16\Omega^2-\ell(\ell+1)]x+8i\Omega+16\Omega^2-\epsilon.
\end{equation}
The indicial equation of \eqref{ReducedFactorisedEquation} at $x=1$ is
\begin{equation}\label{FactorisedIndicialEquation}
\sigma(\sigma-4i\Omega)=0.
\end{equation}
The exponent $\sigma=0$ corresponds to the factorised horizon-ingoing sector, whereas $\sigma=4i\Omega$ corresponds to the factorised horizon-outgoing sector. For $\Omega=-i\alpha$, both exponents have nonnegative real part because $\sigma_1=0$, and $\sigma_2=4\alpha>0$. Equation \eqref{FactorisedIndicialEquation} therefore makes explicit
why boundedness of the factorised field cannot distinguish the two
horizon sectors on the negative imaginary axis. Both compactifications introduced below will be applied to the same reduced equation \eqref{ReducedFactorisedEquation}. In this way, any difference between the two resulting spectra can be attributed to the radial compactification and discretisation rather than to a change in the asymptotic factorisation. We now turn this local observation into a statement about the continuum endpoint domain. For a boundary coordinate $t\in(0,\varepsilon)$, let $C_+^k$ denote the functions admitting a one-sided $C^k$ extension to $t=0$.

\begin{lemma}[Power--log regularity and exponential flatness]
\label{lem:power-log-flatness}
Let $k\in\mathbb N_0$.
\begin{enumerate}[label=(\roman*),leftmargin=2.2em]
\item 
If $\beta>k$ and $a,b\in C^\infty([0,\varepsilon])$, then
$t^\beta\{a(t)+b(t)\log t\}$, extended by zero at $t=0$, belongs to
$C_+^k$.
\item 
If $c,q>0$, $\gamma\in\mathbb R$, and $a\in C^\infty([0,\varepsilon])$, then $t^\gamma e^{-c/t^q}a(t)$, extended by zero at $t=0$, belongs to $C_+^\infty$ and every derivative at $t=0$ vanishes.
\end{enumerate}
\end{lemma}

\begin{proof}
After $j$ differentiations, each term in the first function is a finite
linear combination of terms bounded by $t^{\beta-j}(1+|\log t|)$ as $t\downarrow0$. This tends to zero for $0\leqslant j\leqslant k$ because $\beta>k$. Thus, the one-sided limits of the function and its first $k$ derivatives agree with those of the zero extension, which is therefore $C_+^k$. After any fixed number of differentiations, the second function is a finite sum of terms $t^{-N}e^{-c/t^q}a_N(t)$, with finite $N$ and smooth $a_N$. Since $e^{-c/t^q}=o(t^m)$ as $t\downarrow0$ for every $m>0$, every such term tends to zero. The zero extension is therefore smooth and flat.
\end{proof}

For each frequency $\Omega$, let us recall that $\mathscr S_\Omega$ denotes the two-dimensional complex solution space of the original radial Regge--Wheeler equation on $1<x<\infty$. Furthermore, let
$\mathscr F_\Omega$ denote the corresponding solution space of the
factorised equation. Multiplication by the factor $A(\cdot,\Omega)$
defines a linear isomorphism $T_\Omega:\mathscr F_\Omega\longrightarrow\mathscr S_\Omega$ with $T_\Omega\Phi=A(\cdot,\Omega)\Phi$.

\begin{theorem}[Finite endpoint regularity does not select the Jost sectors]
\label{thm:finite-regularity-nonselection}
Fix $\Omega=-i\alpha$ with $\alpha>0$, and let $\mathscr F_\Omega$ be the two-dimensional solution space of the factorised Regge--Wheeler equation on the open Schwarzschild exterior. Let $t$ and $u$ be defining functions for the horizon and infinity. Assume the following standard local-basis properties.
\begin{enumerate}[label=(\roman*),leftmargin=2.2em]
\item 
Near the horizon, $x-1=t h(t)$ with $h\in C^\infty([0,\varepsilon])$ and $h(0)>0$. If $4\alpha\notin\mathbb N$, the factorised equation has a basis
\begin{equation}
\Phi_{H,0}(t)=a_0(t),\qquad
\Phi_{H,1}(t)=t^{4\alpha}a_1(t),
\end{equation}
where $a_0,a_1$ are smooth and nonzero at $t=0$. If
$m=4\alpha\in\mathbb N$, it has a Frobenius basis of the form
\begin{equation}
\Phi_{H,0}(t)=a_0(t)+c\,t^m\log t\,a_1(t),\qquad
\Phi_{H,1}(t)=t^m a_1(t),
\end{equation}
where the resonant coefficient $c$ is allowed to vanish.
\item 
Near infinity, $x=C u^{-q}+r(u)$ for some $C,q>0$ and smooth
$r$. The factorised outgoing Jost sector belongs to $C_+^k$, while the
opposite sector has a representative
\begin{equation}
\Phi_{\infty,1}(u)=u^{4\alpha q}e^{-4\alpha C/u^q}b(u),
\end{equation}
with smooth $b$. It is equivalent to assume sectorial Jost asymptotics with symbol estimates under differentiation.
\end{enumerate}
If $4\alpha>k$, both local solution sectors belong to $C_+^k$ at each
endpoint. Consequently,
\begin{equation}
\label{eq:regularity-domain-is-maximal}
\left\{\Phi\in\mathscr F_\Omega:
\Phi\text{ is one-sided $C^k$ at both compactified endpoints}\right\}
=\mathscr F_\Omega.
\end{equation}
In particular, endpoint $C^k$ regularity cannot force either the
horizon-outgoing coefficient or the infinity-incoming coefficient to
vanish. More precisely, no homogeneous local endpoint condition depending only on jets of order at most $k$ and satisfied by the desired sector can exclude the opposite sector at that endpoint.
\end{theorem}

\begin{proof}
At the horizon, the smooth positive factor $h(t)^{4\alpha}$ is absorbed into $a_1(t)$. When $4\alpha\notin\mathbb N$, the first part of Lemma~\ref{lem:power-log-flatness} gives
$\Phi_{H,1}\in C_+^k$ whenever $4\alpha>k$, while $\Phi_{H,0}$ is smooth. At a resonance $m=4\alpha\in\mathbb N$, the same lemma gives
$t^m\log t\in C_+^k$ for $m>k$, and both Frobenius basis elements again belong to $C_+^k$. At infinity, the desired sector belongs to $C_+^k$ by hypothesis and the opposite sector is smooth and flat by the second part of Lemma~\ref{lem:power-log-flatness}. Every linear combination of the two local basis elements therefore has the required regularity. Every global interior solution has such a representation at both endpoints, proving Eq.~\eqref{eq:regularity-domain-is-maximal}. Membership in the regularity class consequently places no restriction on either unwanted Jost coefficient. A homogeneous local finite-jet condition that admits the desired sector also admits the zero jet. The unwanted infinity sector has zero $k$-jet by flatness, and the unwanted horizon sector has zero $k$-jet because each of its derivatives through order $k$ tends to zero when $4\alpha>k$. Hence,  neither sector can be excluded by such a condition.
\end{proof}

This theorem is a no-go result for every fixed finite regularity order, not an identification of the finite roots. It does not exclude a genuine Jost zero, prove convergence to the cut, or replace the sectorial asymptotic hypothesis at the irregular infinity endpoint. Requiring more finite derivatives merely moves the horizon threshold to $\alpha>k/4$. Finally, the flat infinity sector remains invisible to every finite jet.

\subsection{First compactification}

Let us now introduce the transformation $x=2/(1-y)$ mapping the point at infinity and the event horizon to $y = 1$ and $y = -1$, respectively. Furthermore, a dot denotes differentiation with respect to the new variable $y$. Then, equation \eqref{ReducedFactorisedEquation} becomes
\begin{equation}\label{ODEynone}
\mathcal{S}_2(y)\ddot{\Phi}_{\Omega\ell\epsilon}(y) + \mathcal{S}_1(y,\Omega)\dot{\Phi}_{\Omega\ell\epsilon}(y) + \mathcal{S}_0(y,\Omega)\Phi_{\Omega\ell\epsilon}(y) = 0,
\end{equation}
where
\begin{align}
\mathcal{S}_2(y)
&=1-y^2,\\
\mathcal{S}_1(y,\Omega)
&=-2(1+y)
 +\frac{(1-8i\Omega)(y^2-2y)+1+8i\Omega}{1-y},\\
\mathcal{S}_0(y,\Omega)
&=\frac{2[16\Omega^2-\ell(\ell+1)]}{1-y}
 +8i\Omega+16\Omega^2-\epsilon.
\end{align}
A straightforward inspection of these coefficients indicates that they do not have any zero in common at the endpoints $y=\pm 1$ and $y=1$ is a simple pole common to $\mathcal{S}_1(y,\Omega)$ and $\mathcal{S}_0(y,\Omega)$. For generic $\Omega$, the coefficients $\mathcal{S}_1$ and $\mathcal{S}_0$ have simple poles at $y=1$. To clear these denominators and obtain polynomial coefficient functions suitable for matrix assembly, we multiply Eq.~\eqref{ODEynone} by $(1-y)$. This multiplication leaves the differential equation unchanged on the open interval $-1<y<1$. Moreover, the endpoint equations are defined through the one-sided limits derived below. As a result, we end up with the following differential equation
\begin{equation}\label{ODEhynone}
\mathcal{M}_2(y)\ddot{\Phi}_{\Omega\ell\epsilon}(y) + \mathcal{M}_1(y,\Omega)\dot{\Phi}_{\Omega\ell\epsilon}(y) + \mathcal{M}_0(y,\Omega)\Phi_{\Omega\ell\epsilon}(y) = 0,
\end{equation}
where
\begin{equation}\label{S210honone}
\mathcal{M}_2(y)=(1-y)^2(1+y), \qquad
\mathcal{M}_1(y,\Omega)=i\Omega N_1(y)+N_0(y), \qquad
\mathcal{M}_0(y,\Omega)=\Omega^2 C_2(y)+i\Omega C_1(y)+C_0(y)
\end{equation}
with
\begin{eqnarray}
N_1(y)&=&8(1+2y-y^2),\quad
N_0(y)=3y^2-2y-1,\quad
C_2(y)=16(3-y).\label{N01}\\
C_1(y)&=&8(1-y),\quad
C_0(y)=-2\ell(\ell+1)-\epsilon(1-y).\label{C012}
\end{eqnarray}
The endpoint limits are
\begin{align}
\lim_{y\to 1^-}\mathcal{M}_2(y)
&=0=
\lim_{y\to -1^+}\mathcal{M}_2(y),\\
\lim_{y\to 1^-}\mathcal{M}_1(y,\Omega)
&=16i\Omega,
&\lim_{y\to -1^+}\mathcal{M}_1(y,\Omega)
&=-16i\Omega+4,\\
\lim_{y\to 1^-}\mathcal{M}_0(y,\Omega)
&=32\Omega^2-2\ell(\ell+1),
&\lim_{y\to -1^+}\mathcal{M}_0(y,\Omega)
&=64\Omega^2+16i\Omega-2\ell(\ell+1)-2\epsilon.
\end{align}
The vanishing of the principal coefficient $\mathcal{M}_2(y,\Omega)$ at both endpoints implies that the limiting equations are of first order. More precisely, suppose that $\Phi_{\Omega\ell\epsilon}(y)$ admits a one-sided $C^2([-1,1])$ expansion at the corresponding endpoint. Then,
the limiting equations at the compactified endpoints are
\begin{equation}\label{FirstCompactificationInfinityRelation}
8i\Omega\dot{\Phi}_{\Omega\ell\epsilon}(1)+\left[16\Omega^2-\ell(\ell+1)\right]\Phi_{\Omega\ell\epsilon}(1)=0
\end{equation}
and
\begin{equation}\label{FirstCompactificationHorizonRelation}
2(1-4i\Omega)\dot{\Phi}_{\Omega\ell\epsilon}(-1)+\left[32\Omega^2+8i\Omega-\ell(\ell+1)-\epsilon\right]\Phi_{\Omega\ell\epsilon}(-1)=0.
\end{equation}
Equations \eqref{FirstCompactificationInfinityRelation} and
\eqref{FirstCompactificationHorizonRelation} are not additional
boundary conditions imposed independently of the differential equation. They are the one-sided endpoint limits of Eq.~\eqref{ODEhynone} for the regular branch. In a Chebyshev--Lobatto implementation, they replace the two differential equation rows at $y=\pm 1$. When a Chebyshev roots grid is used, the endpoints are absent from the collocation set, and the same relations can be employed as a posteriori checks of the reconstructed eigenfunctions. These endpoint relations are necessary regularity conditions, but they do not by themselves identify the required Jost sectors. To see this, let $\Omega=-i\alpha$, with $\alpha>0$. Under the present compactification, the factorised representative of the unwanted horizon-outgoing solution behaves as
\begin{equation}
\Phi[f_H^{\mathrm{out}}](y,-i\alpha)
\sim
\left(\frac{1+y}{1-y}\right)^{4\alpha},
\qquad
y\to-1^+,
\end{equation}
whereas the factorised infinity-incoming solution behaves as
\begin{equation}
\Phi[f_\infty^{\mathrm{in}}](y,-i\alpha)\sim
\left(\frac{1-y}{2}\right)^{4\alpha}\exp\left(-\frac{8\alpha}{1-y}\right),\qquad
y\to 1^-.
\end{equation}
For sufficiently large $\alpha$, the unwanted horizon solution and its
first derivatives vanish at $y=-1$, while the unwanted infinity solution
is $C^\infty$ and flat at $y=1$. Consequently, both unwanted sectors may
satisfy the limiting relations trivially. The endpoint equations therefore
do not replace the analytically continued Jost--Wronskian test. In the final step leading to the application of the spectral method, we recast the differential equation \eqref{ODEhynone} into the following form
\begin{equation}\label{TSCH}
\left(L_0+i\Omega L_1+\Omega^2L_2\right)
[\Phi_{\Omega\ell\epsilon}]=0,
\end{equation}
where the differential operators are
\begin{align}
L_0[\Phi]
&=L_{00}(y)\Phi+L_{01}(y)\dot\Phi+L_{02}(y)\ddot\Phi,
\label{L0none}\\
L_1[\Phi]
&=L_{10}(y)\Phi+L_{11}(y)\dot\Phi+L_{12}(y)\ddot\Phi,
\label{L1none}\\
L_2[\Phi]
&=L_{20}(y)\Phi+L_{21}(y)\dot\Phi+L_{22}(y)\ddot\Phi.
\label{L2none}
\end{align}
Table~\ref{tableZweinone} identifies the coefficient functions and their
one-sided endpoint limits.
\begin{table}[t]
\caption{Definitions of the coefficients $L_{ij}$ and their one-sided limits
at the endpoints of $-1\leqslant y\leqslant 1$.}
\label{tableZweinone}
\centering
\small
\begingroup
\setlength{\tabcolsep}{3pt}
\renewcommand{\arraystretch}{1.08}
\begin{tabular}{@{}cccc@{}}
\toprule
$(i,j)$
& \shortstack{Horizon limit\\$\displaystyle\lim_{y\to-1^+}L_{ij}(y)$}
& $L_{ij}(y)$
& \shortstack{Infinity limit\\$\displaystyle\lim_{y\to1^-}L_{ij}(y)$} \\
\midrule
$(0,0)$ & $-2\ell(\ell+1)-2\epsilon$ & $C_0(y)$ & $-2\ell(\ell+1)$ \\
$(0,1)$ & $4$ & $N_0(y)$ & $0$ \\
$(0,2)$ & $0$ & $\mathcal{M}_2(y)$ & $0$ \\
\addlinespace[2pt]
$(1,0)$ & $16$ & $C_1(y)$ & $0$ \\
$(1,1)$ & $-16$ & $N_1(y)$ & $16$ \\
$(1,2)$ & $0$ & $0$ & $0$ \\
\addlinespace[2pt]
$(2,0)$ & $64$ & $C_2(y)$ & $32$ \\
$(2,1)$ & $0$ & $0$ & $0$ \\
$(2,2)$ & $0$ & $0$ & $0$ \\
\bottomrule
\end{tabular}
\endgroup
\end{table}

\subsection{Second compactification: quadratic map}

In this subsection the symbol $y$ denotes the coordinate defined by
Eq.~\eqref{secondcompactification}. It should not be confused with the
coordinate used in the first compactification. As a second representation of the Schwarzschild exterior, we introduce the quadratic compactification
\begin{equation}\label{secondcompactification}
x=\frac{4}{(1-y)^2},\qquad
y=1-\frac{2}{\sqrt{x}},\qquad
-1\leqslant y<1.
\end{equation}
This transformation maps the event horizon $x=1$ to $y=-1$ and spatial
infinity $x\to+\infty$ to $y\to 1^-$. In contrast with the first
compactification, for which $x\sim2/(1-y)$ at infinity, the present map
satisfies
\begin{equation}
x\sim\frac{4}{(1-y)^2},\qquad
y\to1^-.
\end{equation}
It therefore produces a different algebraic distribution of the physical radial collocation points. Near the event horizon, however,
\begin{equation}
x-1=(1+y)+\mathcal{O}\bigl((1+y)^2\bigr),\qquad
y\to-1^+,
\end{equation}
so that the local Frobenius exponents are preserved. Under \eqref{secondcompactification}, equation \eqref{ReducedFactorisedEquation} becomes
\begin{equation}\label{ODEynoneII}
\widehat{\mathcal{S}}_2(y)\ddot{\Phi}_{\Omega\ell\epsilon}(y) + \widehat{\mathcal{S}}_1(y,\Omega)\dot{\Phi}_{\Omega\ell\epsilon}(y) + \widehat{\mathcal{S}}_0(y,\Omega)\Phi_{\Omega\ell\epsilon}(y) = 0,
\end{equation}
where
\begin{align}
\widehat{\mathcal{S}}_2(y)
&=\frac{1}{4}(1+y)(3-y),\\
\widehat{\mathcal{S}}_1(y,\Omega)
&=\frac{3(1+y)(y-3)}{4(1-y)}
 +\frac{1}{2}\left[
 \frac{64i\Omega}{(1-y)^3}+(1-8i\Omega)(1-y)
 \right],\\
\widehat{\mathcal{S}}_0(y,\Omega)
&=\frac{4[16\Omega^2-\ell(\ell+1)]}{(1-y)^2}
 +16\Omega^2+8i\Omega-\epsilon.
\end{align}
A direct inspection of these coefficients indicates that they do not have any zero in common at the endpoints $y=\pm 1$ and $y=1$ is a pole of order three for $\widehat{\mathcal{S}}_1(y,\Omega)$ while it is a pole of order two for $\widehat{\mathcal{S}}_0(y,\Omega)$. For generic $\Omega\neq 0$, the coefficient
$\widehat{\mathcal{S}}_1(y,\Omega)$ has a pole of order three at $y=1$, while
$\widehat{\mathcal{S}}_0(y,\Omega)$ has a pole of order two there. No
frequency-independent factor vanishing at an endpoint is common to all
three coefficients. To clear these denominators and obtain polynomial
coefficient functions suitable for matrix assembly, we multiply
Eq.~\eqref{ODEynoneII} by $(1-y)^3$. This operation leaves the
differential equation unchanged on the open interval $-1<y<1$. Moreover, its endpoint equations are defined separately through the one-sided limits
derived below. As a result, we end up with the following differential equation
\begin{equation}\label{ODEhynoneII}
\widehat{\mathcal{M}}_2(y)\ddot{\Phi}_{\Omega\ell\epsilon}(y) + \widehat{\mathcal{M}}_1(y,\Omega)\dot{\Phi}_{\Omega\ell\epsilon}(y) + \widehat{\mathcal{M}}_0(y,\Omega)\Phi_{\Omega\ell\epsilon}(y) = 0,
\end{equation}
where
\begin{equation}\label{S210hononeII}
\widehat{\mathcal{M}}_2(y)=\frac{1}{4}(1-y)^3(1+y )(3-y), \qquad
\widehat{\mathcal{M}}_1(y,\Omega)=i\Omega\widehat{N}_1(y)+\widehat{N}_0(y), \qquad
\widehat{\mathcal{M}}_0(y,\Omega)=\Omega^2\widehat{C}_2(y)+i\Omega\widehat{C}_1(y)+\widehat{C}_0(y)
\end{equation}
with
\begin{align}
\widehat N_1(y)
&=32-4(1-y)^4,\\
\widehat N_0(y)
&=\frac{1}{2}(1-y)^2
 \left[\frac{3}{2}(1+y)(y-3)+(1-y)^2\right],
\label{N01II}\\
\widehat C_2(y)
&=16(1-y)(y^2-2y+5),\\
\widehat C_1(y)
&=8(1-y)^3,\\
\widehat C_0(y)
&=-(1-y)\left[4\ell(\ell+1)+\epsilon(1-y)^2\right].
\label{C012II}
\end{align}
The corresponding endpoint limits are
\begin{align}
\lim_{y\to 1^-}\widehat{\mathcal M}_2(y)
&=0=
\lim_{y\to -1^+}\widehat{\mathcal M}_2(y),\\
\lim_{y\to 1^-}\widehat{\mathcal M}_1(y,\Omega)
&=32i\Omega,
&\lim_{y\to -1^+}\widehat{\mathcal M}_1(y,\Omega)
&=-32i\Omega+8,\\
\lim_{y\to 1^-}\widehat{\mathcal M}_0(y,\Omega)
&=0,
&\lim_{y\to -1^+}\widehat{\mathcal M}_0(y,\Omega)
&=256\Omega^2+64i\Omega-8\ell(\ell+1)-8\epsilon.
\end{align}
Suppose that $\Phi_{\Omega\ell\epsilon}(y)$ admits a one-sided $C^2([-1,1])$ expansion at the corresponding endpoint. Since $\widehat{\mathcal M}_2(y)\to 0$ as $y\to\pm 1$, one then has $\widehat{\mathcal M}_2(y)\ddot{\Phi}_{\Omega\ell\epsilon}(y)\longrightarrow 0$. The limiting equation at spatial infinity is
\begin{equation}\label{SecondCompactificationInfinityFull}
32i\Omega\dot{\Phi}_{\Omega\ell\epsilon}(1)=0.
\end{equation}
For $\Omega\neq 0$, this reduces to
\begin{equation}\label{SecondCompactificationInfinityRelation}
\dot{\Phi}_{\Omega\ell\epsilon}(1)=0.
\end{equation}
At the event horizon, the limiting equation is
\begin{equation}\label{SecondCompactificationHorizonRelation}
(1-4i\Omega)\dot{\Phi}_{\Omega\ell\epsilon}(-1)+\left[32\Omega^2+8i\Omega
-\ell(\ell+1)-\epsilon\right]\Phi_{\Omega\ell\epsilon}(-1)=0.
\end{equation}
These limiting relations can be used directly as endpoint rows in an extrema/Lobatto implementation. If a roots grid is used, the endpoints are not collocation points, but the relations remain useful for independently checking the reconstructed eigenfunctions. As in the first compactification, these endpoint relations are necessary regularity conditions but are not sufficient to select the required Jost sectors. For $\Omega=-i\alpha$, with $\alpha>0$, the factorised representative of the unwanted horizon-outgoing solution behaves as
\begin{equation}
\Phi[f_H^{\mathrm{out}}](y,-i\alpha)\sim(1+y)^{4\alpha},\qquad
y\to-1^+,
\end{equation}
up to a nonzero multiplicative constant. At infinity, the factorised
incoming solution behaves as
\begin{equation}
\Phi[f_\infty^{\mathrm{in}}](y,-i\alpha)\sim4^{-4\alpha}(1-y)^{8\alpha}
\exp\left[-\frac{16\alpha}{(1-y)^2}\right],\qquad
y\to 1^-.
\end{equation}
The latter function is $C^\infty$ and flat at $y=1$, while the former
becomes increasingly differentiable as $\alpha$ grows. Consequently,
the unwanted sectors can satisfy the endpoint regularity relations
without their connection coefficients vanishing. The relations above
therefore do not replace the lateral Jost--Wronskian test. In the final step leading to the application of the spectral method, we recast the differential equation \eqref{ODEhynoneII} into the following form
\begin{equation}\label{TSCHII}
\left(\widehat L_0+i\Omega\widehat L_1+\Omega^2\widehat L_2\right)
[\Phi_{\Omega\ell\epsilon}]=0,
\end{equation}
where
\begin{align}
\widehat L_0[\Phi]
&=\widehat L_{00}(y)\Phi+\widehat L_{01}(y)\dot\Phi
 +\widehat L_{02}(y)\ddot\Phi,
\label{L0noneII}\\
\widehat L_1[\Phi]
&=\widehat L_{10}(y)\Phi+\widehat L_{11}(y)\dot\Phi
 +\widehat L_{12}(y)\ddot\Phi,
\label{L1noneII}\\
\widehat L_2[\Phi]
&=\widehat L_{20}(y)\Phi+\widehat L_{21}(y)\dot\Phi
 +\widehat L_{22}(y)\ddot\Phi.
\label{L2noneII}
\end{align}
Table~\ref{tableZweinoneII} records the coefficient functions and their
one-sided endpoint limits.
\begin{table}[t]
\caption{Definitions of the coefficients $\widehat L_{ij}$ and their
one-sided limits at the endpoints of $-1\leqslant y\leqslant 1$.}
\label{tableZweinoneII}
\centering
\small
\begingroup
\setlength{\tabcolsep}{3pt}
\renewcommand{\arraystretch}{1.08}
\begin{tabular}{@{}cccc@{}}
\toprule
$(i,j)$
& \shortstack{Horizon limit\\$\displaystyle\lim_{y\to-1^+}\widehat L_{ij}(y)$}
& $\widehat L_{ij}(y)$
& \shortstack{Infinity limit\\$\displaystyle\lim_{y\to1^-}\widehat L_{ij}(y)$} \\
\midrule
$(0,0)$ & $-8\ell(\ell+1)-8\epsilon$ & $\widehat C_0(y)$ & $0$ \\
$(0,1)$ & $8$ & $\widehat N_0(y)$ & $0$ \\
$(0,2)$ & $0$ & $\widehat{\mathcal M}_2(y)$ & $0$ \\
\addlinespace[2pt]
$(1,0)$ & $64$ & $\widehat C_1(y)$ & $0$ \\
$(1,1)$ & $-32$ & $\widehat N_1(y)$ & $32$ \\
$(1,2)$ & $0$ & $0$ & $0$ \\
\addlinespace[2pt]
$(2,0)$ & $256$ & $\widehat C_2(y)$ & $0$ \\
$(2,1)$ & $0$ & $0$ & $0$ \\
$(2,2)$ & $0$ & $0$ & $0$ \\
\bottomrule
\end{tabular}
\endgroup
\end{table}
\FloatBarrier

\begin{corollary}[Nonselection for C1 and C2]
\label{cor:c1-c2-c2-nonselection}
Assume the local Jost-basis hypotheses of Theorem~\ref{thm:finite-regularity-nonselection}. For either $x_1=2/(1-y_1)$ or $x_2=4/(1-y_2)^2$, and for every $\Omega=-i\alpha$ with $\alpha>1/2$, the one-sided $C^2$ endpoint class contains all of $\mathscr F_\Omega$. The limiting endpoint equations obtained from the compactified differential
equations are therefore necessary identities for this regularity class, not independent radiation conditions. Imposing only those equations and one-sided $C^2$ regularity does not define a discrete set of frequencies on that part of the NIA.
\end{corollary}

\begin{proof}
For C1, we find
\begin{equation}
x_1-1=\frac{1+y_1}{1-y_1},\qquad
x_1=2(1-y_1)^{-1},
\end{equation}
so the horizon map is a smooth local diffeomorphism and the infinity
parameters in Theorem~\ref{thm:finite-regularity-nonselection} are
$(C,q)=(2,1)$. For C2, we have
\begin{equation}
x_2-1=\frac{(1+y_2)(3-y_2)}{(1-y_2)^2},\qquad
x_2=4(1-y_2)^{-2},
\end{equation}
so the horizon map is again a smooth local diffeomorphism and
$(C,q)=(4,2)$. The threshold $4\alpha>2$ is exactly $\alpha>1/2$. For the unwanted horizon sector, the value and first two derivatives needed
by the one-sided $C^2$ class exist when $4\alpha>2$, including the possible resonant power--log term. At infinity, every derivative of the unwanted sector vanishes. The desired sectors satisfy the limiting equations because those equations are obtained by taking the one-sided limit of the differential equation for a $C^2$ solution. Hence, both local basis elements, and therefore every global solution, satisfy the stated endpoint class and limiting rows.
\end{proof}

At this point a remark is in order. The threshold $\alpha>1/2$ pertains to the full one-sided $C^2$ domain. The displayed first-order endpoint rows are already insensitive to the unwanted horizon sector when $4\alpha>1$. In that regime, $\Phi$, $\Phi'$, and the product of the vanishing principal coefficient with $\Phi''$ all tend to zero. Thus, the row-level threshold is $\alpha>1/4$, including the resonant power--log cases. Consequently, throughout the reported damping interval, the isolated pencil roots cannot be attributed to spectral selection by the stated continuum $C^2$ endpoint conditions, since those conditions admit every interior solution. Their isolation arises instead at the level of the finite-dimensional realisation, through polynomial truncation, collocation, and the implementation of the endpoint rows. This observation does not determine the collective limiting interpretation of the roots, nor does it establish convergence to a weighted branch-cut representation. The physical classification of each candidate remains governed independently by the lateral Jost condition $\mathcal D^\pm(\Omega)=0$.

\section{Finite polynomial pencils, exact grid geometry, and diagnostics}
\label{sec:numerical_method}

The compactified frequency-domain equations previously derived have the common quadratic-operator form
\begin{equation}
\label{eq:operator_polynomial_numerics}
\mathcal{Q}(\Omega)\Phi
\equiv
\left(\mathcal{L}_0+i\Omega\mathcal{L}_1+\Omega^2\mathcal{L}_2\right)\Phi
=0,
\end{equation}
where $\mathcal{L}_j=L_j$ for the first compactification in Eq.~\eqref{TSCH}, and $\mathcal{L}_j=\widehat{L}_j$ for the quadratic compactification in Eq.~\eqref{TSCHII}. The discussion below concerns these quadratic frequency-domain pencils. The purpose of the discretisation is to generate and assess candidate frequencies of the compactified problems. It is not used, by itself, to identify zeros of the analytically continued Jost determinant. The latter distinction is essential on the NIA and will be maintained throughout the numerical analysis.

\subsection{Coordinate pullback does not commute with truncation}
\label{subsec:pullback_noncommutation}

The two differential equations represent the same physical exterior, but equal-degree polynomial trial spaces are not mapped into one another. Write $y_1$ and $y_2$ for the C1 and C2 coordinates. At a common radius,
\begin{equation}
\label{eq:compactification-transition-map}
y_1=\tau(y_2):=1-\frac{(1-y_2)^2}{2},\qquad
y_2=\tau^{-1}(y_1):=1-\sqrt{2(1-y_1)}.
\end{equation}
For a function $f=f(y_1)$ expressed in the $\mathrm{C1}$ coordinate,
we use the standard pullback notation
\begin{equation}
\tau^*f:=f\circ\tau,\qquad
(\tau^*f)(y_2)=f\bigl(\tau(y_2)\bigr).
\end{equation}
Thus, $\tau^*$ maps functions expressed in the $\mathrm{C1}$ coordinate
to functions expressed in the $\mathrm{C2}$ coordinate. Similarly,
$(\tau^{-1})^*g:=g\circ\tau^{-1}$ for a function $g=g(y_2)$.

\begin{proposition}[Noncommutation of pullback and polynomial truncation]
\label{prop:pullback-truncation-noncommutation}
Let $\mathbb P_d(y_j)$ denote the space of polynomials of degree at
most $d$ in the coordinate $y_j$. Then,
\begin{equation}\label{eq:c1-to-c2-polynomial-pullback}
\tau^*\bigl(\mathbb P_d(y_1)\bigr):=\left\{\tau^*P:P\in\mathbb P_d(y_1)\right\}=
\left\{P\circ\tau:P\in\mathbb P_d(y_1)\right\}
\subset\mathbb P_{2d}(y_2).
\end{equation}
If $P\in\mathbb P_d(y_1)$ has exact degree $d$, then $\tau^*P$ has
exact degree $2d$. Thus, $\tau^*(\mathbb P_{n-1}(y_1))$ is not contained in $\mathbb P_{n-1}(y_2)$ for $n\geqslant2$. In the reverse direction, a generic $\mathrm{C2}$ polynomial $Q=Q(y_2)$ becomes
$((\tau^{-1})^*Q)(y_1)=Q\bigl(1-\sqrt{2(1-y_1)}\bigr)$, which is generally nonpolynomial and need not be $C^1$ at $y_1=1$. Hence, equal-degree C1 and C2 trial spaces are not related by the exact continuum coordinate change. More precisely, their exact common polynomial core is
\begin{equation}
\label{eq:exact-pullback-common-core}
\tau^*(\mathbb P_d(y_1))\cap\mathbb P_d(y_2)
=\tau^*(\mathbb P_{\lfloor d/2\rfloor}(y_1)),
\end{equation}
which has dimension $\lfloor d/2\rfloor+1$. Equivalently,
\begin{equation}
\label{eq:inverse-polynomial-common-core}
\left\{Q\in\mathbb P_d(y_2):Q\circ\tau^{-1}\in\mathbb P_d(y_1)\right\}
=\left\{R((1-y_2)^2):R\in\mathbb P_{\lfloor d/2\rfloor}\right\}.
\end{equation}
Therefore, exact containment of the C1 degree-$(n-1)$ trial space after
pullback requires C2 degree at least $2(n-1)$, and hence, at least $2n-1$
coefficients.
\end{proposition}

\begin{proof}
Equation~\eqref{eq:compactification-transition-map} follows by equating the two formulas for $x$. Since $\tau$ is quadratic with nonzero quadratic coefficient, composition with a degree-$d$ polynomial has degree $2d$. The inverse contains a square root. For example, the C2 polynomial $P(y_2)=y_2$ becomes $1-\sqrt{2(1-y_1)}$, whose derivative diverges as $y_1\to1^-$. The degree identity also shows that $\deg(P\circ\tau)\leqslant d$ exactly when
$\deg P\leqslant\lfloor d/2\rfloor$, which proves
Eq.~\eqref{eq:exact-pullback-common-core}. For the reverse statement, we put $z=1-y_2$ and split $Q(1-z)$ into its even and odd powers of $z$. Under the inverse map, $z=\sqrt{2(1-y_1)}$. Moreover, polynomiality in $y_1$ forces the odd part to vanish. Thus, $Q(1-z)$ is a polynomial in $z^2$, proving Eq.~\eqref{eq:inverse-polynomial-common-core} and the dimension and containment claims.
\end{proof}

If $\Pi_n^{(r)}$ denotes projection onto the degree-$(n-1)$ trial space in representation $r$, Proposition~\ref{prop:pullback-truncation-noncommutation} states generically that $\Pi_n^{(2)}\tau^*\neq\tau^*\Pi_n^{(1)}$. The observed map dependence is therefore a failure of finite approximation operations to commute, not a change in the Schwarzschild continuum spectrum. Last but not least, compactification charts and their coordinate changes form a groupoid, and pullback acts contravariantly on the ambient function spaces. The assignment that replaces each ambient space by the fixed-degree polynomial subspace $\mathbb P_{n-1}$ is not a subfunctor. More precisely, Eq.~\eqref{eq:c1-to-c2-polynomial-pullback}
shows that pullback leaves the chosen trial space. Equivalently, the family of truncations $\Pi_n^{(r)}$ is not a natural transformation with respect to the nonlinear chart change. This categorical statement is precisely the finite-dimensional obstruction described above. It is not an additional physical assumption.

\subsection{Chebyshev representation and matrix assembly}
\label{subsec:chebyshev_assembly}

At resolution $n$, the regularised radial function is approximated by the truncated Chebyshev series
\begin{equation}\label{eq:chebyshev_expansion_numerics}
\Phi_n(y)=\sum_{k=0}^{n-1}a_kT_k(y),\qquad
\mathbf{a}=(a_0,\ldots,a_{n-1})^{\mathsf T}\neq\mathbf{0},
\end{equation}
where $T_k(y)$ denotes the Chebyshev polynomial of the first kind~\cite{Boyd2000, Trefethen2000}. For the endpoint-excluding discretisation we use the $n$ Chebyshev roots
\begin{equation}\label{eq:chebyshev_roots_grid}
y_j^{\mathrm R}=\cos\left(\frac{(2j-1)\pi}{2n}\right),\qquad j\in\{1,\ldots,n\}.
\end{equation}
The endpoint-inclusive comparison uses the Lobatto points
\begin{equation}\label{eq:chebyshev_lobatto_grid}
y_j^{\mathrm L}=\cos\left(\frac{j\pi}{n-1}\right),\qquad j\in\{0,\ldots,n-1\}.
\end{equation}
The endpoint theorem also explains why an unwanted sector can look
spectrally immaculate. The following exact model calculation quantifies its Chebyshev tail.

\begin{proposition}[Chebyshev tail of an endpoint Frobenius power]
\label{prop:chebyshev-frobenius-tail}
Let $\beta>-1/2$. For $n\geqslant 1$, let $c_n(\beta)$ be the $n$th
Chebyshev coefficient of $f_\beta(y)=(1+y)^\beta$ in the convention
\begin{equation}
c_n(\beta)=\frac{2}{\pi}\int_0^\pi
f_\beta(\cos\theta)\cos(n\theta)\,d\theta.
\end{equation}
Then, we have
\begin{equation}
\label{eq:exact-chebyshev-frobenius-coefficient}
c_n(\beta)=\frac{2^{1-\beta}\Gamma(2\beta+1)}
{\Gamma(\beta+n+1)\Gamma(\beta-n+1)}.
\end{equation}
If $\beta\notin\mathbb N_0$, reflection gives
\begin{equation}
\label{eq:reflected-chebyshev-frobenius-coefficient}
c_n(\beta)=(-1)^{n+1}
\frac{2^{1-\beta}\Gamma(2\beta+1)\sin(\pi\beta)}{\pi}
\frac{\Gamma(n-\beta)}{\Gamma(n+\beta+1)},
\end{equation}
and hence,
\begin{equation}
\label{eq:chebyshev-frobenius-tail-asymptotic}
c_n(\beta)=(-1)^{n+1}
\frac{2^{1-\beta}\Gamma(2\beta+1)\sin(\pi\beta)}{\pi}
n^{-2\beta-1}\bigl[1+O(n^{-1})\bigr].
\end{equation}
For $m\in\mathbb N_0$ and $n>m$, let $d_{n,m}$ denote the $n$th Chebyshev of $(1+y)^m\log(1+y)$ in the same convention, namely
\begin{equation}
\label{eq:chebyshev-power-log-coefficient-definition}
d_{n,m}:=\frac{2}{\pi}\int_0^\pi
(1+\cos\theta)^m\log(1+\cos\theta)\cos(n\theta)\,d\theta,
\qquad n\geqslant 1.
\end{equation}
Then, for $n>m$
\begin{equation}
\label{eq:exact-chebyshev-power-log-coefficient}
d_{n,m}=(-1)^{n+m+1}2^{1-m}(2m)!
\frac{\Gamma(n-m)}{\Gamma(n+m+1)},
\end{equation}
Hence, $d_{n,m}=O\bigl(n^{-2m-1}\bigr)$ as $n\to\infty$.
\end{proposition}

\begin{proof}
Let $u=\theta/2$. Since $1+\cos\theta=2\cos^2u$, the standard
beta--Fourier integral, which converges under the hypothesis
$\beta>-1/2$,
\begin{equation}
\int_0^{\pi/2}\cos^{2\beta}u\cos(2nu)\,du=\frac{\pi\Gamma(2\beta+1)}
{2^{2\beta+1}\Gamma(\beta+n+1)\Gamma(\beta-n+1)}
\end{equation}
gives
\begin{equation}
\int_0^\pi(1+\cos\theta)^\beta\cos(n\theta)\,d\theta=\frac{\pi2^{-\beta}\Gamma(2\beta+1)}
{\Gamma(\beta+n+1)\Gamma(\beta-n+1)},
\end{equation}
which proves Eq.~\eqref{eq:exact-chebyshev-frobenius-coefficient}.
Euler's reflection formula gives
\begin{equation}
\frac{1}{\Gamma(\beta-n+1)}=(-1)^{n+1}\frac{\sin(\pi\beta)}{\pi}\Gamma(n-\beta),
\end{equation}
which proves Eq.~\eqref{eq:reflected-chebyshev-frobenius-coefficient}. The standard gamma-ratio asymptotic $\Gamma(n-\beta)/\Gamma(n+\beta+1)
=n^{-2\beta-1}(1+O(n^{-1}))$ proves Eq.~\eqref{eq:chebyshev-frobenius-tail-asymptotic}. Finally, $(1+y)^m\log(1+y)=\partial_\beta(1+y)^\beta|_{\beta=m}$. For $n>m$, the gamma ratio in Eq.~\eqref{eq:reflected-chebyshev-frobenius-coefficient} is finite at $\beta=m$, while $\sin(\pi m)=0$. Thus, only the derivative of the sine factor survives there, and differentiation yields Eq.~\eqref{eq:exact-chebyshev-power-log-coefficient}.
\end{proof}

\begin{corollary}[Uniform differentiation rate for the endpoint model]
\label{cor:chebyshev-endpoint-model-derivative-rate}
Let $S_N f$ denote the degree-$N$ Chebyshev truncation of $f$. If
$\beta\notin\mathbb N_0$ and $\beta>k\geqslant0$, with $k$ an integer, then
\begin{equation}
\label{eq:frobenius-model-derivative-rate}
\left\|f_\beta^{(k)}-(S_Nf_\beta)^{(k)}\right\|_\infty
=O\left(N^{-2(\beta-k)}\right).
\end{equation}
For $g_m(y)=(1+y)^m\log(1+y)$ and $0\leqslant k<m$, we find
\begin{equation}
\label{eq:power-log-model-derivative-rate}
\left\|g_m^{(k)}-(S_Ng_m)^{(k)}\right\|_\infty
=O\left(N^{-2(m-k)}\right).
\end{equation}
\end{corollary}

\begin{proof}
For fixed $k$, the Chebyshev derivative bound $\|T_n^{(k)}\|_\infty=O(n^{2k})$ combines with Eqs.~\eqref{eq:chebyshev-frobenius-tail-asymptotic} and
\eqref{eq:exact-chebyshev-power-log-coefficient}. Summing the differentiated tails gives respectively
$\sum_{n>N}O(n^{-2\beta-1+2k})$ and
$\sum_{n>N}O(n^{-2m-1+2k})$. The corresponding derivative series converge absolutely and uniformly for every order $0\leqslant j\leqslant k$, because $\beta>k$ and $m>k$. The standard termwise-differentiation theorem therefore identifies their sums with $f_\beta^{(j)}$ and $g_m^{(j)}$; summing the two $k$th-derivative tails yields the stated powers.
\end{proof}

For the nonresonant unwanted horizon sector, $\beta=4\alpha$, so the pure endpoint model has coefficients $c_n=O(n^{-8\alpha-1})$, uniform truncation error $O(N^{-8\alpha})$, and $k$th-derivative error $O(N^{-8\alpha+2k})$ for fixed $k<4\alpha$. Across
$\alpha\simeq15$--$32$, the coefficient scale ranges from roughly
$n^{-121}$ to $n^{-257}$. The unwanted sector can therefore be
indistinguishable from an analytic function over a finite coefficient window. A smooth nonvanishing amplitude preserves this leading endpoint rate provided that no competing interior or endpoint singularity produces a slower tail, but Proposition~\ref{prop:chebyshev-frobenius-tail} is a model-tail theorem, not a convergence theorem for the full pencil. Notice that the two maps can be written together as $x_a(y)=[2/(1-y)]^a$, with $a=1$ for C1 and $a=2$ for C2. This notation turns the observed quadratic-versus-quartic node reach into an exact result.

\begin{theorem}[Exact radial reach of the Chebyshev grids]
\label{thm:exact-chebyshev-radial-reach}
For the $n$ roots in Eq.~\eqref{eq:chebyshev_roots_grid}, the sampled radii are
\begin{equation}
\label{eq:exact-root-radii}
x_{j,n}^{(a),\mathrm R}
=\csc^{2a}\left(\frac{(2j-1)\pi}{4n}\right),
\qquad a\in\{1,2\}.
\end{equation}
In particular, we have
\begin{equation}
\label{eq:exact-root-extreme-radii}
x_{\max,n}^{(a),\mathrm R}
=\csc^{2a}\left(\frac{\pi}{4n}\right),\qquad
x_{\min,n}^{(a),\mathrm R}=\sec^{2a}\left(\frac{\pi}{4n}\right).
\end{equation}
The far-end expansions are
\begin{align}
x_{\max,n}^{(1),\mathrm R}
&=\frac{16}{\pi^2}n^2+\frac13
+\frac{\pi^2}{240n^2}+O(n^{-4}),
\label{eq:c1-root-far-reach}\\
x_{\max,n}^{(2),\mathrm R}
&=\frac{256}{\pi^4}n^4
+\frac{32}{3\pi^2}n^2+\frac{11}{45}+O(n^{-2}),
\label{eq:c2-root-far-reach}
\end{align}
while the near-horizon distances satisfy
\begin{align}
x_{\min,n}^{(1),\mathrm R}-1
&=\frac{\pi^2}{16n^2}
+\frac{\pi^4}{384n^4}+O(n^{-6}),
\label{eq:c1-root-horizon-reach}\\
x_{\min,n}^{(2),\mathrm R}-1
&=\frac{\pi^2}{8n^2}
+\frac{7\pi^4}{768n^4}+O(n^{-6}).
\label{eq:c2-root-horizon-reach}
\end{align}
For the Lobatto grid, $j=0$ represents infinity and $j=n-1$ the horizon. The finite radii are exactly
\begin{equation}
\label{eq:exact-lobatto-radii}
x_{j,n}^{(a),\mathrm L}=\csc^{2a}\left(\frac{j\pi}{2(n-1)}\right),
\qquad 1\leqslant j\leqslant n-1.
\end{equation}
The largest finite interior radii consequently obey
\begin{align}
x_{\max,n}^{(1),\mathrm L,\mathrm{int}}
&=\frac{4}{\pi^2}(n-1)^2+\frac13
+\frac{\pi^2}{60(n-1)^2}+O((n-1)^{-4}),
\label{eq:c1-lobatto-far-reach}\\
x_{\max,n}^{(2),\mathrm L,\mathrm{int}}
&=\frac{16}{\pi^4}(n-1)^4
+\frac{8}{3\pi^2}(n-1)^2+\frac{11}{45}
+O((n-1)^{-2}).
\label{eq:c2-lobatto-far-reach}
\end{align}
At every common computational coordinate, $x_2(y)=x_1(y)^2$.
\end{theorem}

\begin{proof}
For a roots point, we set $\theta_j=(2j-1)\pi/(4n)$. The half-angle identity gives $1-y_j^{\mathrm R}=2\sin^2\theta_j$, proving
Eq.~\eqref{eq:exact-root-radii}. The extreme values correspond to $j=1$
and $j=n$. For a Lobatto point, the same identity with
$\theta_j=j\pi/[2(n-1)]$ proves Eq.~\eqref{eq:exact-lobatto-radii}, and the largest finite interior radius corresponds to $j=1$. The expansions follow after substituting $z=\pi/(4n)$ for the roots grid and $z=\pi/[2(n-1)]$ for the Lobatto grid into
\begin{equation}
\csc^2z=z^{-2}+\frac13+\frac{z^2}{15}+O(z^4),\qquad
\sec^2z=1+z^2+\frac23z^4+O(z^6),
\end{equation}
and from squaring these series.
\end{proof}

\begin{corollary}[Reach exponent for algebraic compactifications]
\label{cor:general-compactification-reach}
Let $x:[-1,1)\to[1,+\infty)$ be a continuous, monotonically increasing
compactification map. If $x(y)\sim C(1-y)^{-p}$ as $y\to1^-$, with
$C,p>0$, then the roots grid has
\begin{equation}\label{eq:general-root-reach-law}
x_{\max,n}^{\mathrm R}
\sim C\left(\frac{8}{\pi^2}\right)^p n^{2p}.
\end{equation}
\end{corollary}

\begin{proof}
The nearest root satisfies $1-y_1^{\mathrm R}=2\sin^2(\pi/(4n))\sim\pi^2/(8n^2)$. Monotonicity of $x$ makes this nearest-to-$1$ root the one with maximal sampled radius. Substituting
the preceding asymptotic into $x(y)\sim C(1-y)^{-p}$ gives
Eq.~\eqref{eq:general-root-reach-law}.
\end{proof}

Theorem~\ref{thm:exact-chebyshev-radial-reach} proves physical node
placement. It explains the geometric source of the $n^2$ and $n^4$ radial scales, but it does not derive the empirical spectral-edge powers of the nonnormal quadratic pencils. That implication requires a separate operator-level asymptotic estimate.

On the roots grid every row is obtained by collocating the differential
equation in the open interval. The endpoint relations derived in the
preceding section are then evaluated as independent a posteriori diagnostic
checks. In the Lobatto realisation, the two rows associated with $y=\pm 1$ are replaced by the corresponding limiting equations, that is Eqs.~\eqref{FirstCompactificationInfinityRelation} and \eqref{FirstCompactificationHorizonRelation} for $x=2/(1-y)$, and Eqs.~\eqref{SecondCompactificationInfinityRelation} and \eqref{SecondCompactificationHorizonRelation} for $x=4/(1-y)^2$, respectively. The endpoint relations are therefore not appended as extra constraints, and the matrix pencil remains square. For either compactification, we write
\begin{equation}\label{eq:operator_component_form}
\mathcal{L}_q[\Phi]=c_{q0}(y)\Phi+c_{q1}(y)\Phi'+c_{q2}(y)\Phi'',
\qquad q\in\{0,1,2\},
\end{equation} 
with the coefficient functions $c_{qr}(y)$ specified by Eqs.~\eqref{L0none}--\eqref{L2none}, or by Eqs.~\eqref{L0noneII}--\eqref{L2noneII}. At an interior collocation point $y_j$, the entries of the three coefficient matrices are
\begin{equation}\label{eq:matrix_assembly_entries}
(M_q^{\mathrm{raw}})_{j,k+1}=c_{q0}(y_j)T_k(y_j)+c_{q1}(y_j)T_k'(y_j)
+c_{q2}(y_j)T_k''(y_j)
\end{equation}
with $q\in\{0,1,2\}$ and $k\in\{0,1,\dots,n-1\}$. The basis functions and their first two derivatives are evaluated pointwise, without expanding $T_k$ into monomials and without extracting symbolic polynomial coefficients. Starting from
$T_0(y)=1$, $T_1(y)=y$, $T_0'(y)=0$, $T_1'(y)=1$, and $T_0''(y)=T_1''(y)=0$, the assembly uses the coupled recurrences
\begin{subequations}\label{eq:chebyshev_derivative_recurrences}
\begin{align}
T_{k+1}(y)&=2yT_k(y)-T_{k-1}(y),\label{eq:chebyshev_recurrence_T}\\
T_{k+1}'(y)&=2T_k(y)+2yT_k'(y)-T_{k-1}'(y),\label{eq:chebyshev_recurrence_DT}\\
T_{k+1}''(y)&=4T_k'(y)+2yT_k''(y)-T_{k-1}''(y).\label{eq:chebyshev_recurrence_DDT}
\end{align}
\end{subequations}
This recurrence-based construction reduces the amount of symbolic algebra
and permits the matrices to be assembled directly at the requested decimal
precision. Before export, the recurrence values of $T_k(y)$, $T_k'(y)$, and
$T_k''(y)$ are cross-checked against direct basis evaluation, and the maximum
relative discrepancies are recorded in the assembly metadata. Collocation of Eq.~\eqref{eq:operator_polynomial_numerics} gives the quadratic matrix polynomial
\begin{equation} \label{eq:quadratic_matrix_pencil}
Q_n^{\mathrm{raw}}(\Omega)\mathbf{a}\equiv
\left(M_0^{\mathrm{raw}}+i\Omega M_1^{\mathrm{raw}}
+\Omega^2M_2^{\mathrm{raw}}\right)\mathbf{a}=\mathbf{0},\qquad M_q^{\mathrm{raw}}\in\mathbb{R}^{n\times n}.
\end{equation}
Equation~\eqref{eq:quadratic_matrix_pencil} defines the \emph{raw} pencil. Matrix assembly and eigensolution are both performed in
arbitrary precision. The matrix entries are exported as decimal text and are
parsed directly into the arbitrary-precision type used by the eigensolver,
without an intermediate binary64 conversion. The decimal precision used to
export the matrices and the working precision used by the eigensolver are
recorded separately. The latter includes guard digits and is increased
whenever the observed conditioning or precision tests require it. Extra
working digits can reduce arithmetic error in the eigensolve, but cannot
restore information that was not retained during matrix assembly. Consequently, no fixed machine-precision assumption enters the construction.

\subsection{Frequency-independent row equilibration} \label{subsec:row_equilibration}

The three raw coefficient matrices can have strongly different row scales. We therefore form the horizontal block matrix
\begin{equation}\label{eq:block_matrix_for_equilibration}
\mathscr{M}^{\mathrm{raw}}=\begin{bmatrix}
M_0^{\mathrm{raw}}&M_1^{\mathrm{raw}}&M_2^{\mathrm{raw}}
\end{bmatrix}
\end{equation}
and define
\begin{equation}\label{eq:row_scaling_definition}
r_j=\left[\sum_{k=1}^{n}\left(|(M_0^{\mathrm{raw}})_{jk}|^2+|(M_1^{\mathrm{raw}})_{jk}|^2+|(M_2^{\mathrm{raw}})_{jk}|^2\right)
\right]^{1/2},\qquad
d_j=r_j^{-1}.
\end{equation}
With $D=\operatorname{diag}(d_1,\ldots,d_n)$, the equilibrated matrices and
pencil are
\begin{equation}\label{eq:equilibrated_pencil}
M_q^{\mathrm{eq}}=DM_q^{\mathrm{raw}},\qquad
Q_n^{\mathrm{eq}}(\Omega)=DQ_n^{\mathrm{raw}}(\Omega)=M_0^{\mathrm{eq}}+i\Omega M_1^{\mathrm{eq}}+\Omega^2M_2^{\mathrm{eq}}.
\end{equation}
The same frequency-independent left scaling is applied to all three
coefficient matrices, and no column scaling is used. Provided no row of
$\mathscr{M}^{\mathrm{raw}}$ vanishes, $D$ is nonsingular, so the raw and equilibrated polynomials are strictly equivalent in exact arithmetic and have the same finite eigenvalues. Row equilibration is therefore a numerical
preconditioning step, not a modification of the spectral problem. It does
not remove intrinsic non-normality or eigenvalue ill-conditioning. The equilibrated pencil is used for the primary eigensolution. The raw pencil is retained as an independent representation check and is solved in selected runs for which the additional cost is acceptable. Besides the raw/equilibrated eigenvalue comparison, the implementation verifies the assembled identities $M_q^{\mathrm{eq}}=DM_q^{\mathrm{raw}}$ and monitors the row-norm dynamic
range before and after scaling.

\subsection{Companion linearisation and generalized Schur solution}
\label{subsec:qep_linearisation}

The following construction applies to either representation. In the primary
solve, $M_q=M_q^{\mathrm{eq}}$ and $Q_n=Q_n^{\mathrm{eq}}$. In a raw-representation solve, $M_q=M_q^{\mathrm{raw}}$ and $Q_n=Q_n^{\mathrm{raw}}$. The quadratic eigenvalue problem is converted into a $2n$-dimensional generalized eigenvalue problem using the first companion linearisation~\cite{Tisseur2001}
\begin{equation}\label{eq:qep_linearisation}
\mathbb{A}\mathbf{v}=\Omega\mathbb{B}\mathbf{v},\qquad
\mathbb{A}=
\begin{pmatrix}
0&I\\
-M_0&-iM_1
\end{pmatrix},
\qquad
\mathbb{B}
=
\begin{pmatrix}
I&0\\
0&M_2
\end{pmatrix},
\end{equation}
where
\begin{equation}\label{eq:linearised_right_vector}
\mathbf{v}=
\begin{pmatrix}
\mathbf{a}\\
\Omega\mathbf{a}
\end{pmatrix}.
\end{equation}
Equation~\eqref{eq:qep_linearisation} avoids any inversion of $M_2$. Hence, a singular leading coefficient and the associated infinite generalized eigenvalues can be handled within the pencil rather than removed by an ad hoc regularisation. The generalized spectrum and right eigenvectors are computed in arbitrary precision with a generalized Schur, or QZ, decomposition~\cite{MolerStewart1973}. A generalized left eigenvector $\mathbf{w}$ satisfies
\begin{equation}\label{eq:adjoint_generalized_pencil}
\mathbb{A}^{\dagger}\mathbf{w}=\overline{\Omega}\,\mathbb{B}^{\dagger}\mathbf{w}.
\end{equation}
In the implementation, these vectors are obtained by solving the adjoint
generalized pencil and matching its finite eigenvalues one-to-one to the
complex conjugates of the eigenvalues of the original pencil. If
$\mathbf{w}=(\mathbf{x}^{\mathsf T},\mathbf{b}^{\mathsf T})^{\mathsf T}$,
then its lower block satisfies
\begin{equation}\label{eq:left_polynomial_vector}
Q_n(\Omega)^{\dagger}\mathbf{b}=\mathbf{0}
\end{equation}
and is therefore the left polynomial eigenvector used below. For the
equilibrated pencil, $\mathbf{b}_{\mathrm{eq}}$ is obtained in this way,
whereas the corresponding raw left vector is
\begin{equation}\label{eq:left_vector_scaling}
\mathbf{b}_{\mathrm{raw}}=D\mathbf{b}_{\mathrm{eq}},
\end{equation}
up to normalisation. Right and left polynomial vectors are normalised to
unit Euclidean norm. Their otherwise arbitrary phase is fixed by making a
largest-modulus component real and positive. The eigenvalue mismatch in
the original/adjoint pairing is retained as a diagnostic. The pairing is
a one-to-one proximity assignment and can become ambiguous inside a tight
cluster. Condition numbers are therefore not trusted when the pairing
error or the corresponding left residual is inaccurate.

\begin{proposition}[Exact equivalence of the finite representations]
\label{prop:qep-representation-equivalence}
Let $Q(\Omega)=M_0+i\Omega M_1+\Omega^2M_2$ be regular. If $R$ is
nonsingular and independent of $\Omega$, then $Q$ and $RQ$ have the same finite eigenvalues, algebraic multiplicities, and right nullspaces. Their left eigenvectors are related by
$b_{\mathrm{raw}}=R^{\dagger}b_{\mathrm{eq}}$. For the companion matrices in Eq.~\eqref{eq:qep_linearisation},
\begin{equation}\label{eq:companion-determinant-identity}
\det(\mathbb A-\Omega\mathbb B)=\det Q(\Omega).
\end{equation}
Moreover, $Q(\Omega)a=0$ for finite $\Omega$ if and only if
$(a^{\mathsf T},\Omega a^{\mathsf T})^{\mathsf T}$ is a generalized
eigenvector of $(\mathbb A,\mathbb B)$.
\end{proposition}

\begin{proof}
The identity $\det(RQ)=\det R\det Q$ shows that the determinant changes only by a nonzero constant, so all finite zeros and their algebraic multiplicities are preserved. Moreover, $RQ(\Omega)a=0$ if and only if $Q(\Omega)a=0$. If $b_{\mathrm{eq}}^\dagger RQ(\Omega)=0$, then $(R^\dagger b_{\mathrm{eq}})^\dagger Q(\Omega)=0$, which gives the stated left-vector relation. For $\Omega\neq 0$, the Schur complement of the upper-left block of $\mathbb A-\Omega\mathbb B$ gives
\begin{eqnarray}
\det(\mathbb A-\Omega\mathbb B)&=&\det(-\Omega I)
\det\left(-iM_1-\Omega M_2-\Omega^{-1}M_0\right),\notag\\
&=&\det(-\Omega I)\det\left(-\Omega^{-1}Q(\Omega)\right)
=\det Q(\Omega).
\end{eqnarray}
Both sides are polynomials in $\Omega$, so the identity also holds at
$\Omega=0$. Finally, writing a companion vector as $(u^{\mathsf T},
v^{\mathsf T})^{\mathsf T}$, its first block row gives $v=\Omega u$. After this substitution, the second block row is $-Q(\Omega)u=0$. This proves both directions of the eigenvector equivalence.
\end{proof}

\subsection{Residuals, backward errors, pseudospectra, and conditioning}
\label{subsec:residuals_conditioning}

Following the normwise polynomial-eigenvalue framework of~\cite{Tisseur2000, Tisseur2001}, for a computed polynomial eigenpair
$(\Omega,\mathbf{a})$ we define
\begin{equation}\label{eq:polynomial_residual}
\mathbf{r}=Q_n(\Omega)\mathbf{a},\qquad
\mathbf{r}_{\mathrm L}=Q_n(\Omega)^{\dagger}\mathbf{b}.
\end{equation}
In the following, the subscript $\star$ selects the matrix norm. For
$\star=2$, it denotes the spectral, or induced operator 2-norm,
$\|M\|_2:=\sigma_{\rm max}(M)=\sqrt{\lambda_{\rm max}(M^\dagger M)}$.
For $\star=F$, it denotes the Frobenius norm
$\|M\|_F:=\sqrt{\sum_{j,k}|M_{jk}|^2}$. In either case, the right scaled
residual is
\begin{equation}\label{eq:right_backward_error}
\eta_{\star}^{\mathrm R}(\Omega)=\frac{\|Q_n(\Omega)\mathbf{a}\|_2}
{\left(\|M_0\|_{\star}+|\Omega|\|M_1\|_{\star}+|\Omega|^2\|M_2\|_{\star}
\right)\|\mathbf{a}\|_2}.
\end{equation}
The corresponding left quantity is
\begin{equation}\label{eq:left_backward_error}
\eta_{\star}^{\mathrm L}(\Omega)=\frac{\|Q_n(\Omega)^{\dagger}\mathbf{b}\|_2}
{\left(\|M_0\|_{\star}+|\Omega|\|M_1\|_{\star}+|\Omega|^2\|M_2\|_{\star}
\right)\|\mathbf{b}\|_2}.
\end{equation}

\begin{proposition}[Exact normwise backward error]
\label{prop:exact-qep-backward-error}
Let $a\neq0$ and $b\neq0$, and let
\begin{equation}
s_\star(\Omega)=\|M_0\|_\star+|\Omega|\|M_1\|_\star
+|\Omega|^2\|M_2\|_\star>0,
\end{equation}
and consider complex, unstructured coefficient perturbations
\begin{equation}
\Delta Q(\Omega)=\Delta M_0+i\Omega\Delta M_1
+\Omega^2\Delta M_2,\qquad
\|\Delta M_q\|_\star\leqslant\varepsilon\|M_q\|_\star.
\end{equation}
For either $\star=2$ or $\star=F$, the smallest $\varepsilon$ for which
$[Q_n(\Omega)+\Delta Q(\Omega)]a=0$ is exactly
$\eta_\star^{\mathrm R}(\Omega)$. The analogous minimum subject to
$[Q_n(\Omega)+\Delta Q(\Omega)]^\dagger b=0$ is
$\eta_\star^{\mathrm L}(\Omega)$.
\end{proposition}

\begin{proof}
Let $\phi_0=1$, $\phi_1=i\Omega$, and $\phi_2=\Omega^2$. Any admissible
perturbation that cancels the right residual $r=Q_n(\Omega)a$ satisfies
$\|r\|_2=\|\Delta Q(\Omega)a\|_2\leqslant\varepsilon s_\star(\Omega)\|a\|_2$, because $\|\Delta M_q\|_2\leqslant\|\Delta M_q\|_\star$ for both allowed norms. This gives the lower bound
$\varepsilon\geqslant\eta_\star^{\mathrm R}$. For attainability, let us set $E=-ra^\dagger/\|a\|_2^2$ and, whenever $\phi_q\neq0$, we choose
\begin{equation}
\Delta M_q=\frac{\overline{\phi_q}}{|\phi_q|}
\frac{\|M_q\|_\star}{s_\star(\Omega)}E.
\end{equation}
Furthermore, we set $\Delta M_q=0$ when $\phi_q=0$. Then,
\begin{equation}
\Delta Q(\Omega)=\frac{\sum_q|\phi_q|\|M_q\|_\star}{s_\star(\Omega)}E=E,
\end{equation}
so $\Delta Q(\Omega)a=Ea=-r$. Since $E$ has rank one,
$\|E\|_2=\|E\|_F=\|r\|_2/\|a\|_2$, and each coefficient perturbation obeys $\|\Delta M_q\|_\star\leqslant\eta_\star^{\mathrm R}\|M_q\|_\star$. Equality holds for the nonzero phase factors. Thus, the lower bound is attained. Repeating the argument for
$Q_n(\Omega)^\dagger=\sum_q\overline{\phi_q}M_q^\dagger$, whose coefficient norms are unchanged, proves the left statement.
\end{proof}

\begin{corollary}[Exact finite-pencil pseudospectrum]
\label{cor:exact-qep-pseudospectrum}
Let $Q_n$ be a regular $n\times n$ matrix polynomial. For
$\varepsilon\geqslant 0$ and $\star\in\{2,F\}$, define the
coefficientwise-relative normwise pseudospectrum
\begin{equation}
\begin{aligned}
\Lambda_{\varepsilon,\star}(Q_n)
=\bigl\{z\in\mathbb C:&\ \text{there exist complex, unstructured }
\Delta M_0,\Delta M_1,\Delta M_2\text{ such that}\\
&\det[Q_n(z)+\Delta M_0+iz\Delta M_1+z^2\Delta M_2]=0,\\
&\|\Delta M_q\|_\star\leqslant
\varepsilon\|M_q\|_\star\ \text{for }q=0,1,2\bigr\}.
\end{aligned}
\label{eq:qep-pseudospectrum}
\end{equation}
If $s_\star(z)>0$, then
\begin{equation}
\label{eq:qep-pseudospectrum-singular-value}
z\in\Lambda_{\varepsilon,\star}(Q_n)
\quad\Longleftrightarrow\quad
\sigma_{\min}(Q_n(z))\leqslant\varepsilon s_\star(z).
\end{equation}
Consequently, the least admissible relative coefficient perturbation that makes the perturbed evaluation singular at $z$ is exactly
\begin{equation}
\frac{\sigma_{\min}(Q_n(z))}{s_\star(z)}.
\end{equation}
\end{corollary}

\begin{proof}
Suppose first that $z\in\Lambda_{\varepsilon,\star}(Q_n)$. Choose a nonzero vector $a$ in the nullspace of the perturbed matrix in
Eq.~\eqref{eq:qep-pseudospectrum}. Proposition~\ref{prop:exact-qep-backward-error} then gives
\begin{equation}
\frac{\|Q_n(z)a\|_2}{s_\star(z)\|a\|_2}\leqslant\varepsilon.
\end{equation}
Minimising over $a\neq 0$ and using the variational characterisation of the smallest singular value proves the forward implication in
Eq.~\eqref{eq:qep-pseudospectrum-singular-value}. Conversely, let $a$ be a unit right singular vector of $Q_n(z)$ associated
with $\sigma_{\min}(Q_n(z))$. If $\sigma_{\min}(Q_n(z))\leqslant\varepsilon s_\star(z)$, the rank-one
construction in Proposition~\ref{prop:exact-qep-backward-error} supplies coefficient perturbations obeying the bounds in
Eq.~\eqref{eq:qep-pseudospectrum} and makes the perturbed evaluation singular. This proves the reverse implication. The same construction is attainable in both the spectral and Frobenius norms because its perturbation has rank one. The displayed minimum follows at once.
\end{proof}

This corollary gives an exact pseudospectrum for the stated finite polynomial pencil and perturbation model. Although the raw and row-equilibrated pencils are strictly equivalent, this pseudospectrum need not be invariant because row equilibration generally changes the coefficient norms and hence, the admissible perturbation set. It is neither a continuum pseudospectrum nor a test for a zero of the physical Jost determinant.

Both the right and left backward errors are evaluated for the raw and
equilibrated pencils. They measure the quality of a computed eigenpair
relative to the sizes of the coefficient matrices. They do not measure the discretisation error of the underlying differential problem. The
generalized-eigenvalue computation is checked independently through the
linearisation residuals
\begin{subequations}\label{eq:linearisation_residuals}
\begin{align}
\rho_{\mathrm{lin}}^{\mathrm R}
&=\frac{\|\mathbb{A}\mathbf{v}-\Omega\mathbb{B}\mathbf{v}\|_2}{\left(\|\mathbb{A}\|_F+|\Omega|\|\mathbb{B}\|_F\right)\|\mathbf{v}\|_2},
\label{eq:linearisation_right_residual}\\
\rho_{\mathrm{lin}}^{\mathrm L}
&=\frac{\|\mathbb{A}^{\dagger}\mathbf{w}-\overline{\Omega}\mathbb{B}^{\dagger}\mathbf{w}\|_2}{\left(\|\mathbb{A}\|_F+|\Omega|\|\mathbb{B}\|_F\right)
\|\mathbf{w}\|_2}.\label{eq:linearisation_left_residual}
\end{align}
\end{subequations}
Writing $\mathbf{v}=(\mathbf{v}_1^{\mathsf T},
\mathbf{v}_2^{\mathsf T})^{\mathsf T}$, we also monitor the companion-block
consistency
\begin{equation}\label{eq:block_consistency}
\rho_{\mathrm{block}}=\frac{\|\mathbf{v}_2-\Omega\mathbf{v}_1\|_2}
{\|\mathbf{v}_2\|_2+|\Omega|\|\mathbf{v}_1\|_2}.
\end{equation}
For a simple finite eigenvalue, the derivative of the matrix polynomial is
\begin{equation}\label{eq:qep_derivative}
Q_n'(\Omega)=iM_1+2\Omega M_2.
\end{equation}
The normwise absolute condition number used in the computations is
\begin{equation}\label{eq:absolute_condition_number}
\kappa_{\mathrm{abs},\star}(\Omega)=\frac{\left(\|M_0\|_{\star}
+|\Omega|\|M_1\|_{\star}+|\Omega|^2\|M_2\|_{\star}\right)\|\mathbf{a}\|_2\|\mathbf{b}\|_2}{\left|\mathbf{b}^{\dagger}Q_n'(\Omega)\mathbf{a}\right|},
\qquad \star\in\{2,F\},
\end{equation}
with relative counterpart
\begin{equation}\label{eq:relative_condition_number}
\kappa_{\mathrm{rel},\star}(\Omega)=\frac{\kappa_{\mathrm{abs},\star}(\Omega)}{|\Omega|}.
\end{equation}

\begin{proposition}[Exact first-order normwise condition number]
\label{prop:qep-first-order-condition-number}
Let $\Omega\neq 0$ be a simple finite eigenvalue with right and left vectors $a$ and $b$, so that $b^\dagger Q_n'(\Omega)a\neq0$. Let
$\mathfrak E_\star$ be the set of complex, unstructured coefficient directions $(E_0,E_1,E_2)$ satisfying $\|E_q\|_\star\leqslant\|M_q\|_\star$. For the local
eigenvalue branch of $M_q(\varepsilon)=M_q+\varepsilon E_q$,
\begin{equation}
\sup_{\mathfrak E_\star}|\Omega'(0)|
=\kappa_{\mathrm{abs},\star}(\Omega),\qquad
\sup_{\mathfrak E_\star}\frac{|\Omega'(0)|}{|\Omega|}
=\kappa_{\mathrm{rel},\star}(\Omega).
\label{eq:qep-exact-first-order-condition-number}
\end{equation}
Consequently, every fixed admissible direction obeys
\begin{align}
|\Omega(\varepsilon)-\Omega|
&\leqslant\kappa_{\mathrm{abs},\star}(\Omega)|\varepsilon|
+O(|\varepsilon|^2),\label{eq:qep-absolute-sensitivity-bound}\\
\frac{|\Omega(\varepsilon)-\Omega|}{|\Omega|}
&\leqslant\kappa_{\mathrm{rel},\star}(\Omega)|\varepsilon|
+O(|\varepsilon|^2).\label{eq:qep-relative-sensitivity-bound}
\end{align}
\end{proposition}

\begin{proof}
Simplicity supplies a differentiable local eigenvalue and eigenvector branch. If we differentiate $Q_{n,\varepsilon}(\Omega(\varepsilon))a(\varepsilon)=0$ at $\varepsilon=0$, the resulting term $Q_n(\Omega)a'(0)$ vanishes after left-multiplication by $b^\dagger$, and hence, we have
\begin{equation}
\Omega'(0)=-\frac{b^\dagger(E_0+i\Omega E_1+\Omega^2E_2)a}
{b^\dagger Q_n'(\Omega)a}.
\end{equation}
For each coefficient,
$|b^\dagger E_q a|\leqslant\|b\|_2\|E_q\|_2\|a\|_2
\leqslant\|b\|_2\|M_q\|_\star\|a\|_2$. The triangle inequality therefore gives $|\Omega'(0)|\leqslant\kappa_{\mathrm{abs},\star}(\Omega)$. To attain the bound, we put $\phi_0=1$, $\phi_1=i\Omega$, $\phi_2=\Omega^2$, and
\begin{equation}
U=\frac{ba^\dagger}{\|b\|_2\|a\|_2},\qquad
E_q=\frac{\overline{\phi_q}}{|\phi_q|}\|M_q\|_\star U.
\end{equation}
All three phase quotients are defined because $\Omega\neq0$. The rank-one matrix $U$ has both spectral and Frobenius norm one, while
$b^\dagger Ua=\|b\|_2\|a\|_2$. Thus, all three numerator contributions are phase aligned and the upper bound is attained. Division by $|\Omega|$ proves the relative identity. The local branch is analytic under these simple-root hypotheses, so its Taylor expansion gives the two stated local bounds.
\end{proof}

The proposition is local. In particular, the product
$\kappa_{\mathrm{rel},\star}\eta_\star^{\mathrm R}$ is a first-order consistency indicator because, without an eigenvalue-isolation radius and control of the higher-order term, it is not a nonasymptotic forward-error enclosure. The relative quantity is undefined at $\Omega=0$, which is excluded from the physical candidate set. Equations~\eqref{eq:absolute_condition_number}
and \eqref{eq:relative_condition_number} are scalar simple-eigenvalue
condition numbers. Near a multiple eigenvalue or a tightly clustered group, they must be interpreted with caution. In particular, a small residual can coexist with a very poorly determined eigenvalue when the derivative pairing $|\mathbf{b}^{\dagger}Q_n'(\Omega)\mathbf{a}|$ is small. The raw and equilibrated condition numbers are both retained because left row scaling changes the normwise coefficient-perturbation model even though it leaves the exact eigenvalues unchanged. When the raw generalized spectrum is available, an equilibrated eigenvalue is additionally assigned the representation gap
\begin{equation}\label{eq:raw_equilibrated_gap}
\Delta_{\mathrm{raw/eq}}(\Omega)=\min_{\lambda\in\sigma_{\mathrm{fin}}(Q_n^{\mathrm{raw}})}|\Omega-\lambda|,
\end{equation}
where $\sigma_{\mathrm{fin}}(Q_n^{\mathrm{raw}})$ denotes the finite
raw-pencil spectrum. Since the two pencils are equivalent in exact arithmetic, $\Delta_{\mathrm{raw/eq}}$ measures numerical representation sensitivity, not a second physical boundary-value problem. In a tight cluster, the nearest raw eigenvalue need not carry an unambiguous mode label, so this gap is interpreted together with the eigenvector and conditioning data.

\subsection{Coefficient tails and endpoint checks}
\label{subsec:coefficient_endpoint_diagnostics}

The Chebyshev coefficients provide a direct test of whether the computed
radial function is resolved in the chosen basis. Let
\begin{equation}\label{eq:tail_index_definition}
m_{10}=\max\{5,\lceil0.1n\rceil\},\qquad
m_{20}=\max\{8,\lceil0.2n\rceil\},
\end{equation}
and let $\mathbf{a}^{(m)}$ denote the final $m$ coefficients of $\mathbf{a}$. We record
\begin{subequations}\label{eq:coefficient_tail_metrics}
\begin{align}
t_{m,2}
&=\frac{\|\mathbf{a}^{(m)}\|_2}{\|\mathbf{a}\|_2},\qquad m\in\{m_{10},m_{20}\},
\label{eq:tail_l2_metric}\\
t_{m,\infty}
&=\frac{\|\mathbf{a}^{(m)}\|_{\infty}}{\|\mathbf{a}\|_{\infty}},
\qquad m\in\{m_{10},m_{20}\},\label{eq:tail_infinity_metric}\\
t_{\mathrm{last}}
&=\frac{|a_{n-1}|}{\max_{0\leqslant k<n}|a_k|}.\label{eq:last_coefficient_metric}
\end{align}
\end{subequations}
The fraction of non-increasing adjacent coefficient magnitudes within the
final $20\%$ tail is also retained. Decaying tails support spectral resolution,
whereas a flat or growing tail is a warning. These tests are heuristic:
coefficient decay at finite $n$ is neither a proof of analyticity nor a pole
criterion. Endpoint values are reconstructed analytically from the
coefficients using
\begin{subequations}\label{eq:chebyshev_endpoint_values}
\begin{align}
T_k(1)&=1,
&
T_k(-1)&=(-1)^k,
\label{eq:chebyshev_endpoint_function_values}\\
T_k'(1)&=k^2,
&
T_k'(-1)&=(-1)^{k-1}k^2,\qquad k\geqslant 1.
\label{eq:chebyshev_endpoint_derivative_values}
\end{align}
\end{subequations}
Consequently, we have
\begin{subequations}\label{eq:reconstructed_endpoint_data}
\begin{align}
\Phi_n(1)&=\sum_{k=0}^{n-1}a_k,
&
\Phi_n(-1)&=\sum_{k=0}^{n-1}(-1)^k a_k,
\label{eq:reconstructed_endpoint_values}\\
\Phi_n'(1)&=\sum_{k=1}^{n-1}k^2a_k,
&
\Phi_n'(-1)&=\sum_{k=1}^{n-1}(-1)^{k-1}k^2a_k.
\label{eq:reconstructed_endpoint_derivatives}
\end{align}
\end{subequations}
The reconstructed $\Phi_n(\pm 1)$ and $\Phi_n'(\pm 1)$ are substituted into
the endpoint relations already derived for the chosen compactification. If such a relation is written as
\begin{equation}\label{eq:generic_endpoint_relation}
R_{\pm}(\Omega)=c_{\pm,1}(\Omega)\Phi_n'(\pm1)+c_{\pm,0}(\Omega)\Phi_n(\pm1),
\end{equation}
its scaled residual is
\begin{equation}\label{eq:generic_endpoint_scaled_residual}
\eta_{\pm}^{\mathrm{end}}=\frac{|R_{\pm}(\Omega)|}
{|c_{\pm,1}(\Omega)|\,|\Phi_n'(\pm1)|+|c_{\pm,0}(\Omega)|\,|\Phi_n(\pm1)|},
\end{equation}
with a small arithmetic floor used only to avoid division by zero. For the
derivative-only infinity relation of the quadratic compactification, the
normalising scale is taken to be $\sum_{k=1}^{n-1}k^2|a_k|$. On a roots
grid the endpoint residuals are independent a posteriori tests because the
endpoints were not collocated. On a Lobatto grid they mainly verify the
correct enforcement of the two endpoint rows and therefore do not
constitute independent evidence of convergence.

\subsection{Cross-resolution matching and mode association}
\label{subsec:cross_resolution_matching}

\begin{proposition}[Reflection symmetry of the finite pencil]
\label{prop:finite-pencil-reflection-symmetry}
If $Q_n$ is regular and $M_0,M_1,M_2$ are real, then
\begin{equation}
\label{eq:finite-pencil-reflection-identity}
Q_n(-\overline\Omega)=\overline{Q_n(\Omega)}.
\end{equation}
Thus every finite eigenvalue $\Omega$ with right eigenvector $a$ has the
partner $-\overline\Omega$ with eigenvector $\overline a$ and the same
algebraic multiplicity.
\end{proposition}

\begin{proof}
Equation~\eqref{eq:finite-pencil-reflection-identity} follows by conjugating
$M_0+i\Omega M_1+\Omega^2M_2$. If $Q_n(\Omega)a=0$, conjugation gives
$Q_n(-\overline\Omega)\overline a=0$. Taking determinants in
Eq.~\eqref{eq:finite-pencil-reflection-identity} gives
$\det Q_n(-\overline\Omega)=\overline{\det Q_n(\Omega)}$; the reflected
polynomial zeros therefore have the same orders, which proves the algebraic
multiplicity statement.
\end{proof}

The generalized pencil contains finite and infinite eigenvalues, zero or
near-zero roots, symmetry partners, and roots outside the frequency region
of interest. Before cross-resolution comparison, we retain finite damped
roots in a prescribed analysis window. By
Proposition~\ref{prop:finite-pencil-reflection-symmetry}, the spectrum is
symmetric under $\Omega\mapsto-\overline{\Omega}$. When convenient, only
the representative branch with non-negative real part is tabulated. A root
is classified numerically as lying on the negative imaginary axis only when
\begin{equation}\label{eq:nia_numerical_classification}
\operatorname{Im}\Omega<-\tau_{\mathrm I},\qquad
|\operatorname{Re}\Omega|\leqslant\tau_{\mathrm{NIA}},
\end{equation}
where the tolerances are numerical classification parameters, not physical
widths. Their values, together with the zero-mode and analysis-window
thresholds, are stored in the run metadata. Modes are associated across three resolutions $n_1<n_2<n_3$. The spectrum at the highest resolution is used as the anchor. Anchors are traversed in increasing lexicographic order of $(|\operatorname{Im}\Omega_3|,\operatorname{Re}\Omega_3,|\Omega_3|)$. For each anchor $\Omega_3$, the nearest unused roots $\Omega_2$ and $\Omega_1$ are selected from the two lower-resolution candidate sets. This is a greedy one-to-one association rather than a global assignment, so close clusters can make individual labels non-unique. We then form
\begin{equation}\label{eq:pairwise_resolution_drifts}
d_{12}=|\Omega_1-\Omega_2|,\qquad
d_{23}=|\Omega_2-\Omega_3|,\qquad
d_{13}=|\Omega_1-\Omega_3|,
\end{equation}
and accept the association only if
\begin{equation}\label{eq:triple_match_criterion}
\Delta_{123}\equiv\max\{d_{12},d_{23},d_{13}\}\leqslant\tau_{123},
\end{equation}
where
\begin{equation}\label{eq:triple_match_tolerance}
\tau_{123}=\max\left\{\tau_{\mathrm{abs}},\tau_{\mathrm{rel}}
\max\left(1,|\Omega_1|,|\Omega_2|,|\Omega_3|\right)
\right\}.
\end{equation}
The representative frequency of an accepted triple is $\Omega_3$, the
highest-resolution value. No arithmetic average is taken. We also report
\begin{equation}\label{eq:relative_resolution_drift}
\delta_{123}=\frac{\Delta_{123}}{\max(1,|\Omega_3|)}
\end{equation}
and, for compact presentation only,
\begin{equation}\label{eq:estimated_converged_digits}
N_{\mathrm{drift}}=-\log_{10}\delta_{123}.
\end{equation}
The latter is an empirical drift indicator, not a rigorous error bound. Frequency proximity is supplemented by a comparison of the right
Chebyshev-coefficient vectors. For a pair of resolutions $n_r$ and $n_s$,
let $m=\min(n_r,n_s)$. We truncate both vectors to their first $m$ entries, and
normalise them to unit Euclidean norm. With
\begin{equation}\label{eq:eigenvector_overlap}
\mathcal{O}_{rs}=\left|(\mathbf{a}^{(r)}_{1:m})^{\dagger}\mathbf{a}^{(s)}_{1:m}
\right|,
\end{equation}
the optimally phase-aligned difference is
\begin{equation}\label{eq:eigenvector_aligned_difference}
\delta^{\mathrm{vec}}_{rs}=\left\|\mathbf{a}^{(r)}_{1:m}-e^{-i\arg c_{rs}}\mathbf{a}^{(s)}_{1:m}\right\|_2,\qquad
c_{rs}=(\mathbf{a}^{(r)}_{1:m})^{\dagger}\mathbf{a}^{(s)}_{1:m}.
\end{equation}
For these unit vectors the two quantities obey the exact identity
\begin{equation}
\delta^{\mathrm{vec}}_{rs}=\sqrt{2(1-\mathcal O_{rs})},
\end{equation}
so they are equivalent summaries rather than independent diagnostics. High
overlap, or equivalently small aligned difference, supports the proposed mode
association. Both are basis-coordinate diagnostics, used only after expressing
the vectors in the same ordered Chebyshev basis; neither is an invariant
distance between continuum spectral subspaces. They do not remove the ambiguity caused by near-degenerate or
strongly non-normal clusters. The matching tolerance in Eq.~\eqref{eq:triple_match_criterion} is therefore
used only to associate modes. Acceptance of a triple is not, by itself, a
claim of convergence. The observed drift must be consistent with the backward errors, condition numbers, coefficient tails, endpoint checks, and precision dependence. Conversely, failure of a nearest-neighbour match at one working precision is not interpreted as physical disappearance of a mode when the relevant part of the spectrum is demonstrably ill-conditioned.

\subsection{Evidence hierarchy and practical reproducibility}
\label{subsec:evidence_hierarchy}

No single scalar diagnostic is used as a universal acceptance test. A
finite-pencil eigenpair is regarded as numerically credible only when the
following information is mutually consistent, i.e. the raw and equilibrated
representations, right and left polynomial residuals, linearisation
residuals, adjoint eigenvalue pairing, eigenvalue conditioning, Chebyshev
coefficient decay, endpoint residuals, cross-resolution frequency drift,
and cross-resolution eigenvector overlap. These quantities are interpreted
relative to the working precision and are then tested under changes of
resolution, arithmetic precision, collocation grid, compactification, and,
where applicable, formulation. The accompanying package supplies the Maple
assembly worksheets, machine-readable spectra and diagnostics, the Julia
Jost calculations, and the source data and scripts for the figures. Its
README records the software versions, units, inputs, and principal commands.
These materials are sufficient to inspect the numerical evidence without
burdening the article with workflow protocol. The logical hierarchy is
\begin{equation}
\label{eq:finite_pencil_evidence_hierarchy}
\text{small finite-pencil residual}
\;\not\!\Longrightarrow\;
\text{continuum convergence}
\;\not\!\Longrightarrow\;
\text{QNM pole}.
\end{equation}
Finally, all diagnostics in this section concern the finite-dimensional
polynomial eigenvalue problem. Small backward errors, moderate condition
numbers, decaying Chebyshev coefficients, endpoint consistency, and
resolution stability can establish that a root is a well-resolved eigenvalue
of a sequence of collocation pencils. They cannot establish that the same
frequency is an isolated pole of the Schwarzschild Green function. In particular, these tests do not by themselves distinguish a physical QNM pole from a cut-related finite representation or spectral pollution, and they say nothing about resonances on unphysical sheets not represented by the calculation. The physical-pole assessment is made only through the lateral Jost--Wronskian analysis described separately.

\section{Lateral Jost--Wronskian calculation and numerical diagnostics}\label{sec:jost_numerics}

The finite pencils generate candidate frequencies but do not determine pole character. We therefore evaluate directly the weighted determinant already defined in Eq.~\eqref{JostDeterminant}, using analytic representations of the horizon-ingoing and infinity-outgoing solutions of the original Regge--Wheeler equation. This construction is independent of the compactified matrix pencils and does not replace a Jost condition by compactified endpoint regularity. For exact solutions, the determinant is independent of the matching radius and satisfies Eq.~\eqref{JostConnectionRelation}. Note that the numerical objective is to compute the two lateral limits in Eq.~\eqref{LateralWronskians} and to decide whether an apparent small value is an isolated zero rather than a truncation, extrapolation, matching-radius, or normalisation effect. Throughout this section, $x=r/(2M)>1$, $\Omega=M\omega$, $F(x)=1-1/x$, and $\Omega=-i\alpha$ on the NIA.

\subsection{Analytic representations of the Jost solutions}
\label{subsec:jost_series}

Direct shooting in the lower half of the frequency plane is poorly suited to
this purpose because the required Jost sector can be exponentially subdominant to the linearly independent local solution. Instead, we employ the convergent Jaff\'e and Leaver--Tricomi-$U$ representations used in the analytic branch-cut construction of Refs.~\cite{Leaver1985PRSLA, Leaver1986PRD, CasalsOttewill2013PRD}. Let $\zeta=1-1/x$ with $0<\zeta<1$. The horizon-ingoing solution is represented as
\begin{equation}\label{eq:jost_num_horizon_factor}
f_H^{\mathrm{in}}(x,\Omega)=P_H(x,\Omega)J(x,\Omega),
\end{equation}
where
\begin{equation}\label{eq:jost_num_jaffe}
P_H(x,\Omega)=(x-1)^{-2i\Omega}x^{4i\Omega}e^{2i\Omega x},\quad
J(x,\Omega)=\sum_{n=0}^{\infty}a_n(\Omega)\zeta^n,\qquad
a_0=e^{-4i\Omega}.
\end{equation}
At truncation order $N$, the derivative is evaluated analytically, that is
\begin{equation}\label{eq:jost_num_horizon_derivative}
J_N'(x,\Omega)=\frac{1}{x^2}\sum_{n=1}^{N}n a_n \zeta^{n-1},\quad
\partial_x f_{H,N}^{\mathrm{in}}=P_H\left[J_N'+(\partial_x\log P_H)J_N\right],
\end{equation}
with
\begin{equation}
\partial_x\log P_H=-\frac{2i\Omega}{x-1}+\frac{4i\Omega}{x}+2i\Omega.
\end{equation}
The infinity-outgoing solution is represented by
\begin{equation}\label{eq:jost_num_infinity_factor}
f_\infty^{\mathrm{out}}(x,\Omega)=P_\infty(x,\Omega)H(x,\Omega),
\end{equation}
where
\begin{equation}\label{eq:jost_num_leaverU}
P_\infty(x,\Omega)=x^{1+s}(x-1)^{-2i\Omega}e^{2i\Omega x},\quad
H(x,\Omega)=\sum_{n=0}^{\infty}\widetilde a_n(\Omega)\mathscr U_n(x,\Omega),
\end{equation}
and
\begin{equation}\label{eq:jost_num_Ubasis}
\mathscr U_n(x,\Omega)=(1-4i\Omega)_nU\!\left(A_n,B,Z\right),\quad
A_n=s+1-4i\Omega+n,\qquad
B=2s+1,\qquad
Z=-4i\Omega x,
\end{equation}
with
\begin{equation}\label{eq:jost_num_U_normalization}
\widetilde a_0=(-4i\Omega)^{s+1-4i\Omega}.
\end{equation}
The complex power in Eq.~\eqref{eq:jost_num_U_normalization} and the Tricomi
function $U$ are evaluated on their principal branches. Consequently, for
$\Omega=\pm\delta-i\alpha$, the argument $Z$ approaches the negative real
axis from opposite sides. The same analytic representation therefore
implements the two continuations in Eq.~\eqref{LateralWronskians}. No
sheet-dependent ad hoc sign is inserted into the recurrence or the
prefactors. The derivative of the basis function is obtained from a contiguous relation. When its denominator is safely separated from zero, the implementation uses
\begin{equation}\label{eq:jost_num_contiguous_derivative}
\partial_x \mathscr U_n=\frac{A_n}{x}\left[\frac{n+1-s-4i\Omega}{n+1-4i\Omega}\mathscr U_{n+1}-\mathscr U_n
\right].
\end{equation}
Near a small value of $n+1-4i\Omega$, the code switches to the direct
identity
\begin{equation}\label{eq:jost_num_direct_U_derivative}
\partial_x \mathscr U_n=4i\Omega A_n(1-4i\Omega)_nU(A_n+1,B+1,Z),
\end{equation}
which follows from $\partial_ZU(A,B,Z)=-A\,U(A+1,B+1,Z)$. The switch avoids
introducing an artificial loss of accuracy through a nearly singular
contiguous formula. The derivative of the full infinity solution is then
\begin{equation}
\partial_x f_{\infty,N}^{\mathrm{out}}=P_\infty\left[H_N'+(\partial_x\log P_\infty)H_N\right],
\end{equation}
where
\begin{equation}
\partial_x\log P_\infty=\frac{1+s}{x}-\frac{2i\Omega}{x-1}+2i\Omega.
\end{equation}
The Jaff\'e and Leaver--$U$ coefficient sequences obey the same three-term
recurrence,
\begin{equation}\label{eq:jost_num_recurrence}
\alpha_n a_{n+1}+\beta_n a_n+\gamma_n a_{n-1}=0,\qquad a_{-1}=0,
\end{equation}
with the appropriate value of $a_0$ for the two representations. To avoid
ambiguity, $n$ in this subsection is the recurrence index, not the QNM
overtone $n_{\rm QNM}$; we set $\chi=2i\Omega$ to keep the frequency
dependence compact. The recurrence coefficients are
\begin{align}
\alpha_n&=(n+1)(n-2\chi+1),\\
\beta_n&=-\left[2n^2+(2-8\chi)n+8\chi^2-4\chi+\ell(\ell+1)+1-s^2\right],\\
\gamma_n&=n^2-4\chi n+4\chi^2-s^2.
\label{eq:jost_num_recurrence_coefficients}
\end{align}
These coefficients reproduce the recurrence in
Ref.~\cite{CasalsOttewill2013PRD} after converting its barred frequency to the
present convention $\bar\omega=2\Omega$. On the NIA, $\chi=2\alpha$, and the forward recurrence has a near-resonant index when $n+1\simeq4\alpha$. All NIA values are therefore approached with $\delta>0$. The recurrence is
never evaluated exactly at a singular lateral frequency. The baseline
truncation is chosen adaptively as
\begin{equation}\label{eq:jost_num_term_policy}
N_0(\alpha)=\max\left\{N_{\min},\left\lceil4\alpha\right\rceil+M_{\rm saf}\right\},
\end{equation}
where $M_{\rm saf}$ is a safety margin beyond the near-resonant index.
Successive levels $N_k=N_0+k\,\Delta N$ are compared rather than accepting a
single truncation. The calibration calculation fixed $N_{\min}=100$,
$M_{\rm saf}=180$, and $\Delta N=40$. Four term levels at 140 decimal digits
were sufficient over the lower and central parts of the ladder, while the
upper part was evaluated at 200 decimal digits with five term levels. The
precise term count used at each frequency is retained in the run metadata.

\subsection{Physical, regularised, and scale-free zero diagnostics}
\label{subsec:jost_zero_diagnostics}

Three related quantities are retained because each exposes a different
failure mode. First, the physical determinant $\mathcal D$ in Eq.~\eqref{JostDeterminant} is the quantity directly connected to the incoming connection coefficient through Eq.~\eqref{JostConnectionRelation}. Its numerical
value is computed as the difference
\begin{equation}
\mathcal D=T_1-T_2,\qquad
T_1=F f_H^{\mathrm{in}}\partial_x f_\infty^{\mathrm{out}},\qquad
T_2=F f_\infty^{\mathrm{out}}\partial_x f_H^{\mathrm{in}}.
\end{equation}
The cancellation indicator
\begin{equation}\label{eq:jost_num_cancellation}
c_{\mathcal D}=\frac{|\mathcal D|}{|T_1|+|T_2|}
\end{equation}
is recorded to reveal strong subtraction between large terms. A small
$c_{\mathcal D}$ is necessary near a determinant zero but is not sufficient
to establish one. Second, the horizon normalisation is regularised according to
\begin{equation}\label{eq:jost_num_regularized_D}
\widehat f_H^{\mathrm{in}}=-\sin(4\pi i\Omega)f_H^{\mathrm{in}},\qquad
\widehat{\mathcal D}=-\sin(4\pi i\Omega)\mathcal D.
\end{equation}
\begin{corollary}[Local order under sine regularisation]
\label{cor:sine-regularization-divisor}
On a domain avoiding $\{-ik/4:k\in\mathbb Z\}$, $\widehat{\mathcal D}$
and $\mathcal D$ have the same zero divisor. If $k\neq 0$ and
$\mathcal D$ admits a meromorphic germ at $\Omega_k=-ik/4$, then
\begin{equation}
\label{eq:sine-order-shift}
\operatorname{ord}_{\Omega_k}\widehat{\mathcal D}
=\operatorname{ord}_{\Omega_k}\mathcal D+1.
\end{equation}
A zero of $\widehat{\mathcal D}$ at a quarter point is therefore not, by itself, evidence for a physical pole.
\end{corollary}

\begin{proof}
Away from the quarter points, the sine factor is holomorphic and nowhere zero, so multiplication by it leaves the divisor unchanged by
Proposition~\ref{prop:intrinsic-jost-divisor}. At a quarter point,
\begin{equation}
\frac{d}{d\Omega}\sin(4\pi i\Omega)\bigg|_{\Omega=\Omega_k}
=4\pi i\cos(k\pi)\neq0,
\end{equation}
so its zero is simple. For the assumed meromorphic germ, additivity of local orders then gives Eq.~\eqref{eq:sine-order-shift}. The restriction $k\neq 0$ keeps this local meromorphic assertion separate from the branch point at $\Omega=0$.
\end{proof}
This factor removes the generic quarter-point poles associated with the chosen
horizon Jost normalisation and is useful when probing the nearly
quarter-spaced finite-pencil ladder. The implementation evaluates
$\widehat{\mathcal D}$ both from the regularised horizon series and from the
product in Eq.~\eqref{eq:jost_num_regularized_D}. Their difference is an
internal consistency test. The regularised determinant must not be used
naively at an exceptional point. In particular, for axial gravitational
perturbations with $s=\ell=2$, the algebraically special frequency is
$\Omega_{\mathrm{AS}}=-2i$. The sine factor vanishes there, but a
normalisation pole of $\mathcal D$ could in principle cancel that zero.
The direct calculation reported below instead finds a finite, nonzero
unregularised $\mathcal D$. Its regularised zero is consequently the local order shift in Corollary~\ref{cor:sine-regularization-divisor}, not a pole condition. At and near this point, the decisive quantities are $\mathcal D$, $A_{\mathrm{in}}$, and the mismatch defined below, supplemented by physical-half-plane contour tests. Third, we compute the scale-free logarithmic-derivative mismatch
\begin{equation}\label{eq:jost_num_mismatch}
\mathfrak m(\Omega;x)=F(x)\left(\frac{\partial_x f_\infty^{\mathrm{out}}}{f_\infty^{\mathrm{out}}}-\frac{\partial_x f_H^{\mathrm{in}}}{f_H^{\mathrm{in}}}
\right).
\end{equation}
Whenever neither Jost solution vanishes at the matching point, $\mathcal D=f_H^{\mathrm{in}}f_\infty^{\mathrm{out}}\,\mathfrak m$. Hence, $\mathfrak m=0$ is equivalent to $\mathcal D=0$ under this nonvanishing assumption. Unlike $\mathcal D$, $\mathfrak m$ is invariant under separate nonzero rescalings of the two Jost solutions. It is therefore the preferred scale-free zero indicator when the adopted Jost normalisations make $|\mathcal D|$ or $|\widehat{\mathcal D}|$ very large. Its nonzero value need not be independent of the matching radius, since
$\mathfrak m=\mathcal D/(f_H^{\mathrm{in}}f_\infty^{\mathrm{out}})$. What must be radius independent is the location of a genuine zero.

\subsection{Lateral extrapolation}\label{subsec:jost_lateral_extrapolation}

For each target $\alpha$, each side is evaluated on a descending set of positive offsets $\delta_1>\delta_2>\cdots>\delta_m>0$ in $\Omega_j^{\pm}=\pm\delta_j-i\alpha$, and the two edges are extrapolated independently. The production calculation uses a conservative low-order Neville construction rather than a single high-degree polynomial through all offsets. If the offsets are ordered from largest to smallest, the primary value $Q_0^{(2)}$ is the quadratic extrapolation through the three smallest offsets. It is compared with the cubic extrapolation $Q_0^{(3)}$ through the four smallest offsets and with the shifted quadratic $\widetilde Q_0^{(2)}$ through the preceding three offsets. For every $Q\in\{\mathcal D,\widehat{\mathcal D},A_{\mathrm{in}},\mathfrak m,f_H^{\mathrm{in}},f_\infty^{\mathrm{out}}\}$, the reported diagnostic is
\begin{equation}\label{eq:jost_num_delta_error}
\varepsilon_\delta[Q]=\max\!\left\{\operatorname{reldiff}\!\left(Q_0^{(2)},Q_0^{(3)}\right),\operatorname{reldiff}\!\left(Q_0^{(2)},\widetilde Q_0^{(2)}\right)\right\},
\end{equation}
where
\begin{equation}\label{eq:jost_num_reldiff}
\operatorname{reldiff}(z,w)=\frac{|z-w|}{\max\{1,|z|,|w|\}}.
\end{equation}
Despite the mnemonic name, this is a mixed scaled discrepancy: it is absolute when both arguments have magnitude
below one and relative to their larger magnitude above unit scale.
This is an empirical stability indicator rather than a rigorous extrapolation bound. The complete finite-$\delta$ sequence is retained, and a proposed zero is accepted only if it persists when the range and number of offsets are changed. Because the radial equation has real coefficients, the adopted branch convention should satisfy the edge symmetry
\begin{equation}\label{eq:jost_num_lip_symmetry}
Q^-(-i\alpha)=\overline{Q^+(-i\alpha)}
\end{equation}
for $Q=\mathcal D,\widehat{\mathcal D},\mathfrak m$. We monitor
\begin{equation}
 \varepsilon_{\mathrm{lip}}[Q]
 =\operatorname{reldiff}\!\left(Q^+,\overline{Q^-}\right).
 \label{eq:jost_num_lip_defect}
\end{equation}
Failure of Eq.~\eqref{eq:jost_num_lip_symmetry} at a level larger than the
truncation and arithmetic uncertainties signals an inconsistent branch,
special-function, or extrapolation evaluation. An optional branch-cut-strength diagnostic is
\begin{equation}\label{eq:jost_num_branch_strength}
q(\alpha)=\frac{f_\infty^{\mathrm{out},+}(-i\alpha)-f_\infty^{\mathrm{out},-}(-i\alpha)}{i f_\infty^{\mathrm{out}}(+i\alpha)}.
\end{equation}
With consistent lateral and positive-imaginary-axis normalisations, this
quantity should be real up to numerical error. It is used as a selected
branch-convention control rather than as a pole criterion.

\subsection{Truncation, matching-radius, and arithmetic diagnostics}
\label{subsec:jost_error_checks}

For a truncated series $S_N=\sum_{n=0}^{N}t_n$, we define the relative tail
indicator
\begin{equation}\label{eq:jost_num_tail}
\tau_N[S]=\frac{\displaystyle\max_{\max(0,N-w+1)\leqslant n\leqslant N}|t_n|}
 {\max\{\varepsilon_{\mathrm{floor}},|S_N|\}},
\end{equation}
where $w$ is a short terminal window. Separate values are recorded for the
physical horizon series, the regularised horizon series, and the infinity
series. For the infinity representation we use the larger of the function
and derivative tails. Tail decay is evidence that a particular series
summation is resolved. It is not evidence that the resulting determinant has
a zero. The change between the two largest term levels is monitored through
\begin{equation}\label{eq:jost_num_term_change}
\varepsilon_N[Q]=\max_{\pm}\operatorname{reldiff}\!\left(Q_{N_k}^{\pm},Q_{N_{k-1}}^{\pm}\right).
\end{equation}
The recurrence diagnostics
\begin{equation}
d_{\mathrm{rec}}=\min_n|\alpha_n|,\qquad
d_U=\min_n|n+1-4i\Omega|
\label{eq:jost_num_denominators}
\end{equation}
quantify proximity to the forward-recurrence and contiguous-derivative
singularities. The number of direct derivative fallbacks in Eq.~\eqref{eq:jost_num_direct_U_derivative} is also recorded. The Wronskians are evaluated at several matching radii $x_j>1$. Their relative spread is
\begin{equation}\label{eq:jost_num_radius_spread}
\varepsilon_x[Q]=\max_j\operatorname{reldiff}\!\left(Q(x_j),Q(x_1)\right),\qquad
Q\in\{\mathcal D,\widehat{\mathcal D}\}.
\end{equation}
A genuine determinant is independent of $x$, so Eq.~\eqref{eq:jost_num_radius_spread} is one of the strongest internal checks
on the simultaneous accuracy of both Jost solutions. The mismatch is also
reported at all radii, but its nonzero magnitude is not expected to be
constant. Rather, an inferred zero must remain at the same frequency for every
admissible matching radius. Tricomi $U$ is evaluated with arbitrary-precision complex ball arithmetic using the Arb/FLINT library through the Julia interface \texttt{Arblib.jl}~\cite{Johansson2017Arb}. The weakest available Arb accuracy in the run is recorded. The current implementation then extracts the complex midpoint of each $U$-ball and performs the recurrence summation, Wronskian algebra, and
extrapolation with arbitrary-precision \texttt{BigFloat} arithmetic. The
calculation is consequently a high-precision numerical assessment, but not a fully
interval-certified proof. Selected fixed targets, together with any credible
feature surviving the preceding checks, are repeated at higher working
precisions. Software versions, inputs, and representative outputs are
recorded in the accompanying repository rather than in the manuscript. The
principal diagnostics and their logical roles are summarised in
Table~\ref{tab:jost_diagnostics}.
\begin{table}[htbp]
\caption{Principal Jost/Wronskian diagnostics. A pass in any single row is not
sufficient for a pole declaration. The diagnostics expose distinct numerical
and analytic failure modes.}
\label{tab:jost_diagnostics}
\centering
\footnotesize
\begingroup
\setlength{\tabcolsep}{4pt}
\renewcommand{\arraystretch}{1.08}
\newcommand{\JDQuantity}[1]{\parbox[t]{0.135\textwidth}{\raggedright #1}}
\newcommand{\JDDefinition}[1]{\parbox[t]{0.295\textwidth}{\raggedright #1}}
\newcommand{\JDRole}[1]{\parbox[t]{0.500\textwidth}{\raggedright #1}}
\begin{tabular}{@{}lll@{}}
\toprule
\JDQuantity{Quantity}
& \JDDefinition{Definition or comparison}
& \JDRole{Numerical role} \\
\midrule
\JDQuantity{$\mathcal D$, $A_{\mathrm{in}}$}
& \JDDefinition{Eqs.~\eqref{JostDeterminant} and \eqref{JostConnectionRelation}}
& \JDRole{Physical zero condition away from exceptional normalisations.} \\
\JDQuantity{$\widehat{\mathcal D}$}
& \JDDefinition{Eq.~\eqref{eq:jost_num_regularized_D}}
& \JDRole{Removes generic quarter-point normalisation poles; not decisive at the algebraically special point.} \\
\JDQuantity{$\mathfrak m$}
& \JDDefinition{Eq.~\eqref{eq:jost_num_mismatch}}
& \JDRole{Normalisation-independent zero indicator, provided neither Jost solution vanishes at the matching point.} \\
\addlinespace[2pt]
\JDQuantity{$\tau_N$, $\varepsilon_N$}
& \JDDefinition{Eqs.~\eqref{eq:jost_num_tail} and \eqref{eq:jost_num_term_change}}
& \JDRole{Tests convergence under series truncation.} \\
\JDQuantity{$\varepsilon_\delta$}
& \JDDefinition{Eq.~\eqref{eq:jost_num_delta_error}}
& \JDRole{Tests stability of each one-sided $\delta\to0$ extrapolation.} \\
\JDQuantity{$\varepsilon_x$}
& \JDDefinition{Eq.~\eqref{eq:jost_num_radius_spread}}
& \JDRole{Tests Wronskian independence of the matching point.} \\
\JDQuantity{$\varepsilon_{\mathrm{lip}}$}
& \JDDefinition{Eq.~\eqref{eq:jost_num_lip_defect}}
& \JDRole{Tests the expected conjugacy of the two lateral continuations.} \\
\JDQuantity{$\varepsilon_{\mathrm{reg}}$}
& \JDDefinition{$\operatorname{reldiff}(\widehat{\mathcal D},-\sin(4\pi i\Omega)\mathcal D)$}
& \JDRole{Checks the independent physical and regularised horizon evaluations.} \\
\addlinespace[2pt]
\JDQuantity{$d_{\mathrm{rec}},d_U$}
& \JDDefinition{Eq.~\eqref{eq:jost_num_denominators}}
& \JDRole{Warns of recurrence or contiguous-derivative near resonances.} \\
\JDQuantity{$c_{\mathcal D}$}
& \JDDefinition{Eq.~\eqref{eq:jost_num_cancellation}}
& \JDRole{Measures subtractive cancellation in the Wronskian, without classifying its cause.} \\
\JDQuantity{Arb accuracy and precision repeat}
& \JDDefinition{Ball-accuracy report and fixed-target reruns}
& \JDRole{Tests special-function evaluation and the arithmetic floor.} \\
\JDQuantity{Contour winding and complex root}
& \JDDefinition{Argument principle and Newton refinement}
& \JDRole{Tests numerically for an isolated off-axis zero of a specified physical continuation; exact exclusion requires certified contour hypotheses.} \\
\bottomrule
\end{tabular}
\endgroup
\end{table}
\FloatBarrier

\subsection{Off-axis controls and local zero refinement}
\label{subsec:jost_offaxis_controls}

Before interpreting the branch cut, the method is calibrated on established
off-axis Schwarzschild QNMs. At each supplied control frequency, we evaluate
$\widehat{\mathcal D}$ and $\mathfrak m$ at successive term levels and
several matching radii. We then refine selected seeds with Newton iteration
on $\widehat{\mathcal D}$, which has the same nonexceptional off-axis zeros
as $\mathcal D$. The derivative is approximated along the real direction by
the fourth-order centred stencil
\begin{equation}\label{eq:jost_num_newton_derivative}
\widehat{\mathcal D}'(\Omega)\simeq\frac{-\widehat{\mathcal D}(\Omega+2h)+8\widehat{\mathcal D}(\Omega+h)-8\widehat{\mathcal D}(\Omega-h)+\widehat{\mathcal D}(\Omega-2h)}{12h}.
\end{equation}
For an analytic determinant, differentiation along the real direction gives
the complex derivative. The complete Newton history, including the complex
step and determinant value, is retained. Recovery of the known off-axis modes
with decreasing truncation error and matching-radius spread is a prerequisite
for interpreting any NIA calculation.

\subsection{Scans, physical-half-plane boxes, and the argument principle}
\label{subsec:jost_searches}

Evaluating the finite-pencil frequencies alone would be incomplete because a true Jost zero could lie between them. We therefore scan $\alpha$ on each edge using $|\widehat{\mathcal D}^{+}|$ as the primary zero diagnostic, while retaining $|\mathcal D^{+}|$, the complex phase, and the left lip as consistency checks. The logarithmic-derivative mismatch is kept only as an auxiliary quantity because it can become extremely small when the product of the two Jost amplitudes is large even though $\mathcal D$ is clearly separated from zero relative to the observed numerical variations in the adopted normalisation. Local minima, component sign changes, and unusually large phase increments are candidate generators only. Any credible feature must be rescanned on finer meshes and reevaluated with the full truncation, lateral-offset, matching-radius, and precision checks. A local minimum of $|\mathcal D|$, $|\widehat{\mathcal D}|$, or $|\mathfrak m|$ is never by itself labelled a pole. Possible off-axis zeros of the physical continuations are tested with rectangular contours that remain strictly in either $\operatorname{Re}\Omega>0$ or $\operatorname{Re}\Omega<0$. These contours are reached from the upper half-plane without crossing the NIA and therefore do not enter unphysical sheets. Along a counterclockwise contour $\Gamma$, the phase accumulation is computed from successive ratios of $\widehat{\mathcal D}$, that is
\begin{equation}\label{eq:jost_num_winding}
\mathcal N_\Gamma\simeq\frac{1}{2\pi}\sum_j\arg\left[\frac{\widehat{\mathcal D}(\Omega_{j+1})}{\widehat{\mathcal D}(\Omega_j)}\right].
\end{equation}

\Needspace{13\baselineskip}
\begin{proposition}[Exact contour criterion]
\label{prop:certified-argument-principle}
Let $U$ be an open subset of one fixed continuation sheet, let
$R\Subset U$ be a rectangle with positively oriented boundary, and let
$E$ be holomorphic on a neighbourhood of $\overline R$. If $E$ is nonzero
on $\partial R$, then
\begin{equation}
\label{eq:exact-argument-principle}
\operatorname{wind}(E(\partial R),0)
=\frac{1}{2\pi i}\int_{\partial R}\frac{E'(\Omega)}{E(\Omega)}\,d\Omega
\end{equation}
is the number of zeros of $E$ in $R$, counted with multiplicity. If $E$ is
meromorphic instead, the integral counts zeros minus poles.
\end{proposition}

\begin{proof}
If $\Omega_0$ is a zero of order $m$, write
$E(\Omega)=(\Omega-\Omega_0)^m h(\Omega)$ with $h(\Omega_0)\neq0$. Then
$E'/E=m/(\Omega-\Omega_0)+h'/h$, so the logarithmic derivative has residue
$m$ at $\Omega_0$. The residue theorem therefore turns the integral into the
sum of the zero multiplicities inside $R$. Along the nonvanishing boundary,
the same integral is the net change of a continuous argument of $E$ divided
by $2\pi$, which is the winding number. If $E$ has a pole of order $m$, the
same factorisation with exponent $-m$ contributes residue $-m$; hence the
meromorphic version counts zeros minus poles.
\end{proof}

Topologically, the left-hand side of
Eq.~\eqref{eq:exact-argument-principle} is the degree of the map
$E|_{\partial R}:\partial R\to\mathbb C^*$. It is unchanged under any
homotopy that keeps the contour image in $\mathbb C^*$. This explains both
the robustness of a well-resolved winding count and the necessity of
excluding boundary zeros before assigning it a zero count. The determinant-line viewpoint also makes this topological integer
independent of the chosen local Jost frames. Indeed, if
$\widetilde E=abE$, where $a$ and $b$ are holomorphic and nowhere zero on a neighbourhood of $\overline R$, then the boundary map $ab:\partial R\to
\mathbb C^*$ has zero winding by the argument principle. Hence,
\begin{equation}
\operatorname{wind}(\widetilde E(\partial R),0)=\operatorname{wind}(E(\partial R),0).
\end{equation}
This is the topological counterpart of
Proposition~\ref{prop:intrinsic-jost-divisor}, i.e. determinant magnitudes depend on the local trivialisation, whereas the zero count does not. If $ab$ is meromorphic and finite and nonzero on $\partial R$, the winding instead changes by $Z_{ab}(R)-P_{ab}(R)$, counted with multiplicity. This normalisation divisor must be accounted for explicitly. Proposition~\ref{prop:certified-argument-principle} is the exact theorem. The present midpoint-based computation supplies stable numerical winding evidence but does not certify every hypothesis. Notice that the full contour image is not enclosed by intervals, and a pole-free holomorphic representation on each closed rectangle is not proved here. A certification-oriented protocol would require nonvanishing enclosures on the boundary, separation from the
truncation and arithmetic scales, stability under contour, precision, and series refinement, and compatibility of nested contours. It would also require every sampled principal phase increment to remain below $\pi$ and to decrease under subdivision. The present strip calculation checks sampled nonvanishing, boundary-point doubling, phase-increment decrease, the regularisation identity, and known-QNM and empty-contour calibrations. It holds the precision, matching radius, term count, and rectangle geometry fixed. We therefore report its result as resolution-refined numerical winding evidence, not as an accepted computer-assisted zero-free theorem. It says nothing about a zero on the cut, in the near-cut gap, or on an unphysical sheet. Any credible minimum from a scan would require targeted refinement at several term counts, radii, and precisions.

\subsection{Physical-pole criteria and numerical checks}
\label{subsec:jost_classification}

The numerical evidence is interpreted through a hierarchy rather than a
single threshold. A finite-pencil NIA candidate is identified as an isolated
physical pole only if a specified physical lateral determinant converges to
zero at one isolated location and, at a matching radius where neither Jost
solution vanishes, the scale-free mismatch does likewise. That location must
also survive all of the following changes
\begin{enumerate}
\item 
decreasing the lateral offset and altering the extrapolation window;
\item 
increasing the Jaff\'e and Leaver--$U$ truncations;
\item 
changing the matching radius;
\item 
increasing the arithmetic and special-function precision;
\item 
replacing the physical determinant by the scale-free mismatch and, away from exceptional points, by the regularised determinant;
\item 
a compatible complex-root or argument-principle calculation in the same physical half-plane continuation.
\end{enumerate}
No individual diagnostic promotes a frequency to pole status. In particular, a small mismatch indicator, a small finite-$\delta$ determinant, or a stable finite-pencil frequency remains an observation until the full hierarchy is satisfied. Table~\ref{tab:jost_checks} organises the calculations by their
scientific purpose. Off-axis QNMs fix the signs and normalisations.
Representative NIA points set adequate numerical parameters. Candidate-point evaluations test the proposed ladder directly. A uniform mesh tests the gaps between those points, and off-cut contours test for nearby zeros in the physical continuations actually computed.
\begin{table}[htbp]
\caption{Decision protocol for the Jost/Wronskian calculations. The first
five rows describe calculations actually performed; the final row is a
conditional follow-up used only if an earlier diagnostic generates a credible
feature. Numerical parameters are retained as run metadata. No row alone
establishes a physical pole.}
\label{tab:jost_checks}
\centering
\footnotesize
\begingroup
\setlength{\tabcolsep}{4pt}
\renewcommand{\arraystretch}{1.08}
\newcommand{\JCStage}[1]{\parbox[t]{0.180\textwidth}{\raggedright #1}}
\newcommand{\JCCalculation}[1]{\parbox[t]{0.340\textwidth}{\raggedright #1}}
\newcommand{\JCPurpose}[1]{\parbox[t]{0.410\textwidth}{\raggedright #1}}
\begin{tabular}{@{}lll@{}}
\toprule
\JCStage{Check} & \JCCalculation{Calculation} & \JCPurpose{Scientific role} \\
\midrule
\JCStage{Special-function and off-axis calibration}
& \JCCalculation{Tricomi-$U$ identity, known off-axis QNM seeds, one Newton recovery, and a few NIA probes}
& \JCPurpose{Checks the special-function interface, signs, normalisations, derivatives, and output diagnostics.} \\
\JCStage{Representative NIA calibration}
& \JCCalculation{The algebraically special neighbourhood and selected low, central, and upper finite-pencil ladder points}
& \JCPurpose{Sets adequate precision, term margin, matching radii, and lateral-offset window.} \\
\JCStage{Candidate-point evaluations}
& \JCCalculation{Both lips at every regular finite-pencil candidate, in batches}
& \JCPurpose{Tests whether any finite-pencil location is also an isolated lateral Jost zero.} \\
\JCStage{Uniform NIA mesh}
& \JCCalculation{A uniform $0.05$-spaced $\alpha$-mesh; local refinement only if a credible feature appears}
& \JCPurpose{Detects minima or zeros between the finite-pencil frequencies.} \\
\JCStage{Off-cut boundary sampling}
& \JCCalculation{Known-QNM and empty-box controls, representative and algebraically-special-adjacent boxes, and full-strip boundaries in both physical continuations}
& \JCPurpose{Computes sampled numerical winding evidence for nearby off-axis zeros; it neither certifies a zero count nor probes unphysical sheets.} \\
\JCStage{Triggered complex refinement}
& \JCCalculation{Conditional Newton refinement and repeats at several precisions, term counts, radii, offsets, and contours}
& \JCPurpose{Tests a persistent minimum for an isolated zero; the uniform NIA mesh triggered no such refinement.} \\
\bottomrule
\end{tabular}
\endgroup
\end{table}
\FloatBarrier
Within the computed physical continuations, we reserve the term \emph{isolated
physical pole} for a zero of the specified physical lateral continuation that
survives all reported variations in truncation, matching radius, precision,
and contour. A finite-pencil candidate with nonzero physical lateral
determinants and no stable nearby complex zero is disfavoured as a pole in the
tested region, but that negative result does not prove that the nodes
approximate the cut. Such a positive interpretation requires the weighted
convergence criterion in Sec.~\ref{sec:keldysh_cut_limit}. An
apparent zero created only by the sine regularisation is a normalisation
artefact; this is precisely what happens when $\mathcal D$ is finite at the
axial algebraically special point but $\widehat{\mathcal D}$ vanishes. A minimum whose location or depth changes materially with the numerical
parameters remains unresolved. A resonance on an unphysical sheet lies outside this classification because assessing it requires an explicit continuation through the NIA cut, as in the Schwarzschild analyses of Refs.~\cite{MaassenVanDenBrink2000PRD, Leung2003CQG}. The hierarchy is deliberately stricter than finite-pencil convergence. It separates an accurate eigenpair of a discretised compactified problem from an
isolated zero of the physical horizon--infinity Jost determinant.

\section{Finite-pencil spectra and representation dependence}
\label{sec:finite_pencil_results}

We now report the finite-dimensional spectra in the axial gravitational sector
$(s,\ell)=(2,2)$. On the NIA we write $\Omega=-i\alpha$, with $\alpha>0$.
The term \emph{finite-pencil root} denotes an eigenvalue of a discretised
quadratic pencil; it does not imply an isolated pole of the analytically
continued Green function. We compare three realisations. C1--R uses
$x=2/(1-y)$ on Chebyshev roots, C1--L uses the same compactification on the
endpoint-inclusive Chebyshev--Lobatto grid, and C2--R uses
$x=4/(1-y)^2$ on Chebyshev roots. Together they separate arithmetic precision,
collocation and endpoint-row treatment, and radial compactification.
Table~\ref{tab:finite_pencil_configurations} summarises the computations. Unless
stated otherwise, a quoted frequency is the value at the largest resolution
in its three-resolution window, not an average over that window.

\subsection{Finite-pencil configurations and controls}
\label{subsec:finite_pencil_configurations}

\begin{table*}[t]
\caption{Finite-pencil configurations for the axial gravitational
$(s,\ell)=(2,2)$ sector. In each pair $(d_{\rm M},d_{\rm J})$, $d_{\rm M}$ is
the decimal precision of the Maple matrix export and $d_{\rm J}$ is the Julia
working precision. The $d_{\rm M}=50$ calculations check the implementation;
$d_{\rm M}=100$ provides the resolution census and $d_{\rm M}=200$ the
high-precision comparison.}
\label{tab:finite_pencil_configurations}
\centering
\small
\begingroup
\setlength{\tabcolsep}{4pt}
\renewcommand{\arraystretch}{1.08}
\begin{tabular*}{0.94\textwidth}{@{\extracolsep{\fill}}lllcc@{}}
\toprule
Realisation & Compactification & Grid
& $(d_{\rm M},d_{\rm J})$ & Resolutions \\
\midrule
C1--R & $x=2/(1-y)$ & Chebyshev roots
 & $(100,130)$ & $n=40,50,\ldots,160$ \\
 & & & $(200,230)$ & $n=120,130,140,150,160$ \\
\addlinespace[2pt]
C1--L & $x=2/(1-y)$ & Chebyshev--Lobatto
 & $(50,80)$ & $n=40,45,50$ \\
 & & & $(200,230)$ & $n=120,130,140,150,160$ \\
\addlinespace[2pt]
C2--R & $x=4/(1-y)^2$ & Chebyshev roots
 & $(50,80)$ & $n=40,45,50$ \\
 & & & $(100,130)$ & $n=40,50,\ldots,160$ \\
 & & & $(200,230)$ & $n=120,130,140,150,160$ \\
\bottomrule
\end{tabular*}
\endgroup
\end{table*}

Every production calculation returned exactly $2n$ finite generalized
eigenvalues. At $d_{\rm M}=200$, the exported equilibrated matrices obeyed
$M_q^{\rm eq}=D M_q^{\rm raw}$ to approximately $5\times10^{-201}$ in the
maximum-entry norm, while the scaled block-row norms differed from unity by at
most approximately $10^{-199}$. Independently executed copies of every shared
resolution were byte-for-byte identical. These checks establish reproducibility of the finite-dimensional problems before any physical interpretation of their spectra is attempted. The fixed-resolution precision comparison is particularly informative. The full NIA counts are shown in Table~\ref{tab:nia_counts_precision}. For C1--R, the $d_{\rm M}=100$ calculation loses a substantial part of the NIA spectrum at $n=150$ and $160$, whereas the $d_{\rm M}=200$ calculation restores the nearly linear count. In contrast, C2--R has identical $d_{\rm M}=100$ and $d_{\rm M}=200$ counts at every common resolution. Thus, the high-resolution failure of the original C1--R census was an arithmetic-precision ceiling, while the different C2--R spectral organisation is not.

\begin{table}[t]
\caption{Full NIA eigenvalue counts at the resolutions common to the
$d_{\rm M}=100$ and $d_{\rm M}=200$ calculations. Counts refer to the complete finite
spectra, not only to the $|\Omega|\leqslant100$ matching window.}
\label{tab:nia_counts_precision}
\centering
\small
\begingroup
\setlength{\tabcolsep}{6pt}
\renewcommand{\arraystretch}{1.08}
\begin{tabular}{@{}ccccc@{}}
\toprule
& \multicolumn{2}{c}{C1--R} & \multicolumn{2}{c}{C2--R} \\
\cmidrule(lr){2-3}\cmidrule(l){4-5}
$n$ & $d_{\rm M}=100$ & $d_{\rm M}=200$ & $d_{\rm M}=100$ & $d_{\rm M}=200$ \\
\midrule
120 & 124 & 124 & 140 & 140 \\
130 & 132 & 132 & 160 & 160 \\
140 & 146 & 146 & 170 & 170 \\
150 & 120 & 154 & 178 & 178 \\
160 & 116 & 166 & 194 & 194 \\
\bottomrule
\end{tabular}
\endgroup
\end{table}

\subsection{Off-axis quasinormal-mode controls}
\label{subsec:off_axis_controls}

The ordinary damped off-axis modes provide the primary implementation control,
because their frequencies should be independent of the compactification and
of whether roots or Lobatto points are used. Table~\ref{tab:off_axis_controls}
uses C1--R at $n=160$ as the reference and compares it with C1--L and C2--R.
The modes are ordered by increasing damping. In every row, the difference
between numerical realisations is smaller than the largest unresolved
three-resolution drift among the three calculations. The expected compactification-independent spectrum is therefore recovered before the NIA
population is interpreted.

\begin{table*}[t]
\caption{First six off-axis control modes at $n=160$. Here
$\Delta_{\rm CL}=|\Omega_{\rm C1-L}-\Omega_{\rm C1-R}|$ and
$\Delta_{\rm C2}=|\Omega_{\rm C2-R}-\Omega_{\rm C1-R}|$. The last column is
the largest three-resolution drift among C1--R, C1--L, and C2--R.}
\label{tab:off_axis_controls}
\centering
\small
\begingroup
\setlength{\tabcolsep}{4pt}
\renewcommand{\arraystretch}{1.08}
\begin{tabular*}{0.96\textwidth}{@{\extracolsep{\fill}}cccccc@{}}
\toprule
$j$
& \shortstack{$\operatorname{Re}\Omega_{\rm C1-R}$}
& \shortstack{$-\operatorname{Im}\Omega_{\rm C1-R}$}
& $\Delta_{\rm CL}$ & $\Delta_{\rm C2}$
& \shortstack{Largest\\drift} \\
\midrule
0 & 0.373671684418 & 0.088962315689 & $8.65\times10^{-34}$ & $5.41\times10^{-34}$ & $4.95\times10^{-31}$ \\
1 & 0.346710996879 & 0.273914875291 & $3.54\times10^{-24}$ & $2.43\times10^{-24}$ & $3.31\times10^{-22}$ \\
2 & 0.301053454612 & 0.478276983223 & $1.36\times10^{-16}$ & $1.02\times10^{-16}$ & $3.27\times10^{-15}$ \\
3 & 0.251504962185 & 0.705148202417 & $1.99\times10^{-11}$ & $1.62\times10^{-11}$ & $2.06\times10^{-10}$ \\
4 & 0.207514606868 & 0.946844894625 & $3.10\times10^{-8}$ & $2.96\times10^{-8}$ & $1.71\times10^{-7}$ \\
5 & 0.169297834948 & 1.195605594871 & $3.11\times10^{-6}$ & $2.13\times10^{-6}$ & $1.17\times10^{-5}$ \\
\bottomrule
\end{tabular*}
\endgroup
\end{table*}
An external continued-fraction comparison gives the same conclusion. After
converting Berti's tabulation from $2M\omega$ to $\Omega=M\omega$, the
absolute C1--R discrepancies for overtones $j=0,\ldots,5$ are respectively
$1.12\times10^{-16}$, $5.55\times10^{-17}$, $1.24\times10^{-16}$,
$1.63\times10^{-11}$, $2.73\times10^{-8}$, and $2.92\times10^{-6}$~\cite{BertiRingdownData}. Each is below the corresponding largest
three-resolution drift in Table~\ref{tab:off_axis_controls}. The control is
therefore external to the three collocation realisations, not merely an
agreement among them.
The algebraically special finite-pencil root is likewise recovered as $\Omega_{\rm AS}=-2i$ in all three realisations. At $d_{\rm M}=200$, its roots-versus-Lobatto difference is approximately $1.6\times10^{-175}$, and the C2--R value at $n=160$ differs from $-2i$ by approximately $9\times10^{-178}$. This is a sensitive check of the matrix signs, the frequency convention and the endpoint relations. It does not, however, settle the exceptional Jost-normalisation issue at the algebraically special frequency, which is deferred to the Wronskian analysis.

\subsection{First compactification: resolution and precision}
\label{subsec:c1_resolution_precision}

Across the five computed resolutions $n=120,130,\ldots,160$, the C1--R NIA
population grows approximately one-for-one with the spectral resolution. Its
terminal root obeys the empirical finite-window law
\begin{equation}\label{eq:c1_terminal_law}
4\alpha_{\max}\simeq n-2.
\end{equation}
For example, at $n=160$ the $d_{\rm M}=200$ roots calculation gives $4\alpha_{\max}^{\rm C1-R}=158.0009243564$. The C1--L value is $4\alpha_{\max}^{\rm C1-L}=158.0009244674$, showing that the terminal scaling is insensitive to the change of grid and to the replacement of the endpoint rows. The lowest NIA root behaves differently. On both C1 grids it lies on an $n^{-2}$ scale across the computed window, but the coefficient is grid dependent, that is
\begin{align}
n^2\alpha_{\min}^{\rm C1-R}&=6.38794,\ldots,6.38655,\\
 n^2\alpha_{\min}^{\rm C1-L}&=3.91956,\ldots,3.90260,
\end{align}
for $120\leqslant n\leqslant160$. This distinction is visible in Fig.~\ref{fig:c1_edge_scalings}. The same figure also displays the arithmetic failure of the C1--R $d_{\rm M}=100$ terminal root at $n=160$ and its recovery at $d_{\rm M}=200$.
\begin{figure*}[t]
\centering
\begin{minipage}{0.48\textwidth}
\centering
\includegraphics[width=\linewidth]{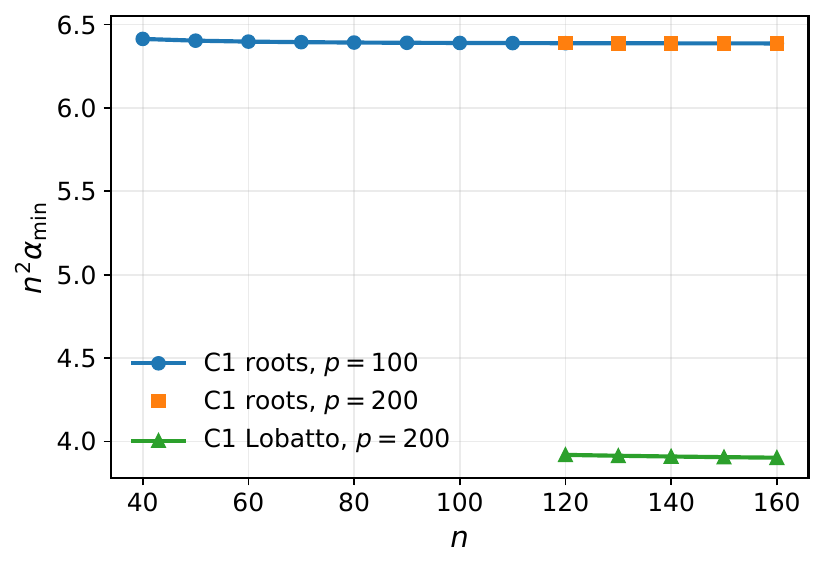}
\end{minipage}
\hfill
\begin{minipage}{0.48\textwidth}
\centering
\includegraphics[width=\linewidth]{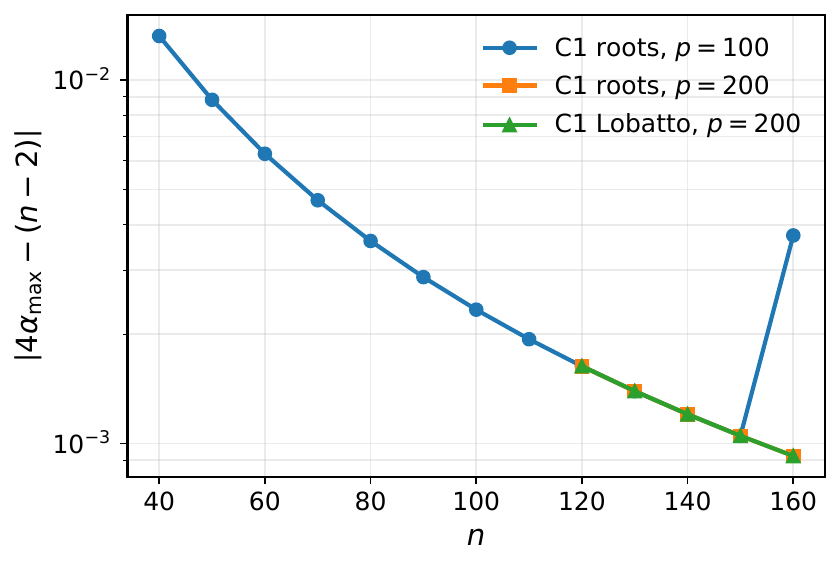}
\end{minipage}
\caption{Edge scalings for the first compactification. Left: the lowest NIA
root lies on an $n^{-2}$ scale, but the coefficient is different on the roots
and Lobatto grids. Right: the terminal root tracks
$4\alpha_{\max}=n-2$ over the computed resolutions on both grids. The upward displacement of the $d_{\rm M}=100$
roots point at $n=160$ accompanies the high-resolution loss of NIA roots and
is removed at $d_{\rm M}=200$.}
\label{fig:c1_edge_scalings}
\end{figure*}
For the regular high-damping C1--R roots at $n=160$, the equilibrated scaled backward errors are of order $10^{-239}$, while the raw/equilibrated frequency gaps lie between approximately $4\times10^{-104}$ and $5\times10^{-89}$. The relative condition indicators are very large, $10^{99}\lesssim\kappa_{\rm rel,2}\lesssim2\times10^{113}$, but the largest observed product with the equilibrated right backward error, $\kappa_{\rm rel,2}\eta_2^{\rm R}$, is approximately $1.4\times10^{-125}$. The smallest relative cross-resolution drift in the same branch is approximately $1.65\times10^{-10}$, corresponding to an absolute drift of approximately $3.72\times10^{-9}$. Both are far larger than the observed arithmetic and representation discrepancies. The remaining motion is therefore consistent with discretisation error rather than an arithmetic-precision ceiling, despite the severe conditioning of the individual eigenvalues. As noted after Proposition~\ref{prop:qep-first-order-condition-number}, $\kappa_{\rm rel,2}\eta_2^{\rm R}$ alone is not a rigorous forward-error enclosure.

\subsection{Roots-versus-Lobatto comparison}
\label{subsec:c1_grid_comparison}

The roots and Lobatto implementations of the first compactification produce the
same high-damping finite-pencil structure. In the window $(140,150,160)$, both contain 70 matched NIA triples: one root approaching the
origin, the algebraically special root, and 68 regular high-damping roots. At
$n=160$, the regular branches occupy
\begin{align}
15.0786158201&\leqslant\alpha\leqslant31.8428207066&&\text{(C1--R)},\\
15.0784880419&\leqslant\alpha\leqslant31.8428207059&&\text{(C1--L)}.
\end{align}
Their median nearest-neighbour spacings are respectively $0.250192684076$, and $0.250192684065$. After removing five modes at each edge, the mean spacings are
$0.2502103020$ and $0.2502101402$, with sample standard deviations of order
$6\times10^{-5}$. Thus, the approximately quarter-spaced lattice is not an
artefact of choosing roots rather than Lobatto points. Figure~\ref{fig:nia_spectrum_comparison} compares the complete NIA spectra at
$n=160$. Near the origin, the two C1 grids already differ visibly and C2--R
introduces a substantially finer hierarchy. At finite damping, C1--R and
C1--L are nearly indistinguishable, whereas C2--R has a different distribution.
The figure is a comparison of full finite spectra, not a claim that each plotted
root defines a converged trajectory.
\begin{figure*}[t]
\centering
\begin{minipage}{0.48\textwidth}
\centering
\includegraphics[width=\linewidth]{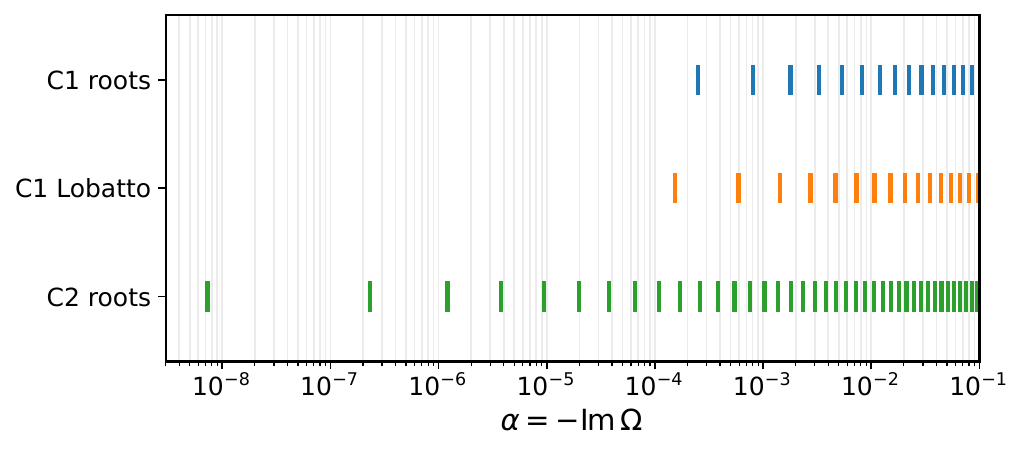}
\end{minipage}
\hfill
\begin{minipage}{0.48\textwidth}
\centering
\includegraphics[width=\linewidth]{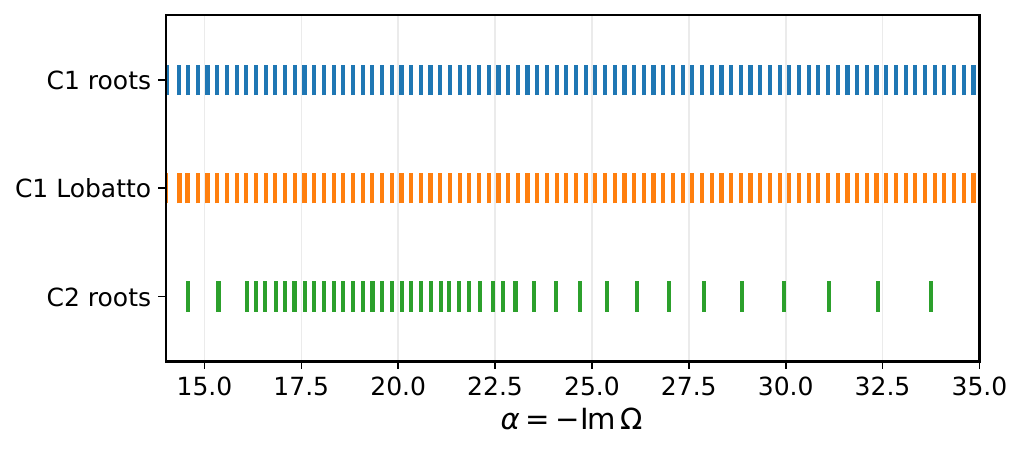}
\end{minipage}
\caption{NIA eigenvalue locations at $n=160$. Each vertical mark is one
finite-pencil root. Left: the near-zero region on a logarithmic scale. Right:
the finite-damping interval $14\leqslant\alpha\leqslant35$. The C1 roots and Lobatto
spectra almost coincide in the regular high-damping region, while the C2 roots
are reorganised by the different compactification.}
\label{fig:nia_spectrum_comparison}
\end{figure*}
For the 68 regular C1 trajectories, the median roots-versus-Lobatto frequency
difference is $\operatorname{median}|\alpha_{\rm C1-R}-\alpha_{\rm C1-L}|
 =2.94\times10^{-10}$, and the maximum is $2.95\times10^{-4}$. Every one of the 68 grid differences is smaller than the larger of the two corresponding within-grid resolution drifts. The median ratio of grid difference to that drift is
$1.58\times10^{-2}$, and the maximum ratio is $0.545$. This comparison is
shown in Fig.~\ref{fig:c1_grid_drift}.
\begin{figure}[t]
\centering
\includegraphics[width=0.96\columnwidth]{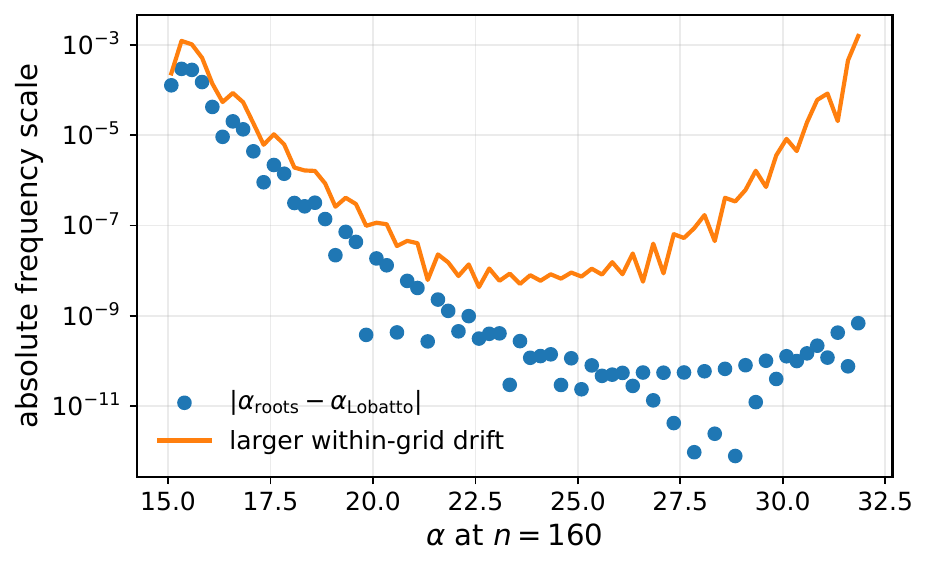}
\caption{Roots-versus-Lobatto frequency differences for the 68 matched regular
C1 NIA trajectories at $n=160$, compared with the larger within-grid
three-resolution drift. The grid difference remains below the unresolved
resolution scale for every trajectory.}
\label{fig:c1_grid_drift}
\end{figure}
The agreement extends to the eigenfunctions. After phase alignment, the
minimum roots-versus-Lobatto Chebyshev-coefficient overlap over the 68 regular
modes is $0.9999999999764556$, while the median phase-aligned unit-vector difference is $4.67\times10^{-12}$ and the maximum is $6.86\times10^{-6}$. The near-zero root behaves differently. Its roots-versus-Lobatto overlap is only
$0.980699$, and the corresponding unit-vector difference is $0.1965$. At
$n=160$ the Lobatto near-zero mode also has a coefficient-tail measure of about
$0.155$. The root approaching $\Omega=0$ is therefore not resolved as a nonzero grid-independent mode. For the Lobatto pencil the limiting endpoint equations are rows of the matrix problem. Their very small residuals consequently verify consistency between the Maple assembly and the Julia reconstruction, but they are not independent out-of-sample boundary tests. The off-axis controls, coefficient decay, raw/equilibrated invariance, and roots-versus-Lobatto comparisons provide the independent validation.

\FloatBarrier
\subsection{Second compactification}
\label{subsec:c2_results}

\begin{figure}[htbp]
\centering
\begin{minipage}{0.48\textwidth}
\centering
\includegraphics[width=\linewidth]{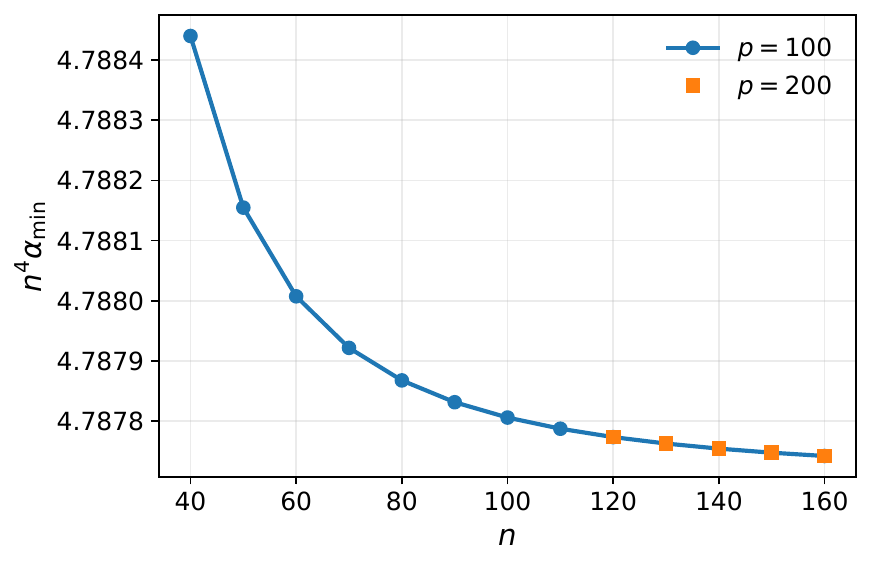}
\end{minipage}
\hfill
\begin{minipage}{0.48\textwidth}
\centering
\includegraphics[width=\linewidth]{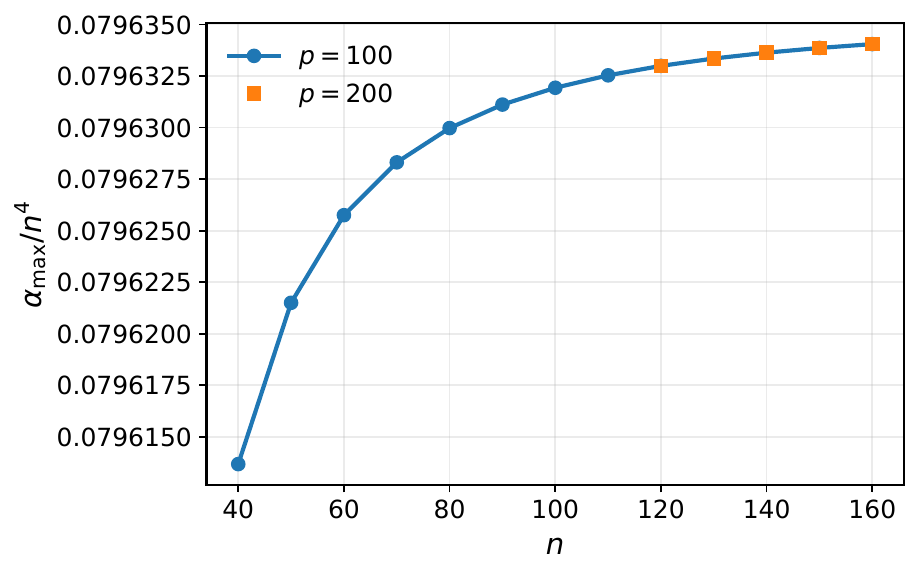}
\end{minipage}
\caption{Normalised C2--R edge scalings over $n=120,130,\ldots,160$. Left:
$n^4\alpha_{\min}$ is nearly constant. Right: $\alpha_{\max}/n^4$ is nearly
constant. The $d_{\rm M}=100$ and $d_{\rm M}=200$ points coincide on the plotted scale,
supporting the interpretation that the observed $n^{\pm4}$ scaling is not a
precision artefact.}
\label{fig:c2_edge_scalings}
\end{figure}
\FloatBarrier

The second compactification reorganises the NIA population even though it
reproduces the off-axis controls. Across the five computed resolutions, the
lowest C2--R root scales approximately as $n^{-4}$, whereas the terminal root
scales approximately as $n^4$. The normalised quantities are shown in
Fig.~\ref{fig:c2_edge_scalings}. Over $120\leqslant n\leqslant160$ the
$d_{\rm M}=200$ values are
\begin{align}
n^4\alpha_{\min}&=4.7877734,\ldots,4.7877420,\\
\frac{\alpha_{\max}}{n^4}&=0.079632994,\ldots,0.079634050.
\end{align}
Fits in powers of $1/n^2$ give the provisional finite-pencil intercepts
$$
\widehat c_{\min}=4.78770\ldots,\qquad
\widehat c_{\max}=0.0796354\ldots,
$$
where the hats denote fitted $1/n^2\to0$ intercepts. These are empirical
finite-resolution fits, not established asymptotic laws or continuum
frequencies. The observed $n^{-4}$ scaling is not restricted to the first root. At fixed
low ordinal rank $j$, the products $n^4\alpha_j(n)$ are nearly constant over
$n=120,130,\ldots,160$; for example,
\begin{align}
n^4\alpha_1&=4.7877734,\ldots,4.7877420,\\
n^4\alpha_2&=153.08735,\ldots,153.09092,\\
n^4\alpha_3&=791.88383,\ldots,792.00800.
\end{align}
These data are consistent with fixed-rank sequences collapsing to $\Omega=0$
rather than approaching nonzero NIA frequencies. Their clean finite-window
scaling should not be confused with convergence to a nonzero continuum mode.
At $n=160$, for example, the lowest
root is $\alpha_1=7.3055\times10^{-9}$, while its final-$20\%$ coefficient-tail measure is approximately $0.114$ and both reconstructed endpoint relative residuals are close to unity. The $d_{\rm M}=100$ and $d_{\rm M}=200$ C2--R spectra agree at fixed $n$. At $n=160$ the largest rank-ordered relative difference between the complete NIA spectra is approximately $1.9\times10^{-21}$, while at $d_{\rm M}=200$ the worst raw/equilibrated relative discrepancy is approximately $6.0\times10^{-121}$. The spectral organisation is therefore arithmetic-clean. Within $|\Omega|\leqslant100$, the $d_{\rm M}=200$ scaled backward errors are between approximately $2\times10^{-241}$ and $7\times10^{-239}$, and the largest product
$\kappa_{\rm rel,2}\eta_2^{\rm R}$ is approximately $6.5\times10^{-158}$. No non-special, order-one NIA frequency satisfies the configured $10^{-4}$ three-resolution matching test in the high-resolution C2 windows. The small set of reported C2 NIA matches consists of fixed-rank near-zero sequences, neighbouring-rank associations generated by the absolute matching tolerance, and the algebraically special point. It must therefore not be interpreted as a list of converged nonzero NIA modes. Table~\ref{tab:c1_c2_representative} illustrates the compactification test at three representative C1 ladder locations. The C1--R sequences remain attached to their target values, whereas the nearest C2--R roots move non-monotonically as $n$ changes. Proximity at one resolution is not convergence.
\FloatBarrier
\begin{table*}[t]
\caption{Nearest NIA roots to three representative C1--R target frequencies.
The C1 entries remain on coherent trajectories, whereas the nearest C2 entries
do not. All entries are $\alpha=-\operatorname{Im}\Omega$.}
\label{tab:c1_c2_representative}
\centering
\small
\begingroup
\setlength{\tabcolsep}{5pt}
\renewcommand{\arraystretch}{1.08}
\begin{tabular*}{0.90\textwidth}{@{\extracolsep{\fill}}ccccccc@{}}
\toprule
& \multicolumn{3}{c}{C1--R} & \multicolumn{3}{c}{C2--R} \\
\cmidrule(lr){2-4}\cmidrule(l){5-7}
$\alpha_{\rm target}$ & $n=140$ & $n=150$ & $n=160$
& $n=140$ & $n=150$ & $n=160$ \\
\midrule
15.0786 & 15.078381 & 15.078528 & 15.078616
& 15.084645 & 15.010692 & 15.362458 \\
20.0845 & 20.084521 & 20.084521 & 20.084521
& 20.263274 & 20.055872 & 20.089977 \\
27.8406 & 27.840594 & 27.840594 & 27.840594
& 27.917521 & 28.318469 & 27.888126 \\
\bottomrule
\end{tabular*}
\endgroup
\end{table*}
\FloatBarrier

\subsection{Consequences for NIA candidate classification}
\label{subsec:finite_pencil_interpretation}

The results instantiate the evidence hierarchy in Eq.~\eqref{eq:finite_pencil_evidence_hierarchy}. Small backward errors, raw/equilibrated agreement, coefficient decay, and endpoint consistency establish a high-quality eigenpair of a particular finite pencil. Resolution and representation tests then ask whether that eigenpair belongs to a compactification-independent continuum structure while pole character is a separate analytic statement about a specified continuation of the horizon--infinity Jost determinant. The first-compactification high-damping lattice is a reproducible and highly accurate feature of that finite-dimensional formulation. It survives the roots-to-Lobatto change with frequency differences below the remaining resolution drift and with essentially identical Chebyshev coefficient vectors. It is nevertheless not reproduced as a convergent finite-frequency sequence by the second compactification, whose NIA spectrum instead shows approximate fixed-rank $n^{-4}$ collapse and terminal $n^4$ growth over the computed resolutions. The low-frequency edge is non-universal already within C1. These observations identify substantial parts of the NIA population as representation-dependent finite-dimensional structures and motivate the direct lateral Jost test reported in Sec.~\ref{sec:jost_results}.

\section{Cut-strength phase and the high-overtone QNM comb}
\label{sec:cut_phase}

The quarter spacing is physically informative but is not a new Schwarzschild scale. Casals and Ottewill parameterise the NIA by $\nu=i\omega>0$ and
$\bar\nu=2M\nu$~\cite{CasalsOttewill2012PRD, CasalsOttewill2013PRD}. Since
$\Omega=M\omega=-i\alpha$ here, $\bar\nu=2\alpha$. For $s=\ell=2$, their large-$\bar\nu$ branch-cut strength satisfies
\begin{equation}
\label{eq:cut-strength-asymptotics}
q(\nu)\sim4\left[
\cos(2\pi\bar\nu)
+\frac{\alpha_1}{\sqrt{\bar\nu}}\sin(2\pi\bar\nu)
+O(\bar\nu^{-1})\right],
\end{equation}
where
\begin{equation}
\label{eq:cut-strength-alpha-one}
\alpha_1=
\frac{\Gamma(1/4)^4[1-\ell(\ell+1)]}{48\pi^{3/2}}
=-3.23242351956863\ldots.
\end{equation}
It follows that
\begin{equation}
\Delta\bar\nu\longrightarrow\frac12,\qquad
\Delta\alpha\longrightarrow\frac14=M\kappa,
\end{equation}
where $\kappa=(4M)^{-1}$ is the Schwarzschild surface gravity. The finite pencil therefore does not discover the scale $1/4$. The scientific question is why C1 tracks the leading known continuum phase as such a coherent matrix spectrum while C2 does not. The absolute phase and its slow drift are more discriminating than the
spacing. Let $\alpha_j^{\mathrm{C1}}$ be the 68 C1--R candidates at
$n=160$. Without fitting a continuous phase or scale, we compare them with two asymptotically equivalent zero prescriptions. The literal truncation of Eq.~\eqref{eq:cut-strength-asymptotics} is
\begin{equation}
\label{eq:q-zero-literal}
\cos(4\pi\alpha)
+\frac{\alpha_1}{\sqrt{2\alpha}}\sin(4\pi\alpha)=0,
\end{equation}
whereas the corresponding modulus--phase form is
\begin{equation}
\label{eq:q-zero-phase}
\cos\left(4\pi\alpha-
\frac{\alpha_1}{\sqrt{2\alpha}}\right)=0.
\end{equation}

\FloatBarrier
\begin{table*}[t]
\caption{Parameter-free comparison of the C1--R $n=160$ ladder with two
large-frequency cut-strength prescriptions and with physical QNM damping
projections. Successive integer labels are fixed by ordering; no continuous
offset or scale is fitted. The two cut formulas agree only through the stated
asymptotic order, so their difference measures truncation ambiguity rather
than numerical error.}
\label{tab:cut-phase-summary}
\centering
\begingroup
\newcommand{\CPone}[1]{\parbox[t]{0.15\textwidth}{\raggedright #1}}
\newcommand{\CPtwo}[1]{\parbox[t]{0.11\textwidth}{\centering #1}}
\newcommand{\CPthree}[1]{\parbox[t]{0.17\textwidth}{\centering #1}}
\newcommand{\CPfour}[1]{\parbox[t]{0.11\textwidth}{\centering #1}}
\newcommand{\CPfive}[1]{\parbox[t]{0.29\textwidth}{\raggedright #1}}
\begin{tabular}{lllll}
\toprule
\CPone{Comparator} & \CPtwo{Pairing} & \CPthree{Residual range} &
\CPfour{RMS} & \CPfive{Interpretation} \\
\midrule
\CPone{Literal cut form, Eq.~\eqref{eq:q-zero-literal}}
& \CPtwo{68 successive zeros}
& \CPthree{$[-5.2202,-1.5545]$\\$\times10^{-3}$}
& \CPfour{$2.6925\times10^{-3}$}
& \CPfive{Leading phase and drift agree, with visible omitted-order displacement.} \\
\addlinespace
\CPone{Phase cut form, Eq.~\eqref{eq:q-zero-phase}}
& \CPtwo{68 successive zeros}
& \CPthree{$[-0.80845,0.90696]$\\$\times10^{-3}$}
& \CPfour{$1.8447\times10^{-4}$}
& \CPfive{Striking phase diagnostic; not an error estimate for the full finite-frequency $q$.} \\
\addlinespace
\CPone{Physical QNM damping}
& \CPtwo{$n_{\rm QNM}=62,\ldots,129$}
& \CPthree{$[-2.06923,-1.30247]$\\$\times10^{-2}$}
& \CPfour{$1.59360\times10^{-2}$}
& \CPfive{Shared high-overtone scale, but a distinct finite-frequency phase and nonzero real part.} \\
\bottomrule
\end{tabular}
\endgroup
\end{table*}

The small RMS from Eq.~\eqref{eq:q-zero-phase} is compelling, but it cannot
be read as agreement with the finite-frequency $q(\alpha)$. Equations
\eqref{eq:q-zero-literal} and 
\eqref{eq:q-zero-phase} differ at the omitted $O(\bar\nu^{-1})$ order, and
their root displacement is already larger than $1.84\times10^{-4}$. The
present result is therefore evidence for a cut-controlled phase mechanism,
not a finite-frequency phase-locking theorem. One direct finite-frequency anchor sharpens this comparison. We evaluated
Eq.~\eqref{eq:jost_num_branch_strength} using symmetric lateral
Leaver--Tricomi-$U$ solutions and a backward Miller construction of the
minimal positive-imaginary-axis solution. Near the first C1 ladder member, the discontinuity changes sign at the numerical zero
\begin{equation}
\label{eq:finite-frequency-q-anchor}
\alpha_q=15.07832396512359.
\end{equation}
The quoted rounding is stable under the tested truncations, lateral offsets,
matching radii, and arithmetic precisions. The positive-imaginary-axis
denominator in Eq.~\eqref{eq:jost_num_branch_strength} remains finite and
nonzero, with magnitude $4.9001\times10^{-7}$, so the vanishing discontinuity
is not a division artefact. Table~\ref{tab:finite-q-anchor} shows that the
finite-frequency correction moves the phase-form prediction toward the C1
point. This is useful evidence for the cut-phase mechanism, but one anchor
does not establish a sequence-wide law, uniqueness of all nearby zeros, or an
interval-certified root theorem. Reversing the labelling of the two lips
changes the sign of $q$ but not the zero in
Eq.~\eqref{eq:finite-frequency-q-anchor}.

\begin{table}[t]
\caption{First-ladder comparison with the directly evaluated
finite-frequency branch-strength zero. All frequencies use
$\Omega=-i\alpha$. The last column is measured from $\alpha_q$; no offset is
fitted.}
\label{tab:finite-q-anchor}
\centering
\small
\begingroup
\setlength{\tabcolsep}{7pt}
\renewcommand{\arraystretch}{1.08}
\begin{tabular}{@{}lcc@{}}
\toprule
Quantity & $\alpha$ & $\alpha-\alpha_q$ \\
\midrule
Phase-form asymptotic zero & $15.078158667423952$ & $-1.65298\times10^{-4}$ \\
Finite-frequency numerical $q$-zero & $15.07832396512359$ & $0$ \\
Nearest C1 finite-pencil root & $15.078615820120019$ & $+2.91855\times10^{-4}$ \\
Literal sine--cosine zero & $15.082668877207599$ & $+4.34491\times10^{-3}$ \\
\bottomrule
\end{tabular}
\endgroup
\end{table}
\begin{figure*}[t]
\centering
\includegraphics[width=\textwidth]{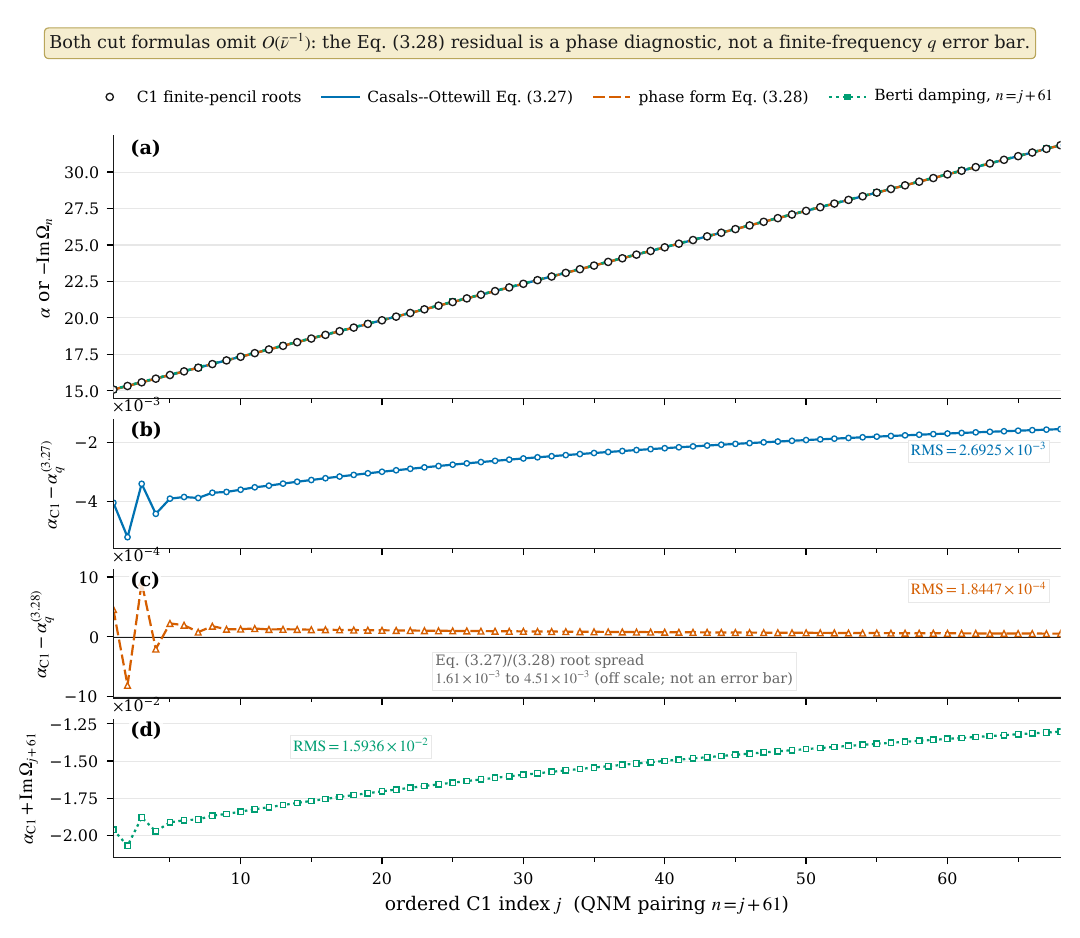}
\caption{Parameter-free phase comparison of the 68 C1--R $n=160$
finite-pencil candidates with the large-$\bar\nu$ Schwarzschild cut strength
and the physical high-overtone damping sequence. (a) Ordered frequencies,
using $\bar\nu=2\alpha$ and converting Berti's $2M\omega$ data to
$\Omega=M\omega$. The QNM pairing is $n_{\rm QNM}=j+61$. (b) Relative to the zeros of
the literal sine--cosine truncation in Eq.~\eqref{eq:q-zero-literal}, the C1
displacements have RMS $2.6925\times10^{-3}$. (c) Relative to the phase form
in Eq.~\eqref{eq:q-zero-phase}, they have RMS $1.8447\times10^{-4}$. The two
same-order prescriptions themselves differ by
$1.6080\times10^{-3}$--$4.5102\times10^{-3}$. Both omit
$O(\bar\nu^{-1})$ terms, so the smaller residual is a phase diagnostic, not
an accuracy estimate for zeros of the full finite-frequency $q$. (d) The
C1 minus projected-QNM-damping residual has RMS $1.5936\times10^{-2}$.
Damping projection does not place the off-axis QNM poles on the cut.}
\label{fig:cut_qnm_phase_comparison}
\end{figure*}
The same interval also contains the established high-overtone QNM comb. To
distinguish the overtone index from the collocation resolution, we denote it
by $n_{\rm QNM}$.
Berti's continued-fraction table is given in $2M\omega$ units; division by
two converts it to $\Omega=M\omega$~\cite{BertiRingdownData, Nollert1993PRD, Motl2003ATMP, CardosoLemosYoshida2004PRD}. The sampled-strip damping window contains
exactly 71 axial overtones, $n_{\rm QNM}=60,\ldots,130$, with endpoints
\begin{align}
\Omega_{60}&\simeq
0.07361873510-14.59782806675i,\\
\Omega_{130}&\simeq
0.06455645238-32.10591316536i.
\end{align}
Their real parts lie in the narrow range $0.0645$--$0.0737$, just outside
the production boundary $\Re\Omega=0.05$. Thus the zero winding of the
narrower strip is compatible with the known spectrum but is not, by itself,
a same-damping positive control. The 68 C1 roots can be ordered beside the damping projections $n_{\rm QNM}=62,\ldots,129$, with mean spacings $0.2502120$ and $0.2501136$, respectively. The equal counts are largely consequences of the window lengths while the independently predicted cut phase, not the count, is the
nontrivial comparison. Finally, notice that the phase comparison does not itself classify poles. The direct test in Sec.~\ref{sec:jost_results} strongly disfavours the C1 ladder as an individual physical-pole family in the tested continuations, but its organisation is not arbitrary. Its spacing and phase suggest that the finite trial space samples information controlled by the continuum cut and the neighbouring high-overtone comb. Establishing that positive interpretation requires exact cut strength or weighted-resolvent convergence, not another unweighted eigenvalue census.

\section{Lateral Jost determinants and physical-pole assessment}
\label{sec:jost_results}

The direct Jost/Wronskian calculation was carried out in the axial gravitational
sector $s=\ell=2$ after the finite-pencil calibration had isolated the
regular C1 ladder. The production calculation used the determinant and
physical-lateral-continuation conventions of Sec.~\ref{sec:jost_numerics}.
Tricomi $U$ was evaluated with Arb/FLINT ball arithmetic, while the recurrence,
Wronskian, and extrapolation algebra used arbitrary-precision midpoint values.
The complementary tests are summarised in Table~\ref{tab:jost_results_summary}.

Before classifying the ladder, we tested the same code on established
off-axis poles at the lower edge, centre, and upper edge of its damping
interval. Starting from Berti's continued-fraction values~\cite{BertiRingdownData}, Newton iteration gives the results in
Table~\ref{tab:high-damping-jost-controls}. Each refined frequency was
rechecked at $x=1.30$, $1.35$, and $1.40$ and at three truncation levels.
This is a direct same-regime validation of the Jost recurrence; it is not a
replacement for a resolved contour count of all intervening poles.

\FloatBarrier
\begin{table}[htbp]
\caption{Same-damping positive controls for the lateral Jost calculation.
Frequencies use $\Omega=M\omega$. The last two columns are evaluated at the
refined root with the production matching radius and truncation. All three
roots remain stable under the radius and truncation repeats described in the
text.}
\label{tab:high-damping-jost-controls}
\centering
\small
\begingroup
\setlength{\tabcolsep}{5pt}
\renewcommand{\arraystretch}{1.08}
\begin{tabular*}{0.92\textwidth}{@{\extracolsep{\fill}}ccccc@{}}
\toprule
$n_{\rm QNM}$ & $\operatorname{Re}\Omega$ & $-\operatorname{Im}\Omega$
& $|\widehat{\mathcal D}|$ & $|\mathfrak m|$ \\
\midrule
60 & 0.0736187351006822 & 14.5978280667549 & $3.47\times10^{-51}$ & $6.36\times10^{-59}$ \\
100 & 0.0673297507837061 & 24.6034778126417 & $5.24\times10^{-48}$ & $5.86\times10^{-61}$ \\
130 & 0.0645564523874952 & 32.1059131653701 & $2.95\times10^{-56}$ & $4.10\times10^{-73}$ \\
\bottomrule
\end{tabular*}
\endgroup
\end{table}

\begin{figure}[htbp]
\centering
\includegraphics[width=\textwidth]{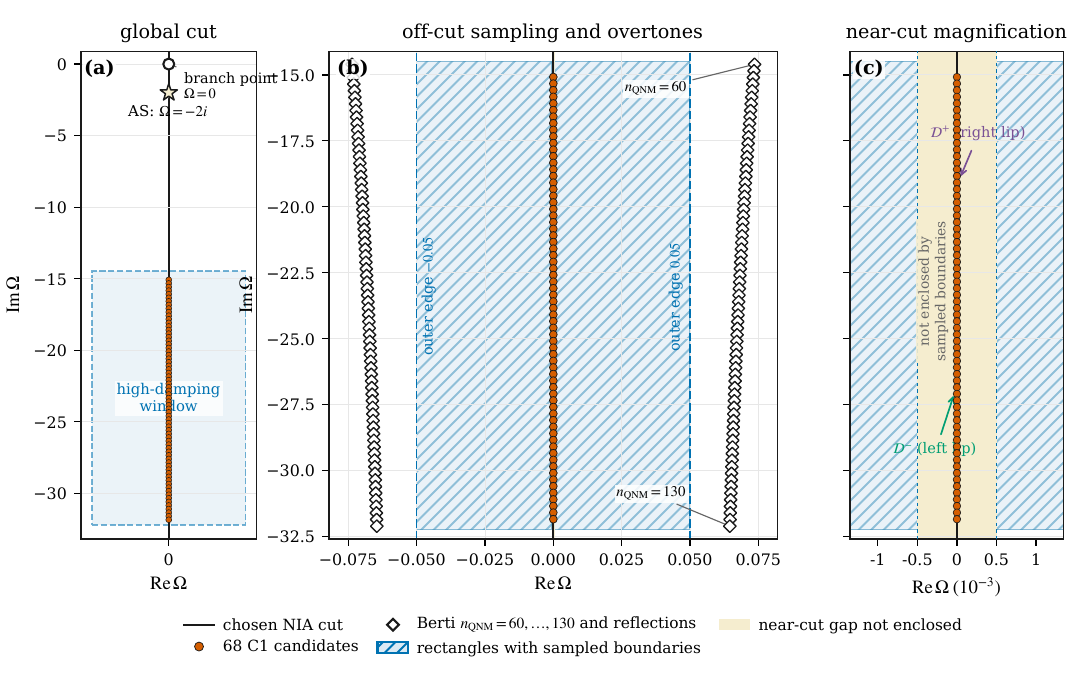}
\caption{Complex-frequency atlas of the lateral Jost calculations, in the
convention $\Omega=M\omega$. (a) The chosen NIA cut, its branch point at
$\Omega=0$, the axial algebraically-special point $\Omega=-2i$, and the
high-damping window. (b) Magnification of
$14.475\leqslant-\operatorname{Im}\Omega\leqslant32.225$. Orange circles are
the 68 C1 finite-pencil candidates. Hatched regions are rectangles whose
off-cut boundaries were sampled,
with $5\times10^{-4}\leqslant|\operatorname{Re}\Omega|\leqslant0.05$.
Diamonds are Berti's $s=\ell=2$ QNMs with
$60\leqslant n_{\rm QNM}\leqslant130$ and
their reflected partners; their real parts lie outside the tested strips.
(c) The near-cut gap $0<|\operatorname{Re}\Omega|<5\times10^{-4}$ not
enclosed by either contour. The indicated transverse offsets of
$\mathcal D^{+}$ and $\mathcal D^{-}$ are diagrammatic; the contours do not
cross the cut.}
\label{fig:complex_frequency_atlas}
\end{figure}

\begin{table*}[t]
\caption{Jost/Wronskian physical-pole assessment for the regular C1
high-damping ladder in the axial $(s,\ell)=(2,2)$ sector. Determinant
magnitudes and numerical diagnostics refer to the right lip. The left lip
satisfies the expected conjugacy relation at the exported precision.}
\label{tab:jost_results_summary}
\centering
\footnotesize
\begingroup
\setlength{\tabcolsep}{4pt}
\renewcommand{\arraystretch}{1.08}
\newcommand{\JRStage}[1]{\parbox[t]{0.160\textwidth}{\raggedright #1}}
\newcommand{\JRCalculation}[1]{\parbox[t]{0.350\textwidth}{\raggedright #1}}
\newcommand{\JROutcome}[1]{\parbox[t]{0.410\textwidth}{\raggedright #1}}
\begin{tabular}{@{}lll@{}}
\toprule
\JRStage{Check} & \JRCalculation{Calculation} & \JROutcome{Numerical outcome} \\
\midrule
\JRStage{Candidate points}
& \JRCalculation{All 68 finite-pencil frequencies,
$15.0786158201\leqslant\alpha\leqslant31.8428207066$, with adaptive 140- and
200-digit policies}
& \JROutcome{$181.8351\leqslant|\mathcal D^{+}|\leqslant276.7755$; every
evaluated candidate remains separated from zero relative to the reported
stability diagnostics in the adopted normalisation.} \\
\JRStage{Uniform NIA mesh}
& \JRCalculation{354-point lateral scan over
$14.525\leqslant\alpha\leqslant32.175$}
& \JROutcome{$\min|\widehat{\mathcal D}^{+}|=149.5033$ and
$\min|\mathcal D^{+}|=151.9008$; the shallow near-quarter modulation resolves
no candidate on-cut zero at the $0.05$ mesh scale.} \\
\JRStage{Contour calibration}
& \JRCalculation{Known-QNM and empty controls, eight representative
physical-half-plane boxes, and two physical-half-plane boxes adjacent to the
algebraically special point}
& \JROutcome{The known fundamental QNM gives winding $+1$, the empty box gives
$0$, and all representative and algebraically-special-adjacent boxes give $0$
at both contour resolutions.} \\
\JRStage{Full-strip boundaries}
& \JRCalculation{Complete strips in the right- and left physical-half-plane
continuations with $5\times10^{-4}\leqslant|\operatorname{Re}\Omega|
\leqslant5\times10^{-2}$ and
$14.475\leqslant-\operatorname{Im}\Omega\leqslant32.225$}
& \JROutcome{Sampled boundary sums give numerical winding $0$ at 192 and 384
points per edge; at the higher resolution the minimum sampled magnitude is
$41.6363$ and the largest sampled phase increment is $2.0377<\pi$.} \\
\bottomrule
\end{tabular}
\endgroup
\end{table*}
\FloatBarrier

Both physical lateral continuations were evaluated at every one of the 68
regular C1 finite-pencil frequencies. Over the complete set, the largest
final truncation change in the physical determinant was
$4.80\times10^{-17}$, the largest matching-radius spread was
$6.11\times10^{-18}$, the largest lateral-extrapolation diagnostic was
$2.21\times10^{-10}$, and the largest series-tail-to-Wronskian-cancellation
ratio was $6.14\times10^{-19}$. These changes are negligible relative to the
determinant scale. The pointwise evaluations therefore strongly resolve
$\mathcal D^{\pm}(-i\alpha)$ as nonzero at all 68 computed frequencies.

\begin{samepage}
A zero could nevertheless lie between the finite-pencil sample points. We
therefore sampled the entire ladder and a margin beyond it on a uniform mesh
with spacing $0.05$. Figure~\ref{fig:j3_determinants} shows the resulting
right-lip magnitudes. On this mesh the regularised determinant never falls
below $149.5033$, while the physical determinant never falls below
$151.9008$. The 71 sampled local minima form shallow, approximately
quarter-spaced oscillations whose lower envelope rises across the interval, and no sampled minimum called for local refinement. This is mesh-wide evidence, not a continuous zero-exclusion theorem.
\end{samepage}

\FloatBarrier
\begin{figure}[htbp]
\centering
\includegraphics[width=0.92\textwidth]{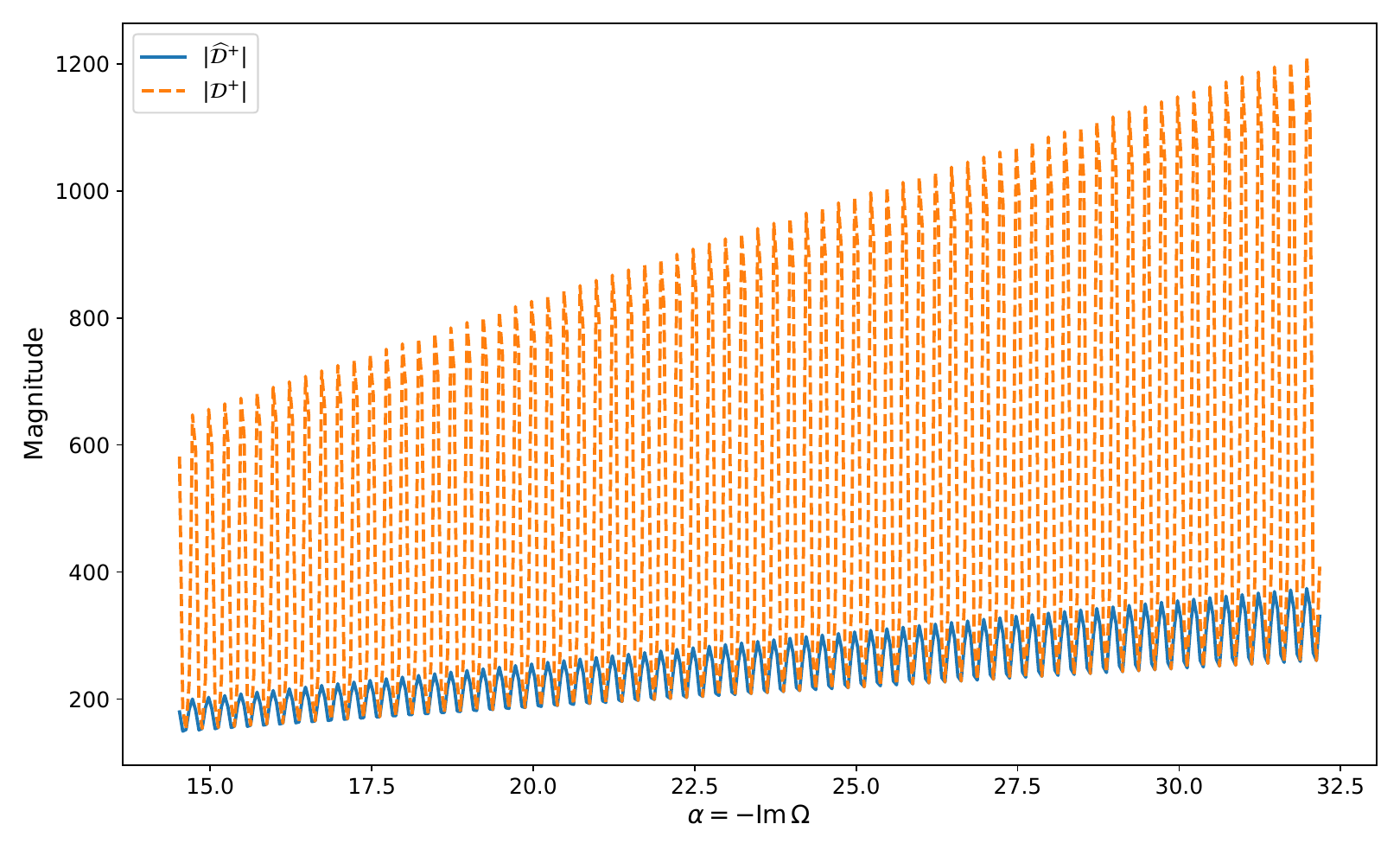}
\caption{Uniform 354-node scan of the right-lip physical and regularised Jost
determinants at spacing $0.05$. The sampled oscillatory near-quarter structure
remains separated from zero in the adopted normalisation over
$14.525\leqslant\alpha\leqslant32.175$; the plot supplies mesh-scale evidence,
not continuous zero exclusion.}
\label{fig:j3_determinants}
\end{figure}
\FloatBarrier

\subsection{Argument-principle tests and physical-pole assessment}

The argument-principle calculation was first calibrated with a
rectangle that enclosed the known fundamental Schwarzschild QNM and returned
numerical winding $+1$, while a nearby empty control rectangle returned $0$. The same
calculation gave zero numerical winding in the physical continuations in both half-planes around four
representative ladder locations. At the algebraically special point
$\Omega=-2i$, the physical determinant is finite and nonzero,
$|\mathcal D^{+}|\simeq0.762949$. At this point the sine factor, rather than
a physical zero, forces $\widehat{\mathcal D}=0$. Two adjacent boxes evaluated
with the physical determinant also had zero numerical winding; no candidate
zero was resolved in the tested region
$0.002\leqslant|\operatorname{Re}\Omega|\leqslant0.05$ and
$1.95\leqslant\alpha\leqslant2.05$.

Full-strip calculations then sampled the boundaries covering the complete
regular-ladder interval in each physical-half-plane continuation. At 192 and
384 points per edge, both
rounded numerical winding sums were zero, and the result was unchanged when
the boundary sampling was doubled. The sampled contour values remained
separated from zero relative to the reported diagnostics in the adopted
normalisation, with minimum magnitude $41.6363$. The maximum principal phase
step decreased from $2.5863$ to $2.0377$ radians under sampling refinement and
remained below $\pi$ at the higher resolution. The largest
series-tail-to-cancellation ratio was $4.24\times10^{-23}$, and the
regularisation identity held to approximately $1.6\times10^{-150}$ relative.
This refinement changed only the number of boundary samples; the 200-digit
precision, matching radius, series length, and rectangle geometry were fixed.
The regularising sine factor has no zeros in these strips and therefore has
zero analytic winding on their boundaries, so the identity gives the same
sampled numerical winding for $\widehat{\mathcal D}$ and $\mathcal D$. The
known-QNM and empty-contour controls calibrate the contour orientation and
phase unwrapping; they do not by themselves certify holomorphy or pole-freeness
throughout the tested domains.

Taken together, the candidate-point calculation strongly resolves the
physical lateral determinants as nonzero at all 68 evaluated finite-pencil
frequencies, and the uniform NIA mesh resolves no additional on-cut candidate
at spacing $0.05$. The contour controls and full-strip boundary samples add
stable numerical zero-winding evidence in the tested physical-half-plane
boxes and strips. By Proposition~\ref{prop:certified-argument-principle}, that
winding would be an exact zero count after holomorphy and pole-freeness on the
closed domains were independently established; those hypotheses are not
certified here. The calculations do not cross the NIA cut and therefore do
not address resonances on unphysical sheets. Together with compactification
dependence, the results identify the C1 ladder as a structured family of
finite-pencil nodes consistent with discretisation of the non-pole NIA cut,
not as random numerical noise and not as an established sequence of cut
quadrature nodes. A positive cut interpretation requires the weighted limit
formulated in Sec.~\ref{sec:keldysh_cut_limit}.

\FloatBarrier
This assessment is numerical rather than a formal interval-certified theorem.
The Arb balls used for Tricomi $U$ are reduced to high-precision midpoints in
the subsequent recurrence and Wronskian algebra, and the physical-half-plane
strip contours leave the narrow geometric gap
$0<|\operatorname{Re}\Omega|<5\times10^{-4}$ between the cut and their inner
boundaries. The contours also intentionally do not cross the NIA, so they
cannot exclude poles on unphysical sheets. Within the computed physical
domains, however, sampled determinant values stay separated from zero in the
adopted normalisation on the cut and on the inner strip boundaries; no
anomalous sampled minimum or nonzero numerical winding is observed.

\section{Keldysh weights and a positive cut-limit criterion}
\label{sec:keldysh_cut_limit}

Declaring that an individual finite-pencil root is not a QNM is only a
negative classification. A stronger question is whether the roots become useful collectively. The natural objects are not the unweighted point clouds but the residues of matrix elements of the finite resolvent. This is the finite-dimensional content of the Keldysh expansion for holomorphic matrix functions~\cite{Beyn2012LAA}. It is also the mechanism used to reconstruct asymptotically flat tails from branch-cut eigenvalues in Ref.~\cite{BessonJaramillo2025GRG}.
The recent continuum constructions also sharpen the target of such a positive programme. The decomposition of Su et al. separates direct, pole, and tail pieces by their Green-function analytic structure and checks their reconstruction in the time domain. Rosato et al. identify low-frequency tail corrections and a surface-gravity-governed redshift sector, and Arnaudo et al. show that a branch-cut contribution can admit a convergent discrete mode sum when obtained from a controlled de Sitter limiting family. Thus, a useful discrete family need not consist of Schwarzschild poles, but it must converge (with specified coefficients, source--observable pairing, continuation, and summation prescription) to a definite continuum object. The Keldysh criterion below is one finite-pencil realisation of that broader requirement. None of these recent results establishes convergence of the present C1 or C2 measures, and none supports interpreting the individual C1 nodes as  QNMs~\cite{Su2026PRD,RosatoDeAmicisPani2026PRD,ArnaudoCarballoWithers2026PRD}.

\begin{proposition}[Residue of a simple quadratic-pencil eigenvalue]
\label{prop:keldysh-simple-residue}
Let $Q_n(z)=M_0+izM_1+z^2M_2$ be regular, and let $\lambda_{j,n}$ be an
algebraically simple finite eigenvalue. Choose nonzero right and left eigenvectors $a_{j,n}$ and $b_{j,n}$ such that
\begin{equation}
Q_n(\lambda_{j,n})a_{j,n}=0,\qquad
b_{j,n}^{\dagger}Q_n(\lambda_{j,n})=0.
\end{equation}
Then, $b_{j,n}^{\dagger}Q_n'(\lambda_{j,n})a_{j,n}\neq 0$, and in a
neighbourhood of $\lambda_{j,n}$, we find
\begin{equation}\label{eq:keldysh-local-resolvent}
Q_n(z)^{-1}=\frac{a_{j,n}b_{j,n}^{\dagger}}
{(z-\lambda_{j,n})
b_{j,n}^{\dagger}Q_n'(\lambda_{j,n})a_{j,n}}
+H_{j,n}(z),
\end{equation}
where $H_{j,n}$ is holomorphic. Consequently, for a discrete source
$d_n$ and observable $c_n$, the scalar resolvent element has residue
\begin{equation}\label{eq:keldysh-scalar-weight}
\rho_{j,n}=\frac{(c_n^{\dagger}a_{j,n})(b_{j,n}^{\dagger}d_n)}
{b_{j,n}^{\dagger}Q_n'(\lambda_{j,n})a_{j,n}}.
\end{equation}
This weight is invariant under independent nonzero rescalings of
$a_{j,n}$ and $b_{j,n}$.
\end{proposition}

\begin{proof}
Algebraic simplicity and regularity imply that the inverse has a simple
Laurent pole. Let us write
\begin{equation}
Q_n(z)^{-1}=\frac{R_{-1}}{z-\lambda_{j,n}}+R_0+O(z-\lambda_{j,n}).
\end{equation}
The coefficients of $(z-\lambda_{j,n})^{-1}$ in the identities
$Q_nQ_n^{-1}=I$ and $Q_n^{-1}Q_n=I$ give $Q_n(\lambda_{j,n})R_{-1}=0$ and $R_{-1}Q_n(\lambda_{j,n})=0$. The right and left nullspaces are
one-dimensional, so $R_{-1}=\gamma a_{j,n}b_{j,n}^{\dagger}$ for some
$\gamma\neq0$. The constant coefficient in $Q_nQ_n^{-1}=I$ is
\begin{equation}
Q_n(\lambda_{j,n})R_0+Q_n'(\lambda_{j,n})R_{-1}=I.
\end{equation}
Left multiplication by $b_{j,n}^{\dagger}$ and substitution of the rank-one form yield
\begin{equation}
\gamma\,b_{j,n}^{\dagger}Q_n'(\lambda_{j,n})a_{j,n}\,
b_{j,n}^{\dagger}=b_{j,n}^{\dagger}.
\end{equation}
Thus, the derivative pairing is nonzero and $\gamma$ is its reciprocal,
which proves Eq.~\eqref{eq:keldysh-local-resolvent}. Applying the resulting rank-one residue to $c_n^{\dagger}Q_n(z)^{-1}d_n$ gives
Eq.~\eqref{eq:keldysh-scalar-weight}. Its numerator and denominator acquire the same factors under independent rescalings of $a_{j,n}$ and $b_{j,n}$, so the weight is invariant.
\end{proof}

For a region containing simple finite eigenvalues, the complex atomic
measure
\begin{equation}
\label{eq:finite-keldysh-measure}
\mu_n^{c,d}=\sum_j\rho_{j,n}\,\delta_{\lambda_{j,n}}
\end{equation}
encodes the singular part of the chosen scalar finite resolvent. Its Cauchy transform is
\begin{equation}
\mathcal C\mu_n^{c,d}(z)=\int\frac{d\mu_n^{c,d}(\zeta)}{z-\zeta}
=\sum_j\frac{\rho_{j,n}}{z-\lambda_{j,n}}.
\end{equation}
\begin{figure}[htbp]
\centering
\begin{tikzpicture}[
  node distance=8mm,
  every node/.style={align=center,font=\small},
  box/.style={draw=blue!55!black,rounded corners=2pt,thick,
    fill=blue!3,text width=0.105\textwidth,minimum height=12mm,inner sep=3pt},
  finite/.style={box,draw=orange!75!black,fill=orange!5,
    text width=0.175\textwidth},
  measure/.style={box,text width=0.12\textwidth},
  arr/.style={-{Latex[length=2.8mm,width=2mm]},very thick,blue!60!black}
]
\node[finite] (c1) {C1 eigenpairs\\$\{\lambda_{j,n}^{(1)},a_{j,n}^{(1)},b_{j,n}^{(1)}\}$};
\node[finite,below=of c1] (c2) {C2 eigenpairs\\$\{\lambda_{j,n}^{(2)},a_{j,n}^{(2)},b_{j,n}^{(2)}\}$};
\node[measure,right=34mm of c1,yshift=-10mm] (measures) {Keldysh measures\\$\mu_n^{(1)},\mu_n^{(2)}$};
\node[box,right=21mm of measures] (green) {common off-cut\\limit of $\langle c,G_\Omega d\rangle$};
\node[box,right=21mm of green] (tail) {inverse transform\\Price-tail\\observable};
\draw[arr] (c1) -- node[pos=0.48,above=1mm,sloped,font=\scriptsize,
  fill=white,inner sep=1pt]{Eq.~\eqref{eq:keldysh-scalar-weight}} (measures);
\draw[arr] (c2) -- node[pos=0.48,below=1mm,sloped,font=\scriptsize,
  fill=white,inner sep=1pt]{Eq.~\eqref{eq:keldysh-scalar-weight}} (measures);
\draw[arr] (measures) -- node[above=7mm,font=\footnotesize]{Cauchy transform} (green);
\draw[arr] (green) -- node[above=7mm,font=\footnotesize]{inverse Laplace\\transform} (tail);
\end{tikzpicture}
\caption{A falsifiable positive continuum programme. Map-dependent finite
nodes are first equipped with left--right Keldysh weights. The scientifically
meaningful convergence test is whether both weighted families approach the
same off-cut Green-function observable and, after inversion, the same Price
tail. This diagram states a research criterion. The present work does not
claim that the two limiting arrows have already been proved.}
\label{fig:keldysh_cut_limit}
\end{figure}
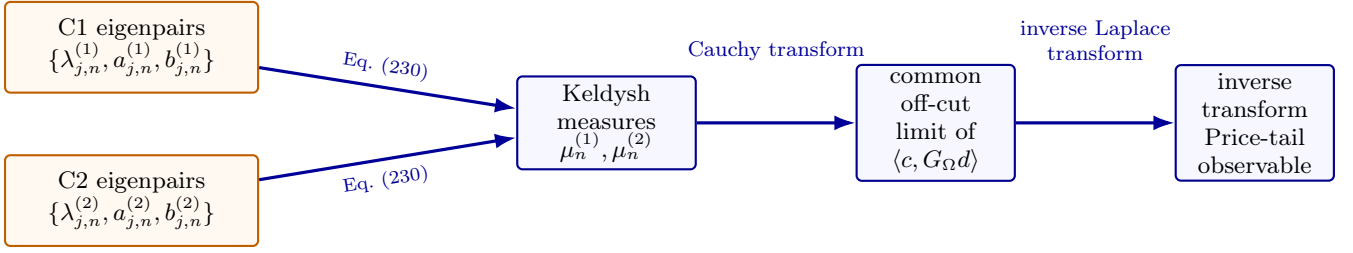
\FloatBarrier

This suggests a precise positive test. Choose discrete sources and
observables that converge in a stated physical norm. Then test whether the C1 and C2 Cauchy transforms converge locally uniformly away from the NIA to the same matrix element of the continued Schwarzschild Green function. A second, stronger test is convergence of the corresponding inverse Laplace transform to the known Price tail. Agreement of unweighted node locations is neither necessary nor sufficient; convergence of these weighted observables is the invariant target. Operationally, the continuum comparator can be refined beyond the undifferentiated full Green function. One may compare separately with the direct and tail pieces isolated by Su et al., with the low-frequency and redshift structures analysed by Rosato et al., or, after introducing and removing a positive cosmological-constant regulator, with the convergent mode sum of Arnaudo et al.~\cite{Su2026PRD,RosatoDeAmicisPani2026PRD,ArnaudoCarballoWithers2026PRD}. These are complementary benchmarks. Agreement with any one of them requires matching normalisations and observables, not merely similar node locations.
The finite pencils therefore define meromorphic matrix-valued functions and complex atomic measures. If the continuum problem is realised as an analytic Fredholm family of index zero that is invertible at one point, analytic Fredholm theory makes its
inverse meromorphic on the slit sheet and holomorphic on pole-free
subdomains. The open problem is convergence of Cauchy transforms there or, equivalently under suitable bounds, weak convergence of the weighted measures against analytic test functions. This formulation separates a physically meaningful approximation theorem from visual convergence of spectral point clouds.

\section{Conclusions and outlook}
\label{sec:conclusions_outlook}

This work has examined whether a highly regular sequence of negative-imaginary-axis frequencies produced by compactified spectral discretisations should be interpreted as a family of physical Schwarzschild quasinormal poles. In the axial gravitational sector with $(s,\ell)=(2,2)$, the evidence does not support that interpretation for the regular high-damping ladder within the frequency domain and physical lateral continuations tested here. The ladder is a genuine and remarkably reproducible feature of one family of finite-dimensional pencils, but it does not survive the representation and physical Jost-pole tests applied to the continuum problem.

The exact analysis identifies the mechanism behind this distinction. For
every fixed finite regularity order $k$, both local scattering sectors satisfy
the compactified one-sided $C^k$ endpoint requirements once the damping
crosses the explicit threshold, while the unwanted infinity sector is flat.
Finite endpoint regularity therefore cannot select the Jost sectors. Under
the nonlinear C1--C2 coordinate change, polynomial pullback doubles degree,
the exact common degree-$d$ core has dimension $\lfloor d/2\rfloor+1$, and
finite polynomial truncation does not commute with pullback. The exact
Chebyshev endpoint coefficients and $O(n^2)$ versus $O(n^4)$ grid-reach laws
then explain why a nonselected sector can nevertheless look spectrally clean
and why the two compactifications resolve different radial scales. These
results expose a precise representation obstruction, and do not assert
spectral equivalence or continuum-resolvent convergence.

The finite-pencil comparison first separated arithmetic accuracy from
representation dependence. The ordinary off-axis QNMs were recovered
consistently with the first compactification on Chebyshev roots and Lobatto
points and with the second compactification on Chebyshev roots. Within the
first compactification, the regular NIA ladder was stable under the
roots-to-Lobatto change: its frequencies, spacing, and Chebyshev coefficient
vectors agreed far more closely than successive resolutions did. The second
compactification, however, reorganised the NIA spectrum and replaced the
first-map edge behaviour by different finite-resolution scalings.
High-precision repeats showed that this change was not caused by an arithmetic
ceiling. Grid stability within one compactification is therefore meaningful
evidence about a finite-pencil structure, but it is not sufficient evidence
for a representation-independent mode. At the finite-dimensional level, the
exact backward-error and pseudospectrum formulas give the least coefficientwise
relative perturbation, in the chosen spectral or Frobenius norm, that makes a
prescribed frequency $z$ singular for the perturbed pencil. Although row
equilibration leaves the exact eigenvalues unchanged, it generally changes
this perturbation model and hence the pseudospectrum.

The branch-cut comparison supplies the corresponding physical mechanism.
Without a fitted scale or offset, all 68 C1 roots track successive zeros of
the leading cut-strength phase form as quantified in
Table~\ref{tab:cut-phase-summary}. A direct calculation resolves the stable
numerical branch-strength zero in Eq.~\eqref{eq:finite-frequency-q-anchor}
near the first ladder member. Neighbouring high-overtone QNMs share the
quarter-spacing scale but remain distinct in phase and real part. This
supports a cut-controlled phase mechanism, not a sequence-wide phase-locking
theorem or a pole classification.

The independent Jost--Wronskian calculation supplied the decisive test of the
physical-pole claim. Both lateral determinants were numerically stable and
separated from zero relative to the reported variations in the adopted
normalisation at the 68 computed frequencies. A 354-node mesh with spacing
$0.05$ resolved no additional on-cut candidate. Known-QNM and empty-contour
controls calibrated orientation, while 192- and 384-point samples gave zero
numerical winding on the stated off-cut rectangular boundaries. The physical
determinant was also sampled as nonzero at the algebraically
special frequency. Because the recurrence/Wronskian calculation is midpoint
based, the contour images were not interval enclosed, and pole-freeness on the
closed domains was not proved, these are resolution-refined numerical winding results
rather than zero-free theorems. Together they strongly disfavour identifying
the regular ladder with physical poles in the tested continuations; they do
not address the near-cut gap or sheets reached through the cut, the latter
being a distinct issue in Schwarzschild scattering near the algebraically
special frequency~\cite{MaassenVanDenBrink2000PRD, Leung2003CQG}.

The evidence is consistent with a structured finite-dimensional
representation of the cut, but the term \emph{cut approximant} remains
provisional until Keldysh-weighted resolvent or tail convergence is shown. The
word \emph{structured} is important. The ladder is too regular, too
reproducible under the first-map grid change, and too coherent in its observed
finite-resolution edge scalings to be dismissed as random numerical noise.
What fails is not the numerical solution of the matrix pencils, but the
inference from a well-resolved finite-pencil eigenpair to an isolated physical
pole of the analytically continued Green function. The calculation
consequently provides a concrete example of the hierarchy
$$
\text{finite-pencil accuracy}
\;\not\!\Longrightarrow\;
\text{representation-independent mode}
\;\not\!\Longrightarrow\;
\text{quasinormal pole}.
$$
Small backward errors, decaying spectral coefficients, stable eigenvectors, and endpoint consistency remain valuable diagnostics, but their physical interpretation must be tested against changes of representation and, on a branch cut, against the appropriate lateral Jost determinant.

The scope of this conclusion should not be enlarged beyond what was computed. It concerns axial Schwarzschild perturbations with $(s,\ell)=(2,2)$, the regular ladder in the damping interval reported above, the two physical lateral continuations on the cut, and the physical-half-plane rectangles whose boundaries were sampled in the argument-principle calculation. It does not exclude isolated NIA poles in other perturbation sectors, at other damping scales, or outside the tested complex domain. It also does not address unphysical sheets reached by continuation through the cut. The narrow region between the cut and the inner physical-half-plane strip boundaries was not enclosed by a single contour, although sampled determinant values on the cut and those inner boundaries remained separated from zero relative to the reported stability diagnostics in the adopted normalisation and showed no anomalous minimum. Moreover, Tricomi $U$ was evaluated with Arb complex balls but the subsequent recurrence, Wronskian, and extrapolation algebra used their high-precision midpoints. The present result is therefore a strongly overdetermined high-precision physical-pole assessment, not a fully interval-certified theorem.

Several directions follow naturally. The first is analytical because the approximately quarter-spaced first-map lattice and the compactification-dependent edge laws call for an asymptotic theory of how spectral discretisations represent the Schwarzschild branch cut. Such a theory should explain which properties are inherited from the continuum discontinuity and which are set by the radial map, collocation geometry, and finite-dimensional truncation. Concretely, one should equip both the C1 and C2 nodes with the weights in Eq.~\eqref{eq:finite-keldysh-measure}, compare their Cauchy transforms with the same off-cut Green-function matrix element, and quantify the error uniformly on compact pole-free subsets. This comparison would decide whether two representation-dependent point clouds encode one representation-independent continuum response.

A second direction is to apply the same evidence hierarchy to other sectors and geometries. Scalar, electromagnetic, and polar gravitational Schwarzschild perturbations provide controlled extensions, but their exceptional frequencies and Jost normalisations should be treated independently rather than inferred from the present axial calculation. More importantly, the overdamped and NIA ladders reported for quantum-corrected and higher-curvature black-hole models should be revisited with a formulation-independent pole test. The present result does not imply that every ladder in a deformed geometry is a cut approximation. It shows instead that regularity, spacing, and spectral convergence alone cannot decide the question.

A complementary formulation check is suggested by the recent complex-scaling treatment of Schwarzschild and Reissner--Nordstr\"om QNMs~\cite{OgawaHiroseMorikawa2026arXiv}. Complex scaling rotates the continuum and converts outgoing radiation conditions into a non-Hermitian eigenproblem. Resonance candidates are then tested for stability under the scaling angle, basis size, and basis parameters. Applying an exterior-complex-scaling implementation to the present high-damping window would therefore probe the same candidates in a functional setting independent of the compactified finite-endpoint regularity used here. Stable isolated eigenvalues would support pole character, whereas angle- or basis-dependent clouds would support a continuum interpretation. Because the current complex-scaling study is strongest for low-lying modes and reports increasing difficulty for broad, highly damped resonances, this is a promising cross-check and development path rather than an existing exclusion test for the 68-point ladder.

A third direction is numerical certification. An end-to-end complex-ball implementation of the recurrences, Wronskian algebra, lateral extrapolation, and contour winding could turn the present high-precision evidence into a computer-assisted zero-exclusion result on specified physical-continuation domains. Adaptive contour subdivision and stable backward-recursion formulations would also extend the calculation to larger damping and more complicated backgrounds. A separate cross-cut construction would be required to investigate unphysical sheets. These developments would strengthen the general methodology, although no further heavy computation is presently indicated for the regular axial Schwarzschild ladder studied here.

The broader conclusion is methodological: convergence of a compactified-pencil
eigenvalue does not establish a quasinormal mode. On a branch cut, defensible
pole classification requires representation tests, the appropriate lateral
Jost determinant, and argument-principle counts on the chosen continuation.

\subsection*{Code availability}

\noindent
The companion repository is being prepared for public release at
\begin{center}
\url{https://github.com/dutykh/schwarzschild-nia-spectra/}
\end{center}
It will contain the Maple worksheets used to assemble the raw and row-equilibrated
Chebyshev matrix pencils, the available machine-readable finite-pencil spectra
and diagnostics, the Julia programs for the independent lateral
Jost--Wronskian calculation, and the source data and scripts for the publication
figures. Repository-relative documentation will record the numerical domains,
normalisations, dependencies, input--output conventions, and principal
reproduction commands. The production Julia eigensolver and postprocessor
that generated the archived finite-pencil tables are absent from the archived
project snapshot. Their absence is a documented provenance limitation, and
the release will not claim end-to-end regeneration of those tables.

\bibliography{DB-DD-MS-QNMsSchwarzschild_Final}

\end{document}